\documentclass{TUBthesis}

\Title{Application of Stochastic Control Algorithms for the
Improvement of the Electron Injection Efficiency of BESSY II} 
\Thesis{Masterarbeit} 
\Author{B. Sc. Alexander Schütt} 
\MNumber{399787} 
\Supervisor{Prof. Dr. Christoph Knochenhauer} 
\Referee{PD Dr. Pierre Schnizer} 
\Date{31.01.2024} 
\Language{english} 

\DeclareMathOperator{\R}{\mathbb{R}} 
\DeclareMathOperator{\N}{\mathbb{N}} 
\DeclareMathOperator{\A}{\mathcal{A}} 
\DeclareMathOperator{\E}{\mathbb{E}} 
\DeclareMathOperator{\wP}{\mathbb{P}} 

\let\phi\varphi 

\usepackage{hyperref}

\usepackage{verbatim}

\usepackage{csquotes} 

\usepackage{listings} 

\usepackage{xcolor}

\usepackage{dsfont}

\DeclareMathOperator*{\argmin}{arg\,min}
\DeclareMathOperator*{\argmax}{arg\,max}

\DeclareMathOperator{\sign}{sign} 

\usepackage[format = plain,font=small,labelfont=bf,labelsep = space]{caption}
\usepackage{float}
\usepackage{array}

\usepackage{algorithm2e} 

\usepackage{booktabs} 

\usepackage[toc]{appendix}

\usepackage{framed} 
\usepackage{bbm}    

\usepackage{cleveref}
\usepackage{graphicx}
\graphicspath{ {./pics/} }

\usepackage[style=numeric,sorting=none]{biblatex}

\usepackage{pdfpages} 

\begin{document}

\pagenumbering{gobble} 


\chapter*{Eidesstattliche Erklärung}
Hiermit erkläre ich, dass ich die vorliegende Arbeit selbstständig und eigenhändig sowie ohne unerlaubte fremde Hilfe und ausschließlich unter Verwendung der aufgeführten Quellen und Hilfsmittel angefertigt habe.
\\[.3cm]
Berlin, den 31.01.2024
\\[2cm]
Alexander Schütt

\cleardoublepage


\chapter*{Zusammenfassung}
In dieser Arbeit wird das Prinzip des Reinforcement Learnings auf ein reales Problem angewandt. Ziel der Arbeit ist die Injektionseffizienz von BESSY II zu optimieren.

Zuerst werden die Grundladen des Reinforcement Learnings zusammengetragen. Wir beweisen unter welchen Annahmen sich mit dem Optimalitätsprinzip von Bellman optimale Strategien finden lassen und stellen zwei auf dem genannten Prinzip basierenden Algorithmen vor. Zusätzlich wird gezeigt, wie die Parameter der Algorithmen automatisch auf das Problem optimiert werden können.

Danach widmen wir uns der Anwendung. Zunächst wird die Synchrotronstrahlungsquelle BESSY II vorgestellt. Wir werden auf die einzelnen Bauteile genauer eingehen, insbesondere auf den Non-Linear Kicker. Jener wird bei der zu optimierenden Injektion verwendet. Um ein Verständnis der Injektion von Elektronen mit dem Non-Linear Kicker aufzubauen, wird gezeigt unter welcher notwendigen Bedingung eine Injektion erfolgreich sein kann.

Daraufhin stellen wir drei verschiedene Modelle vor, wie man die Injektion mathematisch beschreiben kann. Die Modelle unterscheiden sich darin, wie viele Elektronen injiziert werden und wann entschieden werden kann in welcher Runde der Non-Linear Kicker aktiviert wird.

Um gute Modelle mit den vorgestellten Reinforcement Learning Algorithmen zu erhalten, zeigen wir, wie wir eine Simulation der BESSY II Injektion approximiert haben, um diese signifikant zu beschleunigen.

Zum Schluss zeigen wir die besten Injektionsstrategien in allen 3 Modellen und wie man diese für die reale Synchrotronstrahlungsquelle BESSY II benutzen kann.

\cleardoublepage

\tableofcontents
\vspace{-1cm}

\cleardoublepage
\pagenumbering{arabic}

\chapter{Introduction}
In this thesis, we show how stochastic control algorithms are used to increase the injection efficiency of BESSY II. We also give an outlook on how to fully automate the injection process at BESSY II.  

We start with an introduction to the electron injection at BESSY II and reinforcement learning.

\section{Electron Injection at BESSY II}
This thesis explores the application of stochastic control algorithms, specifically reinforcement learning algorithms, to a real world application. The goal is to improve the electron injection efficiency of BESSY II.

BESSY II is a synchrotron light source located in Adlershof, Berlin. First, it accelerates electrons in a synchrotron and then stores them in a storage ring. Then, the stored electrons are accelerated in various directions to generate synchrotron light. To maintain the energy of the synchrotron light at a constant level, new electrons need to be injected regularly. Typically, these injections happen roughly every three minutes.

Currently, electrons are injected by a 4-kicker bump:
\begin{figure}[H]
    \centering
    \includegraphics[width=12cm]{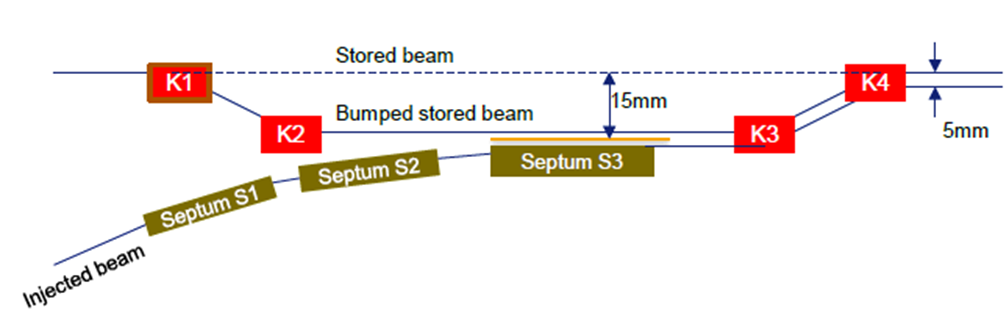}
    \caption{4-Kicker bump, by \cite{4kickerbump}.}
    \label{fig:4kicker_bump}
\end{figure}
In a first step, the stored electrons, that create the synchrotron light, are kicked out of their initial path by 2 magnets (called K1,K2 in figure \ref{fig:4kicker_bump}). The redirected electrons are then merged with newly injected electrons before being kicked back (by magnets K3,K4) to the stored beam position.

This process causes disturbances in the path of the stored electrons, which reduces the quality of the synchrotron light and additionally risks losing 
 a significant fraction of the electrons in the storage ring during this injection procedure.

A new injection procedure, called the Non-Linear Kicker Injection, has been introduced to BESSY II, see \cite{nlk_paper}. This method uses a single magnet with a non-linear magnetic field to inject the electrons, as shown in the figure below.
\begin{figure}[H]
    \centering
    \includegraphics[width=8cm]{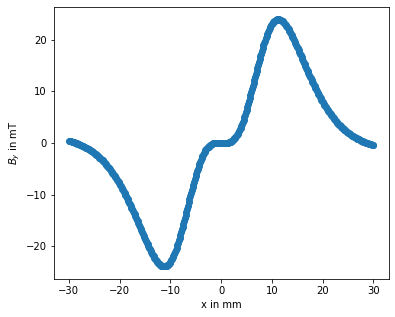}
    \caption{The magnetic field of the non-linear kicker. On the x axis the x-dimensional distance to the centre of the beam line is shown. On the y axis the corresponding amount of magnetic field strength is shown.}
    \label{fig:nonlinearkickerfield}
\end{figure}
This injection method has the advantage of only minimally affecting the stored electrons whose x values are very close to 0, as the magnetic field strength is almost 0 in this region (see figure \ref{fig:nonlinearkickerfield}). As a result, the stored electrons are only slightly disturbed and the quality of the synchrotron light is not deteriorated. 

This non-linear kicker magnet has already been installed at BESSY II in 2011, see \cite{nlk_paper}. It is currently not used because the injection efficiency is lower compared to the 4-kicker bump. Recently some progress has been made to increase the efficiency by optimising the strength of the magnets leading to the non-linear kicker.

\section{Reinforcement Learning}
Reinforcement learning is a subfield of machine learning that deals with decision-making.
Its foundation is based on learning by interaction, which is inspired by the way humans learn.

Through interactions, we humans can analyse how the environment around us responds to our actions. Given such a response, we get the information on whether the environment has changed in the desired way, or if we need to adapt our actions. 

Consider, for example, a game of chess. The outcome of a game is influenced by the actions taken, whether it results in a win, loss or draw. Based on the outcome, we may adjust our play style to improve our chances of winning future games.  

In reinforcement learning, this concept is applied. Reinforcement learning algorithms explore the environment to find a policy that maximise a numerical reward. 

The mathematical foundation of reinforcement learning are Markov Decision Models. They describe the environment by providing the state and action spaces, the dynamics between them, such that the Markov property holds and additionally a reward system that rewards/punishes the agent for taking good/bad actions. Social science speaks about positive/negative reinforcement.

Once we have defined Markov Decision Models, we use the Markov property of the environment to show fundamental mathematical connections that allow us to compute optimal policies. 

As these optimal policies can only be computed efficiently in discrete and finite cases, we describe different reinforcement learning algorithms. They use the mathematical connections, but in an approximate way, using neural networks. Allowing efficient calculations for continuous action and state spaces. Finally, we apply them to the injection problem. 

\section{Scope of the Thesis}
This thesis is structured as follow:

Chapter \ref{chap:rl} introduces the mathematical foundation of Reinforcement Learning. We define Markov Decision Processes and prove under which conditions the Bellman equation holds. After that, we present different algorithms for the search of good strategies. 

In chapter \ref{chap:physics}, the physics models are introduced, which are the basis of the used simulation. These shall give a concise description of the physics involved.

Chapter \ref{chap:problem_modelling} describes how the Markov decision models are formulated. In total we give 4 formulations: The first formulation describes the actual problem in BESSY II. After that we see three formulations that are able to be solved by reinforcement learning algorithms. We show the 1000-step single-electron injection, the 1000-step 1000-electron injection and finally the 1-step 1000-electron injection.

As the preceding chapters laid out the foundations, chapter \ref{chap:calculations} presents how the reinforcement learning algorithms are applied to the injection problem and gives its training results.

\clearpage\null\thispagestyle{empty}

\chapter{Reinforcement Learning}\label{chap:rl}
This chapter is based on chapter 2 of the book \textit{'Markov decision processes with applications to finance'} \cite{mdp_in_finance} by Nicole Bäuerle and Ulrich Rieder. Some adaptations have been made, to fit the setting of our application better. For example, the set of admissible actions is identical in each state, as this is the case in our application and it simplifies the model.

We start this chapter by presenting the foundations of reinforcement learning. After the definition of a Markov Decision Model/Process, policies and the value function, we show the Bellman equation and prove that under some conditions we can find optimal policies using the Bellman equation. Next, we present temporal difference learning and two algorithms to find good policies. After that, we introduce partially observable Markov decision processes, that are needed in the non-linear kicker injection problem. At the end, we describe the process of the hyperparameter search and present different algorithms for the hyperparameter search.   

\section{Markov Decision Processes}
\begin{defn}[Markov Decision Model]
    A Markov Decision Model with planning horizon $N\in\N$ is a tuple $(E,A,p,p_0,r,\beta,g)$, where:
    \begin{itemize}
        \item $E$ is a measurable space with $\sigma$-algebra $\mathcal{E}$, called state space
        \item $A$ is a measurable space with $\sigma$-algebra $\mathcal{A}$, called action space 
        \item $p$ is a stochastic transition kernel from $E\times A$ to $E$, i.e. for all $(x,a)\in E\times A$ and $B\in\mathcal{E}$, the function $B\mapsto p(B|x,a)$ is a probability measure on $\mathcal{E}$ and the mapping $(x,a)\mapsto p(B|x,a)$ is measurable. The value $p(B|x,a)$ is the probability that if at any time $n$ in state $x\in E$ the action $a$ is taken, the next state at time $n+1$ is in $B$. 
        \item $p_0$ is called initial probability on $E$
        \item $r\colon E\times A\to \R$ is a measurable function, called reward function
        \item discount factor $\beta\in(0,1]$
        \item $g\colon E\to\R$ is a measurable function, called terminal reward function
    \end{itemize}
\end{defn}

Our goal now is to define a random process on the Markov Decision Model $(E,A,p,p_0,r,\beta,g)$ with planning horizon $N\in\N$. The random process is going to take values in the state space, that can be influenced by the actions taken. For this, we consider the following measurable space:
\[
    \Omega\coloneqq E^{N+1}\quad\text{with $\sigma$-algebra}\quad \mathcal{F}\coloneqq\bigotimes_{n=0}^N\mathcal{E},
\]
where $(E,\mathcal{E})$ is the state space.
\begin{defn}[Markov Decision Process]
    Let $M = (E,A,p,p_0,r,\beta,g)$, $N\in\N$ and $(\Omega, \mathcal{F})$ be as above. For all $n\in\{0,\dots,N\}$ and $\omega = (\omega_0,\dots,\omega_N)\in\Omega$, we define
    \[
    X_n(\omega)=X_n((\omega_0,\dots,\omega_N)) \coloneqq \omega_n.
    \]
    The random process $X=\{X_n\}_{n=0,\dots,N}$ is called Markov Decision Process.
    
    Since $\mathcal{F}$ is the product $\sigma$-algebra on $\Omega$, the functions $X_n$ are measurable, as they are projections.
\end{defn}

In the following, we assume that $(E,A,p,p_0,r,\beta,g)$ denotes a Markov Decision Model with planning horizon $N\in\N$ and $X=\{X_n\}_{n=0,\dots,N}$ denotes a  Markov Decision Process on $(\Omega,\mathcal{F})$.

To allow  the Markov Decision Process to depend on the actions taken, we first define the set of valid actions: 
\begin{defn}[decision rule/ policy]
We call a measurable map $f\colon E\to A$ a decision rule. A sequence of decision rules $\pi = (f_0,\dots,f_{N-1})$ is called policy. The set of all policies is denoted by $\Xi$.
\end{defn}
Note that if we assume that the state space $E$ and action space $A$ are not empty, the set of all policies $\Xi$ is not empty, as all constant mappings $f\colon E\to A$ are measurable.

\begin{remark} \label{rem:Pdefn}
    Now, we can describe a Markov Decision Process which depends on a policy.  
    
    Let $\pi=(f_0,\dots,f_{N-1})$ be a fixed policy and $x\in E$ an initial state. 
    One can show that there exists a unique probability measure $\wP_x^\pi$ on $(\Omega,\mathcal{F})$, such that for all $B\in\mathcal{E}$:
    \begin{itemize}
        \item $\wP_x^\pi(X_0\in B)= \delta_x(B)$,
        \item $\wP_x^\pi(X_{n+1}\in B|X_1,\dots,X_n)=\wP_x^\pi(X_{n+1}\in B| X_n) = p(B|X_n,f_n(X_n))$.
    \end{itemize}
    The second property is called Markov property, which tells us, that the next state only depends on the current state and action chosen, and it does not depend on the history of states and chosen actions. 

    Note that in our analysis, we will not go into detail about the initial probability $p_0$, as the next definitions all depend on a current state $x\in\E$ at some time point $n$. 

    Further, note that in our definition of $\wP_x^\pi$, the actions are applied indirectly by having the measure $\wP_x^\pi$ depend on the policy $\pi$.

    From now on, we denote by $\wP_{x}^\pi$ the probability measure from above, given policy $\pi\in\Xi$ and initial state $x\in E$, and by $\E_x^\pi$ the expectation w.r.t. $\wP_x^\pi$. Further, we define $\wP_{nx}^\pi$ as conditional probability $\wP_{nx}^\pi(\cdot)\coloneqq \wP^\pi(\cdot \mid X_n = x)$, which is well-defined by the Markov property. Denote by $\E_{n,x}^\pi$ the associated expectation.
\end{remark}

We move on to compare the performances of different policies within a Markov Decision Model. 
In order to ensure that the following definitions are well defined, we impose following assumption:

Let $r^+\colon E\times A\to\R_{\geq0}$ and $g^+\colon E \to\R_{\geq0}$ denote the positive  part of r and g, i.e. for all $x\in E$ and $a\in A$
\[
r^+(x,a)=\max\{r(x,a),0\} \quad\text{and}\quad g^+(x)=\max\{ g(x),0\}.
\]
Henceforth, we assume that the following inequality holds for all $n\in\{0,\dots,N\}$ and $x\in E$:
\begin{align}
    \sup_\pi \E_{n,x}^\pi\left[\,\sum_{k=n}^{N-1}\beta^kr^+(X_k,f_k(X_k))+\beta^Ng^+(X_N)\right]<\infty\label{eq:eq1}\tag{$A_N$},
\end{align}
where the supremum is taken over all policies $\pi=(f_0,\dots,f_{N-1})\in\Xi$.

\begin{defn}[policy reward]\label{def:policy_reward}
    Let $\pi=(f_0,\dots,f_{N-1})\in\Xi$ be a policy. The \textit{expected total reward starting at time $n$}, given policy $\pi$ and state $x\in E$ is defined as
    \[
    V_{n,\pi}(x)\coloneqq\E_{n,x}^\pi\left[\,\sum_{k=n}^{N-1}\beta^kr(X_k,f_k(X_k))+\beta^Ng(X_N)\right]\in[-\infty,\infty).
    \]
    The \textit{maximal expected reward starting at time $n$} is given as
    \[
    V_n(x)\coloneqq \sup_\pi V_{n,\pi}(x)\in[-\infty,\infty),
    \]
    for $x\in E$. We call the function $V_n\colon E\to\R$ value function.
\end{defn}
\begin{remark}
    By Assumption \ref{eq:eq1} the mappings $V_{n,\pi}$ and $V_n$ are both well defined.

    If there exists a policy $\pi\in\Xi$ with $V_{0,\pi}(x)=V_0(x)$ for all $x\in E$, we call the policy $\pi$ \textit{optimal} and denote the optimal policy $\pi$ by $\pi^*$.
\end{remark}

\clearpage
\section{Bellman Equation}
In the last section we defined Markov Decision Processes on Markov Decision Models, together with policies. In this section we show how to iteratively calculate the value function $V_{n,\pi}$ of a policy $\pi$.

\begin{thm}[reward iteration]\label{thm:reward_iteration}
    Let $\pi=(f_0,\dots,f_{N-1})\in\Xi$ be a policy. Then, for all $n = 0,1\dots,N-1$:
    \begin{enumerate}
        \item $V_{n,\pi}(x) = \beta^nr(x,f_n(x))+\int_E V_{n+1,\pi}(x')p(\mathrm{d}x'|x,f_n(x)),\quad\text{for all }  x\in E,
        $\label{enum:enum1:1}
        \item $V_{N,\pi}(x)=\beta^Ng(x)$, for all $x\in E$. \label{enum:enum1:2}
    \end{enumerate}
\end{thm}
\begin{proof}
    Equation \ref{enum:enum1:2} follows due to
    \begin{align*}
    V_{N,\pi}(x)&=\E_{N,x}^\pi\left[\,\sum_{k=N}^{N-1}\beta^kr(X_k,f_k(X_k))+\beta^Ng(X_N)\right]\\
    &= \E_{N,x}^\pi\left[\,\beta^Ng(X_N)\right]
    = \E^\pi\left[\,\beta^Ng(X_N)\mid X_N=x\right]=\beta^Ng(x),
    \end{align*}
    for all $x\in E$.

    Now, we show equation $\ref{enum:enum1:1}$. Let $x\in E$ be arbitrary. By definition, we get
    \begin{align*}
        V_{n,\pi}(x)
        &=\E_{n,x}^\pi\left[\,\sum_{k=n}^{N-1}\beta^kr(X_k,f_k(X_k))+\beta^Ng(X_N)\right]\\
        &=\E_{n,x}^\pi\left[\,\beta^nr(x,f_n(x))\right] +\E_{n,x}^\pi\left[\,\sum_{k=n+1}^{N-1}\beta^kr(X_k,f_k(X_k))+\beta^Ng(X_N)\right].
    \end{align*}
    Using the tower property of conditional expectations we get
    \begin{align*}
    V_{n,\pi}(x)
    &=\beta^nr(x,f_n(x))+\E_{n,x}^\pi\left[\, \E_{n,x}^\pi\left[\,\sum_{k=n+1}^{N-1}\beta^k r(X_k,f_k(X_k)) +\beta^Ng(X_N)\Big | X_{n+1}\right]\right]\\
    &=\beta^nr(x,f_n(x))+\int_E \E_{n+1,x'}^\pi\left[\,\sum_{k=n+1}^{N-1}\beta^k r(X_k,f_k(X_k)) +\beta^Ng(X_N)\right]p(\mathrm{d}x'|x,f_n(x))\\
    &=\beta^nr(x,f_n(x))+\int_E V_{n+1,\pi}(x')\,p(\mathrm{d}x'|x,f_n(x)),
    \end{align*}
    where the second equation follows remark $\ref{rem:Pdefn}$ $\wP_x^\pi(X_{n+1}\in B| X_n) = p(B|X_n,f_n(X_n))$.
\end{proof}

\clearpage
Given Theorem $\ref{thm:reward_iteration}$, we are able to calculate iteratively the value function $V_{n,\pi}$ of a policy $\pi$. For the remaining part of this section we concentrate on calculating an optimal policy. 
\begin{remark}
   Using a similar backward iteration as in Theorem $\ref{thm:reward_iteration}$, we make the following observation.
   For the maximal expected reward functions $(V_n)_{n\in\{0,\dots,N\}}$ the following equations should hold for all $x\in E$
   \begin{align*}
       V_N(x) &= \beta^Ng(x),\\
       V_n(x) &= \sup_{a\in A}\left\{\beta^nr(x,a)+\int_E V_{n+1}(x')p(\mathrm{d}x'|x,a)\right\},\quad \text{for }n\in\{0,\dots,N-1\}.\label{eq:bellman}\tag{Bellman equation}
   \end{align*}
    This equation is called the Bellman equation.
    
    By definition \ref{def:policy_reward}, the maximal expected reward at time $N$ starting in $x\in E$ has to be $\beta^Ng(x)$. 
    On the other hand, at times $n=0,\dots,N-1$, we expect the maximum expected reward $V_n$ to be the supremum over the direct reward and the maximum expected reward at the next time step, given the best action.
    
    In general, the second equality does not need to hold, since there might not be a $a\in A$, that takes the supremum value. In the case where the supremum is attained, the optimal obtained decision rules $f_n^*$, with 
    \[
    f_n^*(x)\in\mathrm{argmax}_{a\in A}\left\{\beta^nr(x,a)+\int_E V_{n+1}(x')p(\mathrm{d}x'|x,a)\right\},\quad \text{for }n\in\{0,\dots,N-1\},
    \]
    do not have to be measurable.

    We continue to analyse which assumptions are needed to use the Bellman equation and prove that the solutions obtained are optimal.  
    \end{remark}

\begin{thm}[Verification Theorem] \label{thm:verification}
    Let $(v_n\colon E\to [-\infty,\infty])_{n=0,\dots,N}$ be a family of measurable functions, satisfying the \ref{eq:bellman}, i.e.
    \begin{align*}
       v_N(x) &= \beta^Ng(x),\\
       v_n(x) &= \sup_{a\in A}\left\{\beta^nr(x,a)+\int_E v_{n+1}(x')p(\mathrm{d}x'|x,a)\right\},\quad \text{for }n\in\{0,\dots,N-1\}.
   \end{align*}
   \clearpage
   Then, the following holds:
   \begin{enumerate}
       \item For all $n=0,\dots,N$ the inequality $v_n\geq V_n$ holds. \label{enum:verf_thm1}
       \item If a decision rule $f_n^*$ exists such that 
        \[
        f_n^*(x)\in\mathrm{argmax}_{a\in A}\left\{\beta^nr(x,a)+\int_E v_{n+1}(x')p(\mathrm{d}x'|x,a)\right\},\quad \text{for all }n\in\{0,\dots,N-1\},
        \]
        then the policy $\pi^*=(f_0^*,\dots,f_{N-1}^*)$ is optimal.
   \end{enumerate}
\end{thm}
\vspace{-.46cm}
\begin{proof}
    We prove both statements using a backward induction argument.
    \begin{enumerate}
        \item For $N\in\N$, we have by definition $v_N=\beta^Ng(x)=V_N$. Let $n\in\{0,\dots,N-1\}$ be arbitrary but fixed and assume that $v_{n+1}\geq V_{n+1}$ holds. Then, for all policies $\pi=(f_0,\dots,f_{N-1})\in\Xi$ we get
        \begin{align*}
            v_n&=\sup_{a\in A}\left\{\beta^nr(x,a)+\int_E v_{n+1}(x')p(\mathrm{d}x'|x,a)\right\}\\
            &\geq \sup_{a\in A}\left\{\beta^nr(x,a)+\int_E V_{n+1}(x')p(\mathrm{d}x'|x,a)\right\} \quad\text{(induction assumption)}\\
            &\geq \beta^nr(x,f_n(x))+\int_E V_{n+1,\pi}(x')p(\mathrm{d}x'|x,f_n(x))=V_{n,\pi}. 
            \end{align*}
         Since $v_n\geq V_{n,\pi}$ holds for all policies $\pi$, we get $v_{n+1}\geq \sup_\pi V_{n,\pi}=V_n$.
         \item In this case, we assume that decision rules $f_n^*$ exist with
         \[
         f_n^*(x)\in\mathrm{argmax}_{a\in A}\left\{\beta^nr(x,a)+\int_E v_{n+1}(x')p(\mathrm{d}x'|x,a)\right\},\quad \text{for }n\in\{0,\dots,N-1\}.
         \]
         Using another backward induction, we prove that $v_n=V_n=V_{n,\pi^*}$. For the time step $N$ these equalities hold by definition. Let $n\in\{0,\dots,N-1\}$ be arbitrary but fixed and assume that $v_{n+1}=V_{n+1}=V_{n+1,\pi^*}$ holds. Let $x\in E$ be arbitrary:
        \begin{align*}
            V_n(x)&\leq v_n(x) &\quad \text{by \ref{enum:verf_thm1}}\\
            &=\beta^nr(x,f_n^*(x))+\int_E v_{n+1}(x')p(\mathrm{d}x'|x,f_n^*(x))&\text{(existence assumption of $f_n^*$)}\\
            &=\beta^nr(x,f_n^*(x))+\int_E V_{n+1,\pi^*}(x')p(\mathrm{d}x'|x,f_n^*(x))&\text{(induction assumption)}\\
            &=V_{n,\pi^*}(x)&\text{(Theorem \ref{thm:reward_iteration})}\\
            &\leq V_n(x).&&
        \end{align*}
         With that we have shown that for all $n\in\{0,\dots,N\}$ $V_n=V_{n,\pi^*}$
    \end{enumerate}
\end{proof}
\begin{remark}
    Particularly, the theorem states that if a solution of the \textit{\ref{eq:bellman}} exists and the functions $f_n^*$ exist and are measurable, the resulting policy $\pi^*$ is optimal. We impose conditions  such that solutions of the  \textit{\ref{eq:bellman}} exist.
\end{remark}
\begin{defn}[Structure Assumption ($\mathrm{SA}_N$)]\label{def:structure_assumption}
    Let \\ $\mathbb{M}(E)\coloneqq \{v\colon E\to[-\infty,\infty)\mid v\text{ measurable}\}$. We assume that sets $\mathbb{M}_n\subset\mathbb{M}(E)$, $n\in\{0\dots,N\}$ exist, such that for all $n\in\{0,\dots,N-1\}$:
    \begin{enumerate}
        \item $g_N\in\mathbb{M}_N$ \label{enum:SAN:1}
        \item If $v\in\mathbb{M}_{n+1}$ we assume that the following function is well-defined
        \[
        E \ni x\mapsto\sup_a\left\{\beta^nr(x,a)+\int_E v(x')p(\mathrm{d}x'|x,a)\right\}\quad\text{and in $\mathbb{M}_n$}.
        \] \label{enum:SAN:2}
        \item For all $v\in\mathbb{M}_{n+1}$ there exists a function $f_n\colon E\to A$ with
         \[
         f_n(x)\in\mathrm{argmax}_{a\in A}\left\{\beta^nr(x,a)+\int_E v(x')p(\mathrm{d}x'|x,a)\right\},
         \]
         further we assume that $f_n$ is measurable. \label{enum:SAN:3}
    \end{enumerate}
\end{defn}
\begin{thm}[Structure Theorem]
    Assume that the \hyperref[def:structure_assumption]{Structure Assumption $(\mathrm{SA}_N)$} is fulfilled. Then, the following statements hold.
    \begin{enumerate}
        \item The sequence of maximal expected reward functions $(V_n)_{n\in\{0,\dots,N\}}$ fulfill the \ref{eq:bellman}.   \label{enum:struc_thm:1}
        \item There exists an optimal policy $\pi^*=(f_0^*,\dots,f_{N-1}^*)$, where for each \\ $n\in\{0,\dots,N-1\}$ and $x\in E$ the decision policy $f_n^*$ has form
        \[
        f_n^*(x)\in\mathrm{argmax}_{a\in A}\left\{\beta^nr(x,a)+\int_E V_{n+1}(x')p(\mathrm{d}x'|x,a)\right\}.
        \] \label{enum:struc_thm:2}
    \end{enumerate}
\end{thm}
\begin{proof}
    We prove both statements together using a backward induction. For each $n\in\{0,\dots,N-1\}$, we show that decision rules $f_n^*,\dots,f_{N-1}^*$ exist, such that for the policy $\pi^*_n=(h_0, \dots,h_{n-1},f_n^*,\dots,f_{N-1}^*)\in\Xi$, where $h_0,\dots,h_{n-1}$ are any arbitrary decision rules, for all $x\in E$ holds
    \[
    V_{n}(x)=V_{n,\pi^*_n}(x)=\sup_{a\in A}\left\{\beta^nr(x,a)+\int_E V_{n+1}(x')p(\mathrm{d}x'|x,a)\right\}.
    \]
    Note. that for the definition of $V_{n,\pi^*_n}$ only the decision rules $f_n^*,\dots,f_{N-1}^*$ are relevant, the other $h_0,\dots,h_{n-1}$ are not and we are using them to fit the definition of a policy.

    Further, we add to the backward induction, that for all  $n\in\{0,\dots,N\}$ we have $V_n\in\mathbb{M}_n$.
    
    We start with the backward induction.
    
    For the time step $n=N$ we got $V_N=g_N$. So by assumption \hyperref[enum:SAN:1]{$(\mathrm{SA}_N)$ (i)}, we know that $V_N\in \mathbb{M}_N$. 

     Let $n\in\{0,\dots,N-1\}$ be arbitrary but fixed. We assume the following
     \begin{enumerate}
         \item $V_k\in\mathbb{M}_k$ for all $k\in\{n+1\dots,N\}$ \label{enum:struc_theorem:IA:1}
         \item decision rules $f_{n+1}^*,\dots,f_{N-1}^*$ exist, such that for all $k\in\{n+1\dots,N-1\}$
         \[
         f_k^*(x)\in\mathrm{argmax}_{a\in A}\left\{\beta^kr(x,a)+\int_E V_{k+1}(x')p(\mathrm{d}x'|x,a)\right\}
         \]
         \item $V_{n+1,\pi^*_{n+1}}=V_{n+1}$, where $\pi^*_{n+1}=(h_0, \dots,h_{n},f_{n+1}^*,\dots,f_{N-1}^*)\in\Xi$ and $h_0,\dots,h_{n}$ are any arbitrary decision rules \label{enum:struc_theorem:IA:2}
     \end{enumerate}
     
    By the induction assumption \ref{enum:struc_theorem:IA:1} we get $V_{n+1}\in\mathbb{M}_{n+1}$. Using
    assumption \hyperref[enum:SAN:3]{$(\mathrm{SA}_N)$ (iii)}, we know that for the function $V_{n+1}\in\mathbb{M}_{n+1}$ a measurable function $f_n^*\colon E\to A$ exists with
    \[
    f_n^*(x)\in\mathrm{argmax}_{a\in A}\left\{\beta^nr(x,a)+\int_E V_{n+1}(x')p(\mathrm{d}x'|x,a)\right\},
    \]
    for all $x\in E$.
    
    Now, consider the policy $\pi^*_n=(h_0, \dots,h_{n-1},f_n^*,\dots,f_{N-1}^*)\in\Xi$, where $h_0,\dots,h_{n-1}$ are any arbitrary decision rules and $f_{n+1}^*,\dots,f_{N-1}^*$ are from the induction assumption. With this policy we get for all $x\in E$
    \begin{align*}
        V_{n,\pi^*_n}(x)&=\beta^nr(x,f_n^*(x))+\int_E V_{n+1,\pi^*_n}(x')p(\mathrm{d}x'|x,f_n^*(x))&\text{(Theorem \ref{thm:reward_iteration})}\\
        &=\beta^nr(x,f_n^*(x))+\int_E V_{n+1,\pi^*_{n+1}}(x')p(\mathrm{d}x'|x,f_n^*(x))&\text{choose $h_n=f_n^*$ in $\pi^*_{n+1}$}\\
        &=  \beta^nr(x,f_n^*(x))+\int_E V_{n+1}(x')p(\mathrm{d}x'|x,f_n^*(x))&\text{(induction assumption \ref{enum:struc_theorem:IA:2})}\\  
        &=  \sup_a\left\{\beta^nr(x,a)+\int_E V_{n+1}(x')p(\mathrm{d}x'|x,a)\right\}.&\text{(definition of $f_n^*$)}
    \end{align*}
    With that we have
    \[
    V_{n}(x)\geq \sup_a\left\{\beta^nr(x,a)+\int_E V_{n+1}(x')p(\mathrm{d}x'|x,a)\right\}.
    \]
    On the other hand, we have for any arbitrary policy $\pi=(f_0,\dots,f_{N-1})\in\Xi$ and all $x\in E$
    \begin{align*}
        V_{n,\pi}(x)&=\beta^nr(x,f_n(x))+\int_E V_{n+1,\pi}(x')p(\mathrm{d}x'|x,f_n(x))&\text{(Theorem \ref{thm:reward_iteration})}\\
        &\leq\beta^nr(x,f_n(x))+\int_E V_{n+1}(x')p(\mathrm{d}x'|x,f_n(x))&\text{$V_{n+1}\geq V_{n+1,\pi}$}\\  
        &\leq\sup_a\left\{\beta^nr(x,a)+\int_E V_{n+1}(x')p(\mathrm{d}x'|x,a)\right\}.&
    \end{align*}
    Combining this inequality with the one from above and the fact that the policy $\pi_n^*$ attains the bound, we get 
    \[
    V_{n,\pi^*_n}(x)= \sup_a\left\{\beta^nr(x,a)+\int_E V_{n+1}(x')p(\mathrm{d}x'|x,a)\right\} =V_n(x),
    \]
    for all $x\in E$. By \hyperref[enum:SAN:2]{$(\mathrm{SA}_N)$ (ii)}, we now have $V_n\in\mathbb{M}_n$, which ends our induction.

    We have proven that under the \hyperref[def:structure_assumption]{Structure Assumption $(\mathrm{SA}_N)$} the functions $(V_n)_{n\in\{0,\dots,N\}}$ fulfil the \ref{eq:bellman} and we have shown how to calculate each $V_n$.
\end{proof}
\begin{remark}
    Under the \hyperref[def:structure_assumption]{Structure Assumption $(\mathrm{SA}_N)$}, the Structure Theorem gives us the following algorithm, to calculate $(V_n)_{n\in\{0,\dots,N\}}$ and optimal policies.

    \SetKwComment{Comment}{/* }{ */}
    \RestyleAlgo{ruled}
    \begin{algorithm}
    \caption{Backward Induction}\label{alg:backward_induction}
    \KwData{Markov Decision Model $(E,A,p,p_0,r,\beta,g)$, $N\in\N$}
    \KwResult{$(V_n)_{n\in\{0,\dots,N\}}$ and optimal policy $\pi^*=(f^*_0,\dots,f^*_{N-1})$}
    $V_N=g$\;
    \For{$n=N-1,\dots,0$}{
    For all $x\in E$ calculate $V_n(x)=\sup_{a\in A}\left\{\beta^nr(x,a)+\int_E V_{n+1}(x')p(\mathrm{d}x'|x,a)\right\}$\;
    Calculate measurable $f_n^*$ that takes the supremum values\;
    }
    \end{algorithm}
    This algorithm works great for the calculation of the maximal expected rewards $V_n$, $n=0,\dots,N$ in case of $A$ and $E$ being finite and if the transition kernel $p$ is known, as the supremum calculation for all $x\in E$ over the action space $A$ is in feasible time possible. In the next section, we see how to use the Bellman equation in case the transition kernel is unknown and in section \ref{sec:rl_algorihms} how to deal with uncountable sets $A$ and $E$.
\end{remark}
In the last part of this section we define the $Q$-functions.
\begin{defn}[$Q$-function]\label{def:Q_funciton}
    Let $\pi=(f_0,\dots,f_{N-1})\in\Xi$ be a policy. For all $n=0,\dots,N-1$ and $a\in A$ we consider the modified policy \\ $\pi_{n,a}=(f_0,\dots,f_{n-1},x\mapsto a,f_{n+1},\dots,f_{N-1})$. We then define for all $n=0,\dots,N-1$, $x\in E$, $a\in A$
    \[
    Q^\pi(n,x,a)=\E_{n,x}^{\pi_{n,a}}\left[\,\sum_{k=n}^{N-1}\beta^kr(X_k,f_k(X_k))+\beta^Ng(X_N)\right],
    \]
    and for time step $n=N$
    \[
    Q^\pi(N,x,a)=\beta^Ng(x),
    \]
    which we call the $Q$-function of the policy $\pi$. The \textit{optimal $Q$-function} or $\textit{action value function}$ is given by
    \[
    Q(n,x,a)=\sup_{\pi\in\Xi}Q^\pi(n,x,a).
    \]
\end{defn}
\begin{remark}
    Note that the modified policy $\pi_{n,a}$ defined in \ref{def:Q_funciton} is a valid policy since the function $f\colon E\to A$ with $x\mapsto a$ is measurable for all $a\in A$ and thus a decision rule. Further, the optimal $Q$-function is well defined by assumption \ref{eq:eq1}.

    Similarly as for the value function $V_n$,  the Bellman equation holds for the $Q$-function:
    \begin{align*}
       Q(n,x,a) &= \beta^nr(x,a)+\int_E \sup_{a'\in A}Q(n+1,x',a')p(\mathrm{d}x'|x,a),
    \end{align*}
    for all $n\in\{0,\dots,N-1\}$.
\end{remark}

\clearpage

\section{Temporal Difference Learning}
This section is based on the Book \textit{Reinforcement Learning: An Introduction} by R. S. Sutton and A. G. Barto \cite{rl_book} and on the lecture notes of the course \textit{Machine Learning with Financial Applications
} by Prof. Dr. Knochenhauer and Dr. Bayer \cite{mlfin}.
We show how to compute the expected total reward $V_{n,\pi}$ for a policy $\pi\in\Xi$ when the transition kernel $p$ is not known.
\subsection{Monte Carlo Method}
One possibility to approximate the function $V_{n,\pi}$ is to average the rewards obtained following the policy $\pi$. Let $n\in\{0,\dots,N\}$ and $x\in E$ be arbitrary and assume that $X^1,\dots, X^M$ are independent realisations of the Markov Decision Process following the policy $\pi=(f_0,\dots,f_{N-1})$, with $X_n^m=x$ for all $m\in\{1,\dots,M\}$. We define
\[
W_M^\pi(n,x)\coloneqq \frac1M\sum_{m=1}^M\left[
\sum_{k=n}^{N-1}\beta^kr(X_k^m,f_k(X_k^m))+\beta^Ng(X_N^m)
\right],
\]
as the average returned reward starting in state $x$ at time $n$ and following policy $\pi$. In the case where the rewards $r$ and $g$ are bounded, one can show with the law of large numbers that $W_M^\pi(n,x)\to V_{n,\pi}(x)$ for $M\to\infty$, in probability.

\begin{remark}
    We make the following observation for all $n\in\{0,\dots,N\}$ and $x\in E$
    \begin{align*}
W_{M+1}^\pi(n,x)&=\frac{1}{M+1}\frac{M}{M} \sum_{m=1}^{M+1} \left[
\sum_{k=n}^{N-1} \beta^k r(X_k^m,f_k(X_k^m)) +\beta^Ng(X_N^m)
\right]\\
&=\frac{1}{M+1}\left[
\sum_{k=n}^{N-1} \beta^k r(X_k^{M+1},f_k(X_k^{M+1})) +\beta^Ng(X_N^{M+1})
\right]\\
&+\frac{1}{M+1}\frac{M}{M}\sum_{m=1}^{M} \left[
\sum_{k=n}^{N-1} \beta^k r(X_k^m,f_k(X_k^m)) +\beta^Ng(X_N^m)
\right]\\
&=\frac{1}{M+1}\left[
\sum_{k=n}^{N-1} \beta^k r(X_k^{M+1},f_k(X_k^{M+1})) +\beta^Ng(X_N^{M+1})
\right]+ \frac{M}{M+1}W_{M}^\pi(n,x)\\
&=W_{M}^\pi(n,x)+\frac{1}{M+1}\left[\left(
\sum_{k=n}^{N-1} \beta^k r(X_k^{M+1},f_k(X_k^{M+1})) +\beta^Ng(X_N^{M+1})\right)-W_{M}^\pi(n,x)
\right].
    \end{align*}
This gives us an iterative way to update the approximation $W_M^\pi$, whenever more realisations of the Markov Decision Process have been made.
\end{remark}
\clearpage
\subsection{Temporal Difference Learning}
Temporal Difference Learning describes a similar approach to approximating the value function $V_{n,\pi}$ as the Monte Carlo Method. It is based on the ideas shown in the proof of the reward iteration theorem \ref{thm:reward_iteration}. In the proof of theorem \ref{thm:reward_iteration}, we found that the following equality holds for all $\pi\in\Xi$, $n\in\{0,\dots,N-1\}$ and $x\in E$
\begin{align*}
V_{n,\pi}(x)&\coloneqq\E_{n,x}^\pi\left[\,\sum_{k=n}^{N-1}\beta^kr(X_k,f_k(X_k))+\beta^Ng(X_N)\right]\\
&=\E^\pi_{n,x}\bigg[\beta^nr(x,f_n(X_n))+V_{n+1,\pi}(X_{n+1})\bigg].
\end{align*}
As for the Monte Carlo method let $n\in\{0,\dots,N\}$ and $x\in E$ be arbitrary and assume that $X^1,\dots, X^M$ are independent realizations of the Markov Decision Process following the policy $\pi=(f_0,\dots,f_{N-1})$, with $X_n^m=x$ for all $m\in\{1,\dots,M\}$.

With that we can define an approximation to $V_{n,\pi}(x)$ as
\[
W_M^\pi(n,x)\coloneqq\frac1M\sum_{m=1}^M\left[\beta^k r(X_n^m,f_k(X_n^m)) + W_M^\pi(n+1,X_{n+1}^m) 
\right],
\]
with $W_M^\pi(N,x')\coloneqq \beta^Ng(x')$ for all $x'\in E$. As the definition depends on $W_M^\pi(n+1,\cdot)$ the iterative update rule is given by
\[
W_{M+1}^\pi(n,x)=W_{M}^\pi(n,x)+\frac{1}{M+1}\bigg(\beta^nr(x,f(x))+W_M^\pi(n+1,X_{n+1}^{m+1})-W_M^\pi(n,x)\bigg).  
\]

\subsection{Policy Improvement}
In this subsection we give the idea to a fundamental concept in reinforcement learning, called policy improvement. Although, it is not directly related to Temporal Difference Learning, we decided to explain the policy improvement here, as the shown algorithms in the next section\ref{sec:rl_algorihms} rely on this concept.

We start with the following theorem.
\begin{thm}[Policy Improvement Theorem]\label{thm:policy_improvement}
Let $\pi=(f_0,\dots,f_{N-1}),\pi'=(f_0',\dots,f_{N-1}')$ be two policies, such that for all $x\in E$ and $n=0,\dots,N$
\[
Q^\pi(n,x,f_n'(x))\geq V_{n,\pi}(x).
\]
Then, for all $x\in E$ and $n=0,\dots,N$ we get
\[
V_{n,\pi'}(x)\geq V_{n,\pi}(x).
\]
\end{thm}
The proof can be done in a straightforward manner, see e.g. in the book \cite{rl_book} section 4.2. 

With this theorem, we are able to create a  better policy $\pi'=(f_0',\dots,f_{N-1}')$, where for all  $x\in E$ and $n=0,\dots,N$ we pick
\[
f_n'(x)\in\argmax_{a\in A}Q^\pi(n,x,a).
\]
Note that these functions $f_n'$ might not be measurable anymore, thus they are invalid decision rules. As we give in this subsection only an idea on how the policy iteration works, we do not go into more detail.

With the observation that $V_{n,\pi}(x)=Q^\pi(n,x,f_n(x))$, we immediately get that
\[
V_{n,\pi}(x)=Q^\pi(n,x,f_n(x))\leq Q^\pi(n,x,f_n'(x)),
\]
thus the policy $\pi'$ satisfies the conditions needed for the policy improvement theorem \ref{thm:policy_improvement}. Therefore, we have found a policy $\pi'$, such that 
\[
V_{n,\pi'}(x)\geq V_{n,\pi}(x),
\]
for all $x\in E$ and $n=0,\dots,N$, by the policy improvement theorem \ref{thm:policy_improvement}.

We can now apply the policy improvement to the new policy $\pi'$ and get an even better policy.  In case the set of all policies $\Xi$ is finite (if e.g. $E,A$ are finite), this procedure will eventually give us to two policies $\pi$ and $\pi'$, such that $\pi=\pi'$. In this case, we know that
\[
V_{n,\pi}(x)=\max_{a\in A}\left\{\beta^nr(x,a)+\int_E V_{n+1,\pi}(x')p(\mathrm{d}x'|x,a)\right\},\quad \text{for }n\in\{0,\dots,N-1\},
\]
thus by the verification theorem \ref{thm:verification}, we get that $\pi$ is optimal.

This idea is implemented in the following algorithms.

\section{Algorithms}\label{sec:rl_algorihms}
In this section, we see how the theoretical advances are used to establish an algorithm for finding good policies. To do this, we first present the DDPG algorithm. Then, we introduce a slight modification of the DDPG algorithm, called the TD3 algorithm, which we used in the application. 
\subsection{DDPG}
This section is based on the paper \textit{'Continuous control with deep reinforcement learning'}, see \cite{ddpg_paper}. 

The initial idea for the DDPG-algorithm is the following. First, assume we have a set of functions
\[
\{Q_{\theta^Q}\colon E\times A\to\R\mid \theta^Q\in\R^d\},
\]
for some $d\in\N$, parameterised by some $\theta^Q\in\R^d$ for example by using neural networks. For each $\theta^Q$, we consider the greedy policy $\pi_{\theta^Q}(s)\in\mathrm{arg}\max_a Q^{\theta^Q}(s,a)$ and assume it exists. Note that this decision rule is time-independent. As the DDPG algorithm uses the same decision rule for all time steps, we use the term policy for the decision rule. In the case where time dependence is important to a problem, the current time step can be added to the state space, allowing time dependence. Our goal is to find a function approximator $Q_{\theta^Q}$ that minimises the loss
\begin{align}
L(\theta^Q)=\E_{s_t\sim \rho^{\theta^Q}}[(Q_{\theta^Q}(s_t,\pi_{\theta^Q}(s_t))-y_{\theta^Q,t})^2],\label{func:L}
\end{align}
where $\rho^{\theta^Q}$ is the state visitation distribution for the policy $\pi_{\theta^Q}$ and $y_{\theta^Q,t}$ is defined as 
\[
y_{\theta^Q,t}\coloneqq r(s_t,\pi_{\theta^Q}(s_t))+\beta Q_{\theta^Q}(s_{t+1} ,\pi_{\theta^Q}(s_{t+1})).
\]
We use this loss to train the $Q$-function to satisfy the reward iteration property \ref{enum:enum1:1}.

In this initial idea, some steps cannot be efficiently implemented or could lead to problems. We present the solutions described in the paper \cite{ddpg_paper}, which result in the formulation of the DDPG-algorithm.

In the following, we point out several problems in this initial idea, show solutions that the DDPG algorithm uses and finally write down the algorithm. 

We encounter the first problem in the computation of the greedy policy $\pi_{\theta^Q}$. Finding the maximising action $a_t$ for $Q^{\theta^Q}(s_t,a_t)$ for each time-step $t$ is computationally expensive for uncountable action spaces. Instead, we parameterise the policy function like the $Q$-function and update the policy parameters to maximise the reward. Let $\theta^\pi\in\R^m$, for some $m\in\N$, be the parameters for functions $\pi_{\theta^\pi}\colon E\to A$, for example by using neural networks.

Let $J\coloneqq \E_{x\sim p_0}[V_{0,\pi_{\theta^\pi}}(x)]$. According to the paper by Silver et al. (2014) \cite{dpg_paper}, the gradient of the policy performance can be written as
\begin{align*}
\nabla_{\theta^\pi} J &\approx \E_{s_t\sim \rho^{\theta^\pi}}\big[\nabla_{\theta^\pi}Q_{\theta^Q}(s,a)|_{s=s_t,a=\pi_{\theta^\pi}(s_t)}\big]\\
&=\E_{s_t\sim \rho^{\theta^\pi}}\big[
\nabla_aQ_{\theta^Q}(s,a)|_{s=s_t,a=\pi_{\theta^\pi}(s_t)} \nabla_{\theta^\pi}\pi_{\theta^\pi}(s)|_{s=s_t}\big].
\end{align*}
We update the policy using this gradient, such that the update maximises the expected return from the initial distribution. The parameters $\theta^Q$ for the $Q$-function are updated according to the loss \hyperref[func:L]{$L$}, see \ref{func:L}.

To compute the expectation with respect to the state visitation distribution $\rho^{\theta^Q}$, we store recent observed transitions $(s_t,a_t,r_t,s_{t+1})$ in a finite length queue called \textit{replay buffer}. At each time step we uniformly sample a minibatch from the replay buffer and update the parameters $\theta^Q$ and $\theta^\pi$. This helps to make the samples more independent and identically distributed than if they were generated by sequentially exploring the environment.

Another problem observed was that due to the definition of the loss $L$ in the equation \ref{func:L} the updates were prone to divergence, as the updated parameters $\theta^Q$ were also part of the target $y_{\theta^Q,t}$. The proposed solution is to use copies $\theta^{Q'},\theta^{\pi'}$ of the parameters $\theta^{Q},\theta^{\pi}$ for the target calculation in $y_{\theta^{Q},t}$ in the equation \ref{func:L}. The copies $\theta^{Q'},\theta^{\pi'}$ are then updated according to the following rule
\[
\theta'\leftarrow \tau\theta+(1-\tau)\theta',
\]
for a $\tau\ll1$, $\theta\in\{\theta^Q,\theta^\pi\}$ and the learned parameters $\theta^Q,\theta^\pi$. This may slow down the learning process but increases the stability of the problem.

The final challenge is to figure out how to explore the continuous action space. The paper \cite{ddpg_paper} suggests to change the policy used by adding noise from a noise process $\mathcal{R}$ to the policy $\pi_{\theta^\pi}$. So, we use the function
\[
\pi'(s_t)=\pi_{\theta^\pi}(s_t)+\mathcal{R}_t
\]
to do the exploration. Note that by adding the noise, the action may no longer be within the action space $\mathcal{A}$. Typically, clipping the action resolves this. 

The choice of the noise process $\mathcal{R}$ depends on the environment. 
Typically, the Ornstein-Uhlenbeck process  or the Gaussian process is used.

\vspace{.5cm}
Algorithm \ref{alg:ddpg} shows the DDPG algorithm. In the beginning, the different parameters and the replay buffer are initialised. After that, for a finite amount of episodes beginning at initial observed state $s_1$, actions are made according to the current policy plus noise. These transitions are then stored in the replay buffer, from which we sample a random minibatch and update the parameters $\theta^Q,\theta^\pi$. Note that in the update of the $Q$-function, the target values $y_i$ depend on the target parameters $\theta^{Q'},\theta^{\pi'}$. These target parameters are then updated in the last step according to the update rule $\theta'\leftarrow \tau\theta+(1-\tau)\theta'$, for $\theta \in\{\theta^Q,\theta^\pi\}$.   

\SetKwComment{Comment}{/* }{ */}
    \RestyleAlgo{ruled}
    \begin{algorithm}[H]
    {\small
    \caption{DDPG algorithm}\label{alg:ddpg}
    
    Randomly initialise Q-network $Q^{\theta^Q}(s,a)$ and policy $\pi^{\theta^\pi}(s)$ with parameters $\theta^Q,\theta^\pi$.\\
    Initialise target parameters with $\theta^{Q'}\leftarrow \theta^Q$, $\theta^{\pi'}\leftarrow \theta^\pi$.\\
    Initialise replay buffer $R$.\\
    \For{episode $=1,\dots,M$}{
    Initialise random process $\mathcal{R}$ for action exploration.\\
    Receive initial observation state $s_1$.\\
    \For{$t=1,\dots,N$}
    {
    Select action $a_t=\pi^{\theta^\pi}(s_t)+\mathcal{R}_t$\\
    Execute action $a_t$, observe reward $r_t$ and new state $s_{t+1}$.\\
    Store transition $(s_t,a_t,r_t,s_{t+1})$ in $R$\\
    Sample random minibatch of $K$ transitions $(s_i,a_i,r_i,s_{i+1})$ from $R$\\
    Set $y_i=r_i+\beta Q^{{\theta^{Q}}'}(s_{i+1},\pi^{{\theta^{\pi}}'}(s_{i+1}))$\\  
    Update $Q$-function parameters $\theta^Q$, minimising $L=\frac{1}{K}\sum_i(y_i-Q^{\theta^Q}(s_i,a_i))^2$\\
    Update policy parameters using the sampled policy gradient
    \[
    \nabla_{\theta^\pi}J\approx \frac{1}{K}\sum_i\nabla_a Q^{\theta^Q}(s,a)|_{s=s_i,a=\pi^{\theta^\pi}(s_i)}\nabla_{\theta^\pi}\pi^{\theta^\pi}(s)|_{s=s_i}
    \]
    
    Update target networks
    \[
    \theta'\leftarrow\tau\theta+(1-\tau)\theta', \text{ for $\theta\in\{\theta^Q,\theta^\pi\}$}
    \]
    }
    }
    }  
    \end{algorithm}

\vspace{.5cm}
We see that the DDPG algorithm depends on many parameters that need to be chosen. To begin with, there are structure parameters for the $Q$ and policy network, which are typically deep neural networks, further, parameters for the size of the replay buffer, for the random noise process $\mathcal{R}$, learning rates for the $Q$-function and policy updates and at last the parameter $\tau$. The performance of the DDPG algorithm strongly depends on those so called hyperparameters. Thus, in section \ref{section:hyper}, we investigate how to efficiently find good hyperparameters.


\subsection{TD3}\label{sec:TD3}
In this subsection, we introduce the Twin Delayed Deep Deterministic policy gradient algorithm, in short TD3 algorithm, which is an extension of the DDPG algorithm. The subsection is based on the paper \textit{'Addressing Function Approximation Error in Actor-Critic Methods'} by S. Fukimoto, H. van Hoof and D. Meger \cite{td3_paper}. In this paper, the use of target networks is linked to the overestimation of the $Q$-function, a common problem in value-based reinforcement learning methods, like in the DDPG algorithm. Here we show the improvements they suggest and the resulting algorithm.

In the paper \textit{'Issues in Using Function Approximation for Reinforcement Learning'} by S. Thrun and A. Schwartz \cite{overestimation_paper}, it is shown that if the value estimates are updated with a greedy target $y=r+\beta \max_{a'}Q(s',a')$ and the target is susceptible to some mean-zero error $\epsilon$ then the value with the error will, in general, be greater than the true maximum, as $\E_\epsilon[\max_{a'}(Q(s',a')+\epsilon)]\geq\max_{a'}Q(s',a')$. This overestimation is then propagated through the Bellman equation to the whole $Q$-function. Although, in the DDPG algorithm, the value updates aren't directly updated with a greedy target, the presence of the overestimation bias is still given, as shown in experiments in (Fukimoto, van Hoof and D. Meger, 2018) \cite{td3_paper}.

As this problem is induced by the function approximators, it is unavoidable. The following changes were made in the DDPG algorithm to try to minimise the impact.

The first improvement is to train 2 different $Q$-functions, instead of one. This will give us 2 separate value estimates. We then take the smaller value in the target computation $y_i$ and use that value to update both $Q$-functions. To be more precise, 2 random $Q$-value parameters $\theta_1,\theta_2$ will be initialised at the beginning. For a transition $(s_t,a_t,r_t,s_{t+1})$ from the replay buffer $\mathcal{R}$, we calculate the target value $y$ as $y=r_t+\beta\min_{j=1,2}Q^{{\theta_j}'}(s_{t+1},\pi^{{\theta^\pi}'}(s_t))$ and update the parameters $\theta_1,\theta_2$ to minimise $(y-Q^{\theta_j}(s_t,a_t))^2$ for $j=1,2$.

The next change is that the updates for the policy parameters $\theta^\pi$ and target networks will only be executed after a fixed number $d\in\N$ of updates on the $Q$-function parameters $\theta_1,\theta_2$. The reason for this change is to reduce the likelihood of repeatedly updating the policy and target parameters in  the case where the $Q$-function parameters do not change. This lowers the variance of the value estimates for the policy updates, and thus improves the quality of the policy updates.

The last change is in regard of the concern that the policy may overfit in the value estimate. The updates on the $Q$-function are highly affected by the function approximation errors if the learned target is using a deterministic policy, which increases the variance of the target. In order to reduce the variance, a regularisation term has been introduced. The target $y$ is calculated for a transition $(s_t,a_t,r_t,s_{t+1})$ as
 \[
 y=r_t+\beta\min_{j=1,2}Q^{{\theta_j}'}(s_{t+1},\pi^{{\theta^\pi}'}(s_t)+\epsilon),
 \]
where $\epsilon\sim\mathrm{clip}(\mathcal{N}(0,\sigma),-c,c)$, for some hyperparameters $\sigma,c$ and $\mathcal{N}$ is the normal distribution and $\mathrm{clip}$ is a clipping operator. The clipping has been added, such that the target remains close to the original action. 

In the paper introducing the TD3 algorithm \cite{td3_paper}, they set the noise process $\mathcal{R}$ as the normal distribution $\mathcal{N}(0,\sigma)$. We change that and allow any noise process, such that it can be fitted to the problem. Further, we make a small modification at the initialisation process of the replay buffer $R$. For a finite number of runs $L$, we add transitions to the replay buffer, using a uniform random policy, allowing transitions to be in the replay buffer from the start.

The TD3 algorithm is then given by

    \begin{algorithm}[H]
    {\small
    \caption{TD3 algorithm}\label{alg:td3}
    
    Randomly initialise Q-networks $Q^{\theta_1},Q^{\theta_2}$ and policy $\pi^{\theta^\pi}$ with parameters $\theta_1,\theta_2,\theta^\pi$.\\
    Initialise target parameters with $\theta_j'\leftarrow \theta_j$, for $j=1,2$, $\theta^{\pi'}\leftarrow \theta^\pi$.\\
    Initialise replay buffer $R$. Add $L$ runs of uniform random transitions to $R$.\\
    \For{episode $=1,\dots,M$}{
    Initialise random process $\mathcal{R}$ for action exploration.\\
    Receive initial observation state $s_1$.\\
    \For{$t=1,\dots,N$}
    {
    Select action $a_t=\pi^{\theta^\pi}(s_t)+\mathcal{R}_t$\\
    Execute action $a_t$, observe reward $r_t$ and new state $s_{t+1}$.\\
    Store transition $(s_t,a_t,r_t,s_{t+1})$ in $R$\\
    Sample random minibatch of $K$ transitions $(s_i,a_i,r_i,s_{i+1})$ from $R$\\
    Define $\tilde a_i\leftarrow \pi^{{\theta^\pi}'}+\epsilon$, where $\epsilon\sim\mathrm{clip}(\mathcal{N}(0,\sigma),-c,c)$\\
    Set $y_i=r_i+\beta\min_{j=1,2}Q^{{\theta_j}'}(s_{i+1},\tilde a_i)$\\  
    Update $Q$-function parameters $\theta_j$, minimising $L=\frac{1}{K}\sum_i(y_i-Q^{\theta_j}(s_i,a_i))^2$, for $j=1,2$\\
    \If{$t\;\mathrm{mod}\;d$}
    {
    Update policy parameters using the sampled policy gradient
    \[
    \nabla_{\theta^\pi}J\approx \frac{1}{K}\sum_i\nabla_a Q^{\theta_1}(s,a)|_{s=s_i,a=\pi^{\theta^\pi}(s_i)}\nabla_{\theta^\pi}\pi^{\theta^\pi}(s)|_{s=s_i}
    \]
    Update target networks
    \[
    \theta'\leftarrow\tau\theta+(1-\tau)\theta', \text{ for $\theta\in\{\theta_1,\theta_2,\theta^\pi\}$}
    \]
    }
    }
    }
    }  
    \end{algorithm}
Note that our TD3 algorithm implantation has even more hyperparameters than the DDPG algorithm. In total we have the following hyperparameters: shape $Q$-networks, policy networks, size replay buffer, amount of start runs $L$, number of episodes $M$, random process $\mathcal{R}$, size of minibatches $K$, clipping parameter $c$ and the number $d\in\N$.

In some applications the discount factor $\beta$ can be seen as a hyperparameter. For example, in cases where a Markov Decision Model has the discount factor $\beta=1$, one could change the discount factor in the algorithm and possibly get better results using the new discount factor.

\subsection{Neural Networks}
This subsection provides a brief introduction to neural networks. It is based on the book \textit{Understanding Machine Learning:
From Theory to Algorithms} by Shai Shalev-Shwartz and Shai Ben-David \cite{nn_book} and on my Bachelor's thesis \cite{BA_schuett}. 

Reinforcement learning algorithms use neural networks as approximators for the $Q$-function and policy. 
They are exceptionally useful due to their capacity for flexible adjustment of their complexity. Moreover, it can
be shown that neural networks are able to approximate any continuous function up to any $\varepsilon>0$, see \ref{thm:general_approx_thm}.

\begin{defn}[Neural Networks]
Let $\mathbb{X}\subseteq \R^{d_1}$ be the input space and $\mathbb{Y}\subseteq\R^{d_{n+1}}$ the output space for some $d_1,d_{n+1}\in\N$, of a function which we want to approximate.

Neural networks are a composition of functions $\Phi^{(1)},\dots,\Phi^{(n)}$, where each $\Phi^{(i)}\colon\R^{d_i}\to\R^{d_{i+1}}$, $d_i,d_{i+1}\in\N$, for $i\in\{1,\dots,n\}$, i.e.
\[
\Phi(x)=\Phi^{(n)}(\cdots(\Phi^{(1)}(x))\cdots),
\]
where for $i=1,\dots,n-1$ the function $\Phi^{(i)}$ has the structure:
\[
\Phi^{(i)}(x)=\psi(W^{(i)} x+b^{(i)}),
\]
and $\Phi^{(n)}$ has the form
\[
\Phi^{(n)}(x)=W^{(n)} x+b^{(n)},
\]
for a non-linear function $\psi\colon\R\to\R$, that is applied component wise, $W^{(i)}\in\R^{d_{i+1}\times d_i}$ and $b^{(i)}\in\R^{d_{i+1}}$.

The function $\psi$ is called the activation function, $W^{(i)}$ the weight matrix and $b^{(i)}$ the bias and define  $W:=(W^{(1)},\dots,W^{(n)})$ and $b:=(b^{(1)},\dots,b^{(n)})$. The number $n$ is called the number of hidden layers and for all $i\in\{1,\dots,n\}$ the number $d_i$ is called the number of nodes in the hidden layer $i$. 
\end{defn}

\begin{remark}
    For neural networks, the number of hidden layers $n$ and the number of nodes $d_i$ in each hidden layer $i$, for $i\in\{1,\dots,n\}$, can be used to vary the complexity of the neural network.

    Choosing the right number of hidden layers and nodes within each hidden layer is not an easy task and is highly dependent on the given problem. Therefore, in the Hyperparameter search section \ref{section:hyper}, we show how to efficiently find good hyperparameters. 

    In our further analysis, we use the non-linear activation function called \textit{rectified linear unit (ReLU)}, which is defined as 
    \[
    \psi_{\mathrm{relu}}(x)\coloneqq\left\{\begin{array}{cc}
      x& \text{if }\,\, x>0  \\
      0& \text{else} 
 \end{array}\right..
    \]
\end{remark}

The following theorem shows us that any continuous function can be approximated by neural networks. 
\begin{thm}[General Approximation Theorem, by \cite{wolf2018mathematical}]\label{thm:general_approx_thm}
Let $K\subseteq\R^{d_1}$ be compact and let $\psi\colon\R\to\R$ be an activation function, that is continuous and non-polynomial. Then, the set of functions representable by a neural network with a single hidden layer and activation function $\psi$ is dense in the space of continuous functions $f\colon K\to\R^{d_{n+1}}$ in the topology of uniform convergence. I.e. for all continuous function $f\colon K\to\R^{d_{n+1}}$ and $\varepsilon>0$ there exists a neural network $\Phi$, with depth 1, s.t.
\[
\sup_{x\in K}||\Phi(x)-f(x)||<\varepsilon.
\]
\end{thm}
The proof can be read in the book \textit{Mathematical foundations of supervised learning} by M. Wolf \cite{wolf2018mathematical} (see Theorem 2.7 in the book). Note that the ReLU activation function is continuous and non-polynomial.

\vspace{.5cm}
Now, we define another type of neural networks called convolutional neural networks. They are typically used, when 2D arrays are given as input. One could flatten the input spaces, but for some 2D array inputs, like pictures, the information of pixels next to each other might get lost. We have used them in our application to encode a histogram plot of the electrons to a 4-dimensional vector, that gives us information about the shape of the electrons in the histogram plot.   

The following part is based on \cite{cnn_book}.
\begin{defn}[Convolutional neural networks]\label{defn_cnn}
Let $\mathbb{X}\subseteq \R^{d_1,d_1'}$ be the input-space and $\mathbb{Y}\subseteq\R^{d_{n+1}}$ the output space for some
$d_1,d_1',d_{n+1}\in\N$, which we want to approximate.

Convolutional neural networks are a composition of functions $\Phi^{(1)},\dots,\Phi^{(n)}$, where each $\Phi^{(i)}\colon\R^{d_i,d_i'}\to\R^{d_{i+1},d_{i+1}'}$, $d_i,d_i',d_{i+1},d_{i+1}'\in\N$, for $i\in\{1,\dots,n-1\}$ and $\Phi^{(n)}\colon\R^{d_{n},d_{n}'}\to\R^{d_{n-1}\cdot d_{n-1}'}$, i.e.
\[
\Phi(x)=\Phi^{(n)}(\cdots(\Phi^{(1)}(x))\cdots),
\]
where $\Phi^{(n)}$ flattens/vectorises 2D arrays into vectors.

For all other $\Phi^{(i)}$, one of the following description holds
\begin{itemize}
    \item (activation layer) For a non-linear function $\psi\colon\R\to\R$, the function $\Phi^{(i)}\colon \R^{d_i,d_i'}\to\R^{d_i,d_i'}$ is the function, that applies the non-linear function component wise
    \item (max pooling layer) Let $A=(A_{i,j})\in\R^{d_i,d_i'}$. For stride lengths $s,s'\in\N$ and  a window size $(p+1)\times (p'+1)$, where $p,p'$ both  are even, the output $\Phi^{(i)}(A)=\widetilde A\in\R^{d_i-2s,d_i'-2s'}$ is given by the following rule for all $i\in\{s+1,\dots,d_i-s\}$ and $j\in\{s'+1,\dots,d_i'-s'\}$
    \[
    \widetilde A_{i-s,j-s'}=\max\big\{A_{k,l}\;|\;k\in\{i-\frac{p}{2},\dots,i+\frac{p}{2}\},l\in \{j-\frac{p'}{2},\dots,j+\frac{p'}{2}\}\big\},
    \]
    where in case the window leaves the matrix $A$, all non-defined $A_{k,l}$ are set to $-\infty$.
    
     In example \ref{ex:convNN} we see the function in more detail. 
    \item (convolutional layer) Let $A=(A_{i,j})\in\R^{d_i,d_i'}$. For stride lengths $s,s'\in\N$, a window size $(p+1)\times (p'+1)$, where $p,p'$ both are even, a weight matrix $W\in\R^{p+1,p'+1}$ an bias $b\in\R$ the output $\Phi^{(i)}(A)=\widetilde A\in\R^{d_i-2s,d_i'-2s'}$ is given by the following rule for all $i\in\{s+1,\dots,d_i-s\}$ and $j\in\{s'+1,\dots,d_i'-s'\}$
    \begin{align*}
    \widetilde A_{i-s,j-s'}=b+\mathrm{sum}\big\{A_{k,l}\cdot W_{k-(i-\frac{p}{2})+1,l-(j-\frac{p'}{2})+1}\;|&\;k\in\{i-\frac{p}{2},\dots,i+\frac{p}{2}\},\\
    &l\in \{j-\frac{p'}{2},\dots,j+\frac{p'}{2}\}\big\},
    \end{align*}
    where in case the window leaves the matrix $A$, all non-defined $A_{k,l}$ are set to $0$.
    
     As above, we see the convolution layer in more detail in the following example \ref{ex:convNN}. 
     
     We extend the definition by considering not one, but multiple weight matrices within one convolutional layer. The output is then in the space $\R^c\times\R^{d_i-2s,d_i'-2s'}$, where $c$ is the amount of weight matrices used, called the number of channels. Each weight matrix gets its own dimension. If the next layer gets such an 3-dimensional array as input, the activation layer remains to be defined component wise, the max pooling layer uses the 2 dimensional window on all channels and for a convolutional layer with number of channels $\tilde c$, the convolution of a weight matrix is done over all input channels and then summed up.

\end{itemize}
\end{defn}
\begin{example}\label{ex:convNN}
    Consider the matrix 
    \[
    A=\left(\begin{array}{ccccc}
         1&2&3&4&5  \\
         6&7&8&9&10  \\
         11&12&13&14&15 \\
         6&7&8&9&10  \\
         1&2&3&4&5  
    \end{array}\right).
    \]
    For stride lengths $s,s'=1$, window size $(3,3)$, the max pooling layer returns the following matrix
    \[
    \widetilde A = \left(\begin{array}{ccc}
         13&14&15  \\
         13&14&15 \\
         13&14&15 
         \end{array}\right).
    \]
    For the same stride lengths and window size, we get with the weight matrix
    \[
    W=\left(\begin{array}{ccc}
         1&1&1  \\
         0&0&0 \\
         -1&-1&-1 
         \end{array}\right)
    \]
    and bias $b=0$, the following output after a convolution layer with weights $W$
    \[
    \widetilde A=\left(\begin{array}{ccc}
         -30&-30&-30  \\
         0&0&0 \\
         30&30&30 
         \end{array}\right).
    \]
\end{example}
\begin{remark}
    The weights in a convolutional neural network are trainable parameters. The number and composition of activation/max-pooling and convolutional layers are hyperparameters of a problem which need to be chosen.
\end{remark}

\clearpage
\section{Partially Observable Markov Decision Processes}
In our application, the data which we are able to access is only a part of the underlying Markov Decision Process. Instead of the information on all electrons, we are only able to access accumulations (histogram plots) via sensors. Therefore, we can not observe the whole state space $E$, but an observable part. In this section we introduce partially observable Markov Decision Processes and connect them to Markov Decision Processes. The section is based on the book \textit{'Markov decision processes with applications to finance'} \cite{mdp_in_finance} by Nicole Bäuerle and Ulrich Rieder. 
\begin{defn}[Partially Observable Markov Decision Model]
    We call a tuple $(E_X\times E_Y,A,p,p_0,r,\beta,g)$ a partially observable Markov Decision Model, where
    \begin{itemize}
        \item $E_X\times E_Y$ is called state space. We assume that $E_X\subseteq \R^m$ and $E_Y\subseteq \R^n$, for some $m,n \in\N$ equipped with the Borel $\sigma$-algebras $\mathcal{E}_X,\mathcal{E}_Y$. For $(x,y)\in E_X\times E_Y$ we call $x$ the observable component of the state and $y$ the part that cannot be observed
        \item $A$ is a measurable space with $\sigma$-algebra $\mathcal{A}$, called state space
        \item $p$ is a stochastic transition kernel from $E_X\times E_Y\times A$ to $E_X\times E_Y$, i.e. for all $(x,y,a)\in E_X\times E_Y\times A$ and Borel sets $B\in \mathcal{E}_X\otimes\mathcal{E}_Y$ the function $(x,y,a)\mapsto p(B|x,y,a)$ is a probability measure on $\mathcal{E}_X\otimes\mathcal{E}_Y$ and the mapping $(x,y,a)\mapsto p(B|x,y,a)$ is measurable. The value $p(B|x,y,a)$ is interpreted as probability of set $B$ given current state $(x,y)\in E_X\times E_Y$ and chosen action $a\in A$.
        \item $p_0$ is called initial probability on $Y$
        \item $r\colon E_X\times E_Y\times A\to\R$ measurable function called reward function
        \item discount factor $\beta \in (0,1]$
        \item $g\colon E_X\times E_Y\to\R$ measurable function, called terminal reward function
    \end{itemize}
\end{defn}
The definition of a partially observable Markov decision model is similar to the definition of a Markov decision model. The main difference is that we distinguish between an observable and an unobservable part. 

We now define policies.
In contrast to the book \cite{mdp_in_finance}, where policies can take into account the entire observable history up to that point, we define policies $\pi$ as a sequence of decision rules $\pi=(f_0,\dots,f_{N-1})$, where all $f_i\colon E_X\to\A$ are measurable. 

Similar to the Markov decision model case, for a policy $\pi$ and initial state $x\in E_X$, a probability measure $\wP_x^\pi$ on $(E_X\times E_Y)^{N+1}$, equipped with the product $\sigma$-algebra, is defined as the initial conditional distribution $p_0$ together with transition probability $p$.

For $\omega=(x_0,y_0,\dots,x_N,y_N)\in(E_X\times E_Y)^{N+1}$, the random variables $X_n$ and $Y_n$ are defined as
\[
X_n(\omega)=x_n\quad\text{and}\quad Y_n(\omega)=y_n.
\]
The partially observable Markov decision process is then defined as $(X_n,Y_n)_{n=0,\dots,N}$.
\vspace{.8cm}

Note that the algorithms shown in section \ref{sec:rl_algorihms} only work for Markov decision models. This is because they are based on the Bellman equation and temporal difference learning insights that rely on the Markov property of the environmental dynamics.  

Nevertheless, the algorithms defined in section \ref{sec:rl_algorihms} are used in partially observable Markov decision models (see \cite{td3_for_pomdp_1}). In our application, we will apply the TD3 algorithm and check its performance.

\clearpage
\section{Hyperparameter search}\label{section:hyper}
This section is based on the paper \textit{Algorithms for Hyper-Parameter Optimization} \cite{hyper_para_paper}.

A common problem in training machine learning models is the choice of hyperparameters. For example, when training a deep neural network, the number of hidden layers and the number of nodes within each hidden layer must be carefully chosen, as these parameters determine the complexity of the model and thus lead to possible overfitting or underfitting. Another important hyperparameter that often needs to be tuned is the learning rate. The learning rate determines how much the parameters of the model change after each gradient descent step. If the learning rate is set too low, the model may get stuck in local minima, while if it is set too high, it may keep jumping over the global minima during gradient descent steps. These examples show that the hyperparameters play a crucial role in the training process and the performance of the trained model. 

Reinforcement learning algorithms tend to be very sensitive to the choice of parameters used. They also tend to have many hyperparameters. In this section, we show how to automatise the search for good hyperparameters.

In the paper \cite{hyper_para_paper}, on which this section is based, the authors use tree-structured hyperparameter spaces $\mathcal{X}$. In the sense that some leaf variables (such as the number of nodes in the 2nd layer of a deep neural network) are only well defined if node variables (such as the number of hidden layers of a deep neural network) take certain values. 

In our case, we assume that the space of all valid hyperparameters $\mathcal{X}$ can be modelled as a bounded subset of $(\R^m\times \N^n)$, for some $m,n\in\N$. 

The problem is to find hyperparameters $x^*\in\mathcal{X}$, that minimise a loss
\[
x^*\in\argmin_{x\in\mathcal{X}}\E[f(x)],
\]
where $f(x)$ is a random variable on $(\R,\mathcal{B}(\R))$. Typically, the value $f(x)$ describes a loss on a training set, after training a model with hyperparameters $x$. As the training process of models is typically random, the performance on the test set is also random.
\vspace{.3cm}

Examples of real-valued hyperparameters are the learning rate or the variance of a noise process. On the other hand, the number of hidden layers and the number of nodes in the first and last layers of a deep neural network can be given as integers. The ideal number of nodes in the $k$-th layer of a deep neural network can be found by predefining a list of possible neural network shapes and letting the algorithm find the best index from the list. Implementing this concept, we do not need to use the tree-hyperparameter space.

In this bounded case, the simplest and most commonly used methods are the \emph{grid search} and the \emph{random search}.

For each hyperparameter, the grid search method discretises the set of  parameters into a finite set. Then, for each combination of different hyperparameter values the performance of the model is checked (typically on a test set) and the best performing hyperparameters are used. This method is great for finding combinations of hyperparameters that perform well. The problem is that for a good performance the mesh-size of the grid needs to be small, which is computationally expensive. For 10 hyperparameters with each only 5 different values in the grid search, a total of $5^{10}=9765625$ independent training and evaluation cycles must be performed.

The other commonly used method is \emph{random search}. Here, uniformly (or log-uniformly) hyperparameters are sampled from the bounded set of all valid hyperparameters, evaluated and the performing hyperparameters are then used.

In \cite{random_vs_grid_paper}, it was shown that random search, typically, finds hyperparameters as good or better than grid search, within a small fraction of the computation time. 

Note that it is possible to find good hyperparameters by hand. But this process usually takes a lot of time and when many hyperparameters are used, it is difficult to keep track of which parameters are causing good or bad training performances.

In the following, we formalise the hyperparameter search process,  generalise the set of valid hyperparameters and show algorithms that are able to take into account the history of the hyperparameter search, thus minimising the time spent evaluating bad hyperparameters.


Now, we consider two different Sequential Model-Based Global Optimization (SMBO) algorithms. These algorithms approximate a costly to evaluate function $f\colon\mathcal{X}\to\R$. The following algorithm shows the pseudocode of SMBO algorithms:

\begin{algorithm}[H]
    \caption{SMBO algorithm}\label{alg:SMBO}
    \KwData{function $f\colon \mathcal{X}\to\R$, Model of function $f$ called $M_0$, number of iterations $T\in\N$, loss function $S$}
    \KwResult{Observation history $\mathcal{H}=((x_i,f(x_i)))_{i=1}^T$}
    $\mathcal{H}\leftarrow\emptyset$\;
    \For{$i=1,\dots,T$}{
    $x_i\leftarrow\argmin_x S(x,M_{i-1})$\;
    Evaluate $f(x_i)$\;
    $\mathcal{H}\leftarrow \mathcal{H}\cup(x_i,f(x_i))$\;
    Fit new model $M_i$ to $\mathcal{H}$\;
    }
\end{algorithm}

SMBO algorithms differ in the loss function $S$ used and on how a new model is fitted to the observation history $\mathcal{H}$.

For the hyperparameter-search, the function $f\colon\mathcal{X}\to\R$ is random, i.e. there exists a measurable space $(\Omega,\mathcal{F})$ with $f\colon\mathcal{X}\times\Omega\to\R$ and $f(x,\cdot)$ is measurable for all $x\in\mathcal{X}$. To evaluate the function $f$, a model with parameters in $\mathcal{X}$ is trained for a fixed amount of time or iterations and then $f(x)$ is given as the performance of the model on a test set. It is important to note that in this section the performance function $f$ becomes smaller the better the performance. Instead of maximising some reward, we want to minimise some loss.

The following algorithms will optimise the Expected Improvement (EI) criterion. For a model $M$ of $f$, the expected improvement is the expected value of how much $f$ negatively exceeds some threshold $y^*\in\R$, i.e.
\[
\mathrm{EI}_{y^*}(x)\coloneqq\int_{-\infty}^\infty\max(y^*-y,0)p_M(y|x)\mathrm{d}y.
\]
\subsection{Gaussian Processes Approach}
In this subsection, we first briefly explain Gaussian processes and then show how they can be used for the hyperparameter search. The theoretical explanation of Gaussian processes is based on \cite{gaussian_processes_book}.

We start with the assumption that $f$ has the domain $\R^m$ for a $m\in\N$ and $f$ is linear, i.e. a $w\in\R^m$ exists with $f(x)=x^Tw$. Further, we assume that the observations of $f$, which we denote by $y$, are given by $y=f(x)+\varepsilon$, where epsilon is normally distributed with mean 0.

Let $\mathcal{D}$ be a set of $n$ observations, i.e. $\mathcal{D}=\{(x_i,y_i)|i=1,\dots,n\}\subseteq \R^m\times\R$. 
We denote the matrix where the $i$-th column equals $x_i$ as $X\in\R^{m\times n}$ and $y\coloneqq(y_1,\dots,y_n)$. Then, we can write $\mathcal{D}=(X,y)$.

Now, assume that the noise $\epsilon$ depends on the size of $\mathcal{D}$, i.e. $\epsilon \sim \mathcal{N}(0,\sigma_n^2)$.

The likelihood (probability density) of observing $y$ given $X$ and $w$ is 
\begin{align*}
    p(y|X,w)&=\prod_{i=1}^np(y_i|x_i,w)=\prod_{i=1}^n\frac{1}{\sqrt{2\pi}\sigma_n}\exp\left(-\frac{(y_i-x_i^Tw)^2}{2\sigma^2_n}\right)\\
    &=\frac{1}{(2\pi\sigma_n^2)^{n/2}}\exp\left(-\frac{1}{2\sigma^2_n}||y-X^Tw||_2^2 \right),
\end{align*}
which is the probability density of $\mathcal{N}(X^Tw,\sigma_n^2I)$.

Gaussian processes use priors over the space of parameters. We assume that 
\[
w\sim\mathcal{N}(\mathbf{0},\Sigma_p),
\]
where $\Sigma_p$ is a covariance matrix. By Bayes' rule
\[
p(w|y,X)=\frac{p(y|X,w)p(w)}{p(y|X)},
\]
where $p(y|X)$ is called the marginal likelihood and is given by
\[
p(y|X)=\int p(y|X,w)p(w)\mathrm{d}w.
\]
The marginal likelihood $p(y|X)$ is independent of the parameter $w$. We obtain
\begin{align*}
    p(w|X,y)&\propto \exp\left( -\frac{1}{2\sigma_n^2}(y-X^Tw)^T(y-X^Tw)
    \right)\exp\left(-\frac{1}{2}w^T\Sigma_p^{-1}w\right)\\
    &\propto \exp\left( -\frac{1}{2}(w-\Bar{w})^T(\frac{1}{\sigma_n^2}XX^T+\Sigma_p^{-1})(w-\Bar{w}) \right),
\end{align*}
where $\Bar{w}\coloneqq\sigma_n^{-2}(\sigma_n^{-2}XX^T+\Sigma_p^{-1})^{-1}Xy$. With this, we observe that $p(w|X,y)$ is the probability density of a normal distribution with mean $\Bar{w}$ and covariance matrix $A^{-1}$, where $A\coloneqq \sigma_n^{-2}XX^T+\Sigma_p^{-1}$. Note that $\Bar{w}=\frac{1}{\sigma_n^2}A^{-1}Xy$.

We analyse the predictive distribution for $f_*\triangleq f(x_*)$ at $x_*$. Using the formula
\[
p(f_*|x_*,X,y)=\int p(f_*|x_*,w)p(w|X,y)\mathrm{d}w,
\]
it is shown that $p(f_*|x_*,X,y)$ is equal to the probability density of $\mathcal{N}(\frac{1}{\sigma_n^2}x_*^TA^{-1}Xy,x_*^TA^{-1}x_*)$. Thus, the observation $y$ of $f(x_*)$ with the most likelihood equals $\frac{1}{\sigma_n^2}x_*^TA^{-1}Xy$. Additionally, with the variance $x_*^TA^{-1}x_*$ shows the confidence in this prediction, given the data $(X,y)$.  
\vspace{.5cm}

Using a feature map or kernel (kernel trick), briefly explained in \cite{BA_schuett}, we extend this theory from a linear problem to non-linear problems. For a feature map $\phi\colon\R^m\to\R^M$, we set the model of $f$ to $f(x)=\phi(x)^Tw$. The predictive distribution $p(f_*|x_*,X,y)$ has the density of the following normal distribution
\begin{align}\label{eq:normal_feature}
\mathcal{N}\left(\frac{1}{\sigma_n^2}\phi(x_*)^TA^{-1}\Phi y,\phi(x_*)^TA^{-1}\phi(x_*)\right),
\end{align}
where $\Phi=\Phi(X)$ and $\Phi(X)$ is the aggregation of columns $\phi(x_i)$ and $A\coloneqq \sigma_n^{-2}\Phi\Phi^T+\Sigma_p^{-1}$.

\begin{example}
    The following function is a commonly used feature map:
    
    For some $n\in\N$, one can consider the function $\phi_n\colon\R\to\R^{n+1}$, where
        \[
        \phi_n(x)=(1,x,x^2,\dots,x^n)^T
        \]
    Given this feature map, we are now able to model the functions $f$ as 
    \[
    f(x)=\phi_n(x)^Tw=w_0+w_1x+w_2x^2+\dots+w_nx^n,
    \]
    for a $w=(w_0,\dots,w_n)\in\R^{n+1}$. So with this feature map at hand, we are able to model polynomial functions.
    
\end{example}
    
A Gaussian process is formally defined as:
\begin{defn}[Gaussian process]
    A Gaussian process is a collection of random variables, such that every finite combination of the random variables has a joint Gaussian distribution.
\end{defn}

\begin{remark}
    A Gaussian process is characterised by its mean and covariance functions. For a real-valued function $f(x)$, the mean function $m(x)$ and the covariance function $k(x,x')$ are given by
    \begin{align*}
    m(x)&=\E[f(x)],\\
    k(x,x')&=\E[(f(x)-m(x))(f(x')-m(x'))].
    \end{align*}
    In this case, we write for the function $f$:
    \[
    f(x)\sim\mathcal{GP}(m(x),k(x,x')).
    \]
\end{remark}
Now, we show how Gaussian processes are used to maximise $\mathrm{EI}_{y^*}$, which is defined as
\[
\mathrm{EI}_{y^*}(x)\coloneqq\int_{-\infty}^\infty\max(y^*-y,0)p_M(y|x)\mathrm{d}y,
\]
for some threshold $y^*\in\R$ and $M$ a model of $f\colon\mathcal{X}\to\R$. First, set $y^*$ to the best (lowest) value in the history $\mathcal{H}$, i.e. $y^*\coloneqq\min\{f(x_i),1\leq i\leq n\}$. Let the model $p_M$ in the function $\mathrm{EI}_{y^*}$  be the posterior Gaussian processes after observing $\mathcal{H}$. Then, the expected improvement $\mathrm{EI}_{y^*}$ gives us for all $x\in\mathcal{X}$ the expected value over all $y$ that satisfy $y\leq {y^*}$. Since a Gaussian process posterior is used, the $\mathrm{EI}$ function captures uncertainty and thus models the difference between explored and under-explored regions. How to chose the next to be observed $x_i$ can be read in detail in \cite{hyper_para_paper}. In short, one uses gradient-free evolutionary algorithms for the optimization. 

In case the configuration space $\mathcal{X}$ is tree-structured, multiple Gaussian processes over the space $\mathcal{X}$ are needed, since in the space $\mathcal{X}$ not all variables are always well defined: For example, the number of nodes in the 3rd layer of a 2-layer deep neural network. In this case, individual Gaussian processes are used for these conditional hyperparameters.


\subsection{Tree-structured Parzen Estimator Approach}\label{sec:tree_struc}
The Tree-structured Parzen Estimator Approach (TPE) is another SMBO algorithm that is used in the application of this thesis. It works well for tasks with many hyperparameters and unlike the Gaussian process approach, it does not model $p(y|x)$ directly, but instead models $p(x|y)$ and $p(y)$ and then obtains $p(y|x)$ with Bayes' rule. In this subsection, we give a brief overview of the approach, for a more detailed explanation see \cite{hyper_para_paper}.

The algorithm is specifically designed for tree-structured configuration spaces and deals with them in the following way: The generative process of the configuration space (for example first choosing the number of layers then the number of nodes within them) is replaced by distributions of configuration priors with nonparametric densities. 

In our case, where the hyperparameter space $\mathcal{X}$ is given as $(\R^m\times \N^n)$, for some $m,n\in\N$ we define the nonparametric densities directly on the configuration priors.

The probabilities $p(x|y)$ for $x\in\mathcal{X}$ and $y\in\R$ are then defined using densities $l,g$ as 
\[
p(x|y)=\left\{\begin{array}{lc}
     l(x)&\text{if $y<y^*$}  \\
     g(x)&\text{if $y\geq y^*$} 
\end{array}\right.,
\]
where $l,g$ are non-parametric densities. After observing $\{x^{(1)},\dots,x^{(k)}\}$ the density $l(x)$ is the non-parametric density using the observations $\{x^{(i)}\}$, such that for the loss $f(x^{(i)})$ holds $f(x^{(i)})<y^*$. The density $g(x)$ uses the other observations. Unlike the Gaussian process algorithm, where $y^*$ is chosen as the smallest loss, the TPE algorithm chooses a $y^*$ that is some quantile $0<\gamma<1$ of all observations $y_i$, i.e. $p(y<y^*)=\gamma$, which can be approximated without  any model for $p(y)$. Using such a quantile guarantees that both, $l$ and $g$ will get observations $x^{(i)}$.

We address now the question which $x$ to choose after observing the history $\mathcal{H}$. Note that for all $x\in\mathcal{X}$, we have
\[
\mathrm{EI}_{y^*}(x)=\int_{-\infty}^{y^*}(y^*-y)p(y|x)\mathrm{d}y=\int_{-\infty}^{y^*}(y^*-y)\frac{p(x|y)p(y)}{p(x)}\mathrm{d}y.
\]
With the quantile $\gamma=p(y<y^*)$, we observe that 
\[
p(x)=\int_{\R} p(x|y)p(y)\mathrm{d}y=\gamma l(x)+(1-\gamma) g(x).
\]
We continue with
\begin{align*}
\int_{-\infty}^{y^*}(y^*-y)p(x|y)p(y)\mathrm{d}y&=\int_{-\infty}^{y^*}(y^*-y)l(x)p(y)\mathrm{d}y\\
&=l(x)\left(\int_{- \infty}^{y^*}y^*p(y)\mathrm{d}y-\int_{- \infty}^{y^*}yp(y)\mathrm{d}y\right)\\
&=l(x)y^*\gamma-l(x)\int_{-\infty}^{y^*}yp(y)\mathrm{d}y. 
\end{align*}
Combined we achieve
\begin{align*}
    \mathrm{EI}_{y^*}(x)&=\frac{1}{p(x)}\int_{-\infty}^{y^*}(y^*-y)p(x|y)p(y)\mathrm{d}y\\
    &=\frac{ l(x)y^*\gamma-l(x)\int_{-\infty}^{y^*}yp(y)\mathrm{d}y  }{\gamma l(x)+(1-\gamma) g(x)}\\
    &=\frac{ l(x)(y^*\gamma-\int_{-\infty}^{y^*}yp(y)\mathrm{d}y)}{\gamma l(x)+(1-\gamma) g(x)}\\
    &\propto\frac{ l(x)}{\gamma l(x)+(1-\gamma) g(x)}\\
    &=\left(\gamma+\frac{g(x)}{l(x)}(1-\gamma)\right)^{-1}.
\end{align*}
To maximise the expected improvement, we need to find a $x\in\mathcal{X}$ that has high probability under $l(x)$ and low probability under $g(x)$. In each iteration, a large sample of candidates is drawn from the distribution $l$ and then evaluated by $g(x)/l(x)$. The candidate $x$ with the lowest $g(x)/l(x)$ score is then picked to evaluate $f(x)$.

\section{Summary}
In this chapter, we have seen the foundations of reinforcement learning.

After defining a Markov Decision Model/Process and a policy, we were able to describe a Markov Decision Process which depends on a policy $\pi$, by changing the probability measure. Then, we presented the Bellman equation and showed that under the structure assumptions \ref{def:structure_assumption} an optimal policy can be calculated, see algorithm \ref{alg:backward_induction}. As this algorithm only works great for finite action and state spaces, we introduced the algorithms DDPG and TD3 in section \ref{sec:rl_algorihms}. As the performance of the algorithms depends on the chosen hyperparameters, we introduced in section \ref{section:hyper} the automatisation of the search for good hyperparameters.

\chapter{Physical Background}\label{chap:physics}
This chapter introduces the physics behind the synchrotron light source BESSY II. We start with a brief introduction to the basic functionality of BESSY II. Then, we introduce the functionality of the non-linear kicker.

\section{BESSY II}
This section is based on my Bachelor's thesis \cite{BA_schuett}, where minor changes have been made. It gives an intuitive idea of how BESSY II works.
\vspace{.5cm}

BESSY II is a synchrotron light source located in Adlershof, Berlin. Inside, electrons are first accelerated in a linear accelerator. They are then accelerated to almost the speed of light in a synchrotron. Such a synchrotron consists of multiple deflection magnets that keep the electrons on a circular path, with linear acceleration sections in between to speed them up. 

Once the electrons have reached a certain energy level, they are injected into the storage ring. The aim of this thesis is to optimise this injection. The storage ring consists of more deflection magnets and undulators. Undulators are special magnetic arrangements that force the electrons to follow a a slalom course. In this slalom course, the electrons experience acceleration and thus emit energy in form of electromagnetic radiation, called synchrotron radiation. For further explanation, we refer to a paper by \textit{Schwinger \cite{synchrotron_schwinger}}. This synchrotron radiation is then used in a variety of experiments, for example in materials science.

The following image shows all the parts together:
\begin{figure}[H]
    \centering
    \includegraphics[width=.96\linewidth]{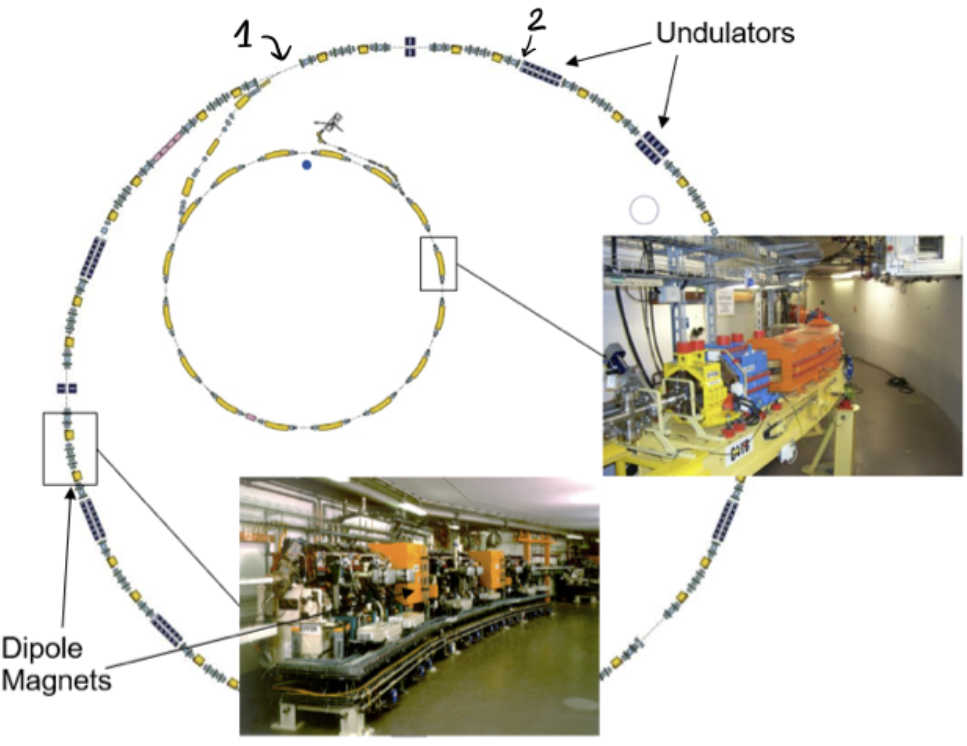}
    \caption{Basic setup of BESSY II, provided by the Helmholtz Zentrum Berlin. The inner ring shows a synchrotron, the outer ring is called storage ring. The small straight section above the synchrotron is the electron gun, where electrons are created. At the mark 1, the septum-sheet is located; at mark 2 the non-linear kicker. The role of the septum-sheet is explained in subsection \ref{subsec:without_kick}.}
    \label{fig:bessy2}
\end{figure}
This image shows both rings. The inner ring is the synchrotron, where the electrons are accelerated and the outer ring is the storage ring. The photo on the right shows some magnets in the synchrotron, while the photo on the left shows a section of the storage ring. Furthermore, inside the outer ring, above the synchrotron, we can see the electron gun, that creates and injects the electrons. The injector pre-accelerates the electrons and then injects them into the synchrotron. 

In the following, we take a closer look at the different parts.
\subsection{Deflection magnets}
To force the electrons into the circular track, BESSY II uses deflection magnets. 
Using the principle of the Lorentz force, these magnets bend the path of the electrons.

The Lorentz force law states that the electromagnetic force $\vec{F}_L$ on a point charge of charge $q$ at a given point and time is a function of velocity $\vec v$ and magnetic field strength $\vec{B}$ as
\[
\vec{F}_L=q (\vec{v} \times\vec{B}).
\]
Visually, the Lorentz force law for an electron can be understood as:
\begin{figure}[H]
    \centering
    \includegraphics[width=0.4\linewidth]{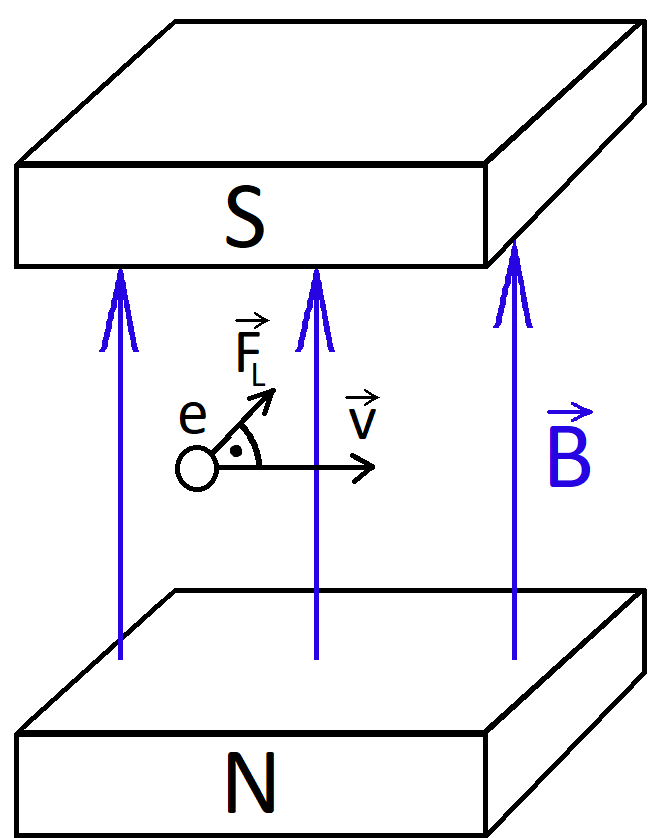}
    \caption{Visualisation of the Lorentz force law.}
\end{figure}
Next, we note that the faster the electrons get, the stronger the magnetic fields have to be. This is achieved by using electromagnets whose magnetic field strength can be controlled by the intensity of the electric current. 
\subsection{Linear acceleration sections}
In this subsection, we analyse the linear acceleration sections. These are implemented to accelerate the electrons. This is achieved by applying an alternating voltage to tabular electrodes that are aligned in a straight line. When an electron approaches an electrode, the electrode becomes positive to attract the electron. Ideally, as soon as the electron passes through half of the electrode, the electrode gets negative and repels the electron. With this technique we can accelerate the electron by first pulling it towards the electrode and then pushing it away.

This procedure is visualised in the following 2 figures.

First, we attract the electron with the positive electrode:
\begin{figure}[H]
    \centering
    \includegraphics[width=0.7\linewidth]{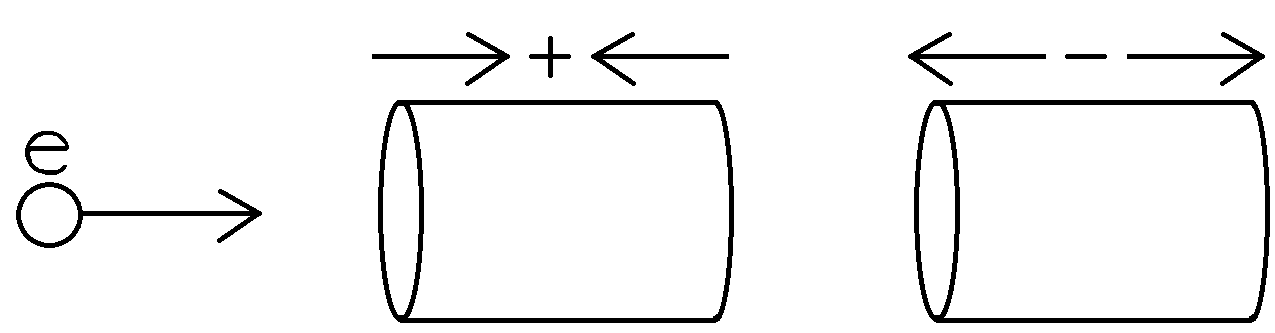}
    \caption{Visualisation of attraction of an electron.}
\end{figure} 
If the electron were to pass through half of the electrode and we wouldn't change its electrical charge, we would decelerate the electron, which we don't want.

That is, why we swap the electrical charges of the electrodes:
\begin{figure}[H]
    \centering
    \includegraphics[width=0.5\linewidth]{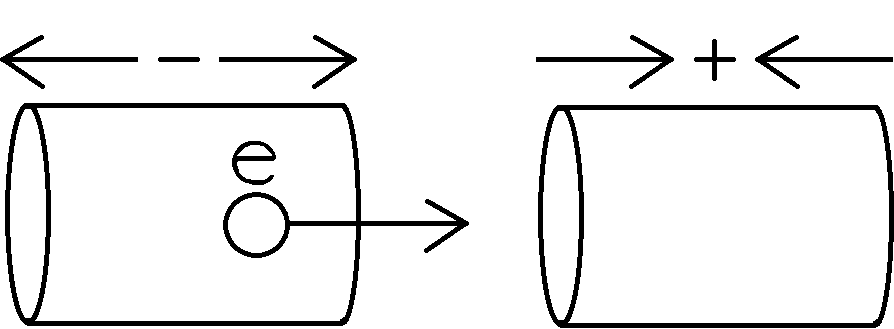}
    \caption{Visualisation of how the electrode changes its electrical charge and pushes the electron.}
\end{figure} 
Now, the first electrode pushes the electron and the next electrode pulls the electron, accelerating it. Note that there is a gap between the two electrodes, so the electron has space to accelerate. 

In practice, there are even more electrodes set up after another, to repeat this procedure and accelerate the electron multiple times.

For this technique it is essential to calculate the exact frequencies at which the electrodes must change their electrical charge, otherwise no acceleration takes place and the electrons might slow down. Further, note that these frequencies need to be adapted to the current speed of the electrons, since in the synchrotron setting the electrons go through the same linear acceleration setting  several times.

\clearpage
\section{Non-Linear Kicker}
In this section, we look at the physical background behind the non-linear kicker. The figure below shows its structure:
\begin{figure}[H]
    \centering
    \includegraphics[width=9cm]{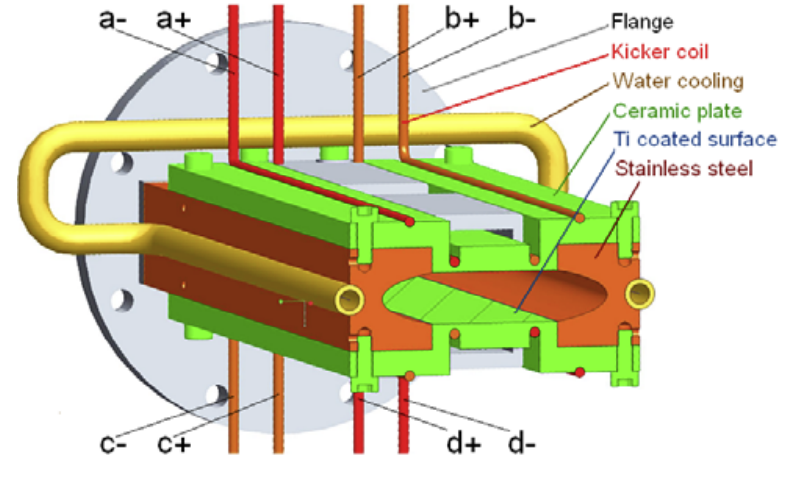}
    \caption{Structure of the Non-Linear Kicker, by \cite{nlk_paper}}
    \label{fig:nlk_structure}
\end{figure}
In figure \ref{fig:nlk_structure}, we can see the schematic structure of the non-linear kicker. At the front of the figure, the oval shape shows the cross-section of the beam-line where the stored and newly injected electrons are located. We  focus on the two main components: the 4 kicker coils, which we call wires and the titanium coating on the ceramic surface.

First, we see how a magnetic field is created around a single wire. Then, we analyse eddy-currents, a side-effect of the  changing magnetic fields around a single wire, on the titanium coating and in the last two subsections we combine both analyses and show how to model the non-linear kicker.  
\subsection{Field around a wire}
This subsection is based on the book \textit{Electromagnetic fields and energy}, by H. A. Haus and J. R. Melcher \cite{physics_book}.

Now, we analyse the magnetic fields created around the wires. The special arrangement of the wires, combined with the titanium coating,  creates the non-linear magnetic field. 

Let $\mathcal{I}=[0,T]\subset\R$ be a time interval, for a $T>0$ and $S$ be the cross-sectional area of the wire.

Ampere's law shows the relationship between the magnetic field density $\mathbf{H}\colon\R^3\times \mathcal{I}\to\R^3$ and current density $\mathbf{J}\colon\R^3\times \mathcal{I}\to\R^3$:
\begin{align}\label{physik_1}
    \oint_{C}\mathbf{H}\mathrm{d}\mathbf{s}= \int_S \mathbf{J}\mathrm{d}\mathbf{a}+\frac{\mathrm{d}}{\mathrm{d}t}\int_S \epsilon_0\mathbf{E} \mathrm{d}\mathbf{a},
\end{align}
where $C\subset\R^3$ is the contour of a surface $S\subset\R^3$, $\mathbf{E}\colon\R^3\times\mathcal{I}\to\R^3$ is the electric field, $\epsilon_0\in\R$ is the vacuum permittivity and we fix any time in $\mathcal{I}$. In the following, we assume that the second term on the right hand side describing the field around a changing electric field $\mathbf{E}$ is negligible, following the magnetoquasistatic approximation (MQS).  

We model the wires as lines, so let the surface area $A$ of $S$ go to $0$ and the field density to infinity. We get as current $I\colon\mathcal{I}\to\R$ 
\begin{align}\label{physik_2}
    I=\lim_{\substack{J\to\infty\\ A\to0}}\int_S\mathbf{J}\mathrm{d}\mathbf{a}.
\end{align}
Combining equations \ref{physik_1} and \ref{physik_2}
 we find that  the field around a wire is given by
 \begin{align}
     \oint_C \mathbf{H}\mathrm{d}s=I,
 \end{align}
using the MQS approximation. If $C$ is a circle with radius $r$, we can reformulate the integral on the left hand side to
\begin{align}
I=\oint_C\mathbf{H}\mathrm{d}s=\int_0^{2\pi}H_{\varphi}r\mathrm{d}\varphi.
\end{align}
Since the magnetic field density $H_\varphi$ does not depend on the angle $\varphi$, we find that
\begin{align}
    I=\int_0^{2\pi}H_\varphi r\mathrm{d}\varphi=H_{\varphi_0}2\pi r
\end{align}
for any $\varphi_0\in[0,2\pi]$.

Using the linear constitutive law, the  magnetic flux density $\mathbf{B}\colon\R^3\times\mathcal{I}\to\R^3$, can be written as $\mathbf{B}=\mu_0 \mathbf{H}$, with $\mu_0\in\R$ as magnetic vacuum permeability. This gives us
\begin{align}
    I=H_{\varphi_0}2\pi r=\frac{1}{\mu_0}B_{\varphi_0}2\pi r,
\end{align}
which is equivalent to
\begin{align}\label{physik_6}
    B_{\varphi_0}=\mu_0\frac{I}{2\pi r}l,
\end{align}
for all $\varphi_0\in[0,2\pi]$.

This equation allows us to describe the magnetic flux density around a single wire, which we need later.


\subsection{Modelling eddy current effects}
This subsection is based on the book \textit{Electromagnetic fields and energy}, by H. A. Haus and J. R. Melcher \cite{physics_book}.

In the following, we set $S$ as the surface of the titanium coating.

When a surface $S$ is exposed to a changing magnetic field, induced for example by a changing line current, the motion of the electrons on the surface creates a new current, called an eddy current, which itself induces a magnetic field. 

For the non-linear kicker, the changing magnetic field is caused by the changing wire current.      

We note that a magnetic field $\mathbf{B}$ induces an electric field $\mathbf{E}$. Their relation is given by
\begin{align}\label{physik_3}
    \oint_C\mathbf{E}\mathrm{d}\mathbf{s}=\frac{\mathrm{d}}{\mathrm{d}t}\int_S\mathbf{B}\mathrm{d}\mathbf{a}. 
\end{align}
We now focus on the currents induced in the electrically conductive sheet. According to Ohm's "law", the current induced in the sheet is given by
\begin{align}
    \mathbf{J}=\sigma \mathbf{E}
\end{align}
with the conductance $\sigma=\frac{1}{\rho}$ of the material and $\rho$ equal to its specific resistance.

Using Ampere's law, see \ref{physik_1}, we can write the surface current $K$ as the fields tangential to the surface, i.e.
\begin{align}
    -H_t^a+H_t^b=K, 
\end{align}
where $H_t^a$ is the magnetic field density generated by the wire and 
$H_t^b$ by the eddy current.

Since $K$ is the surface current and $J$ the current density, we get
\begin{align}\label{physik_4}
    K\equiv \Delta_m J=\Delta_m\sigma E,\quad\text{thus}\quad E=\frac{K}{\Delta_m\sigma},
\end{align}
where $\Delta_m$ is the thickness of the material.

Combining equations \ref{physik_3} and \ref{physik_4} we get
\begin{align}
    K\oint_C\frac{\mathrm{d}s}{\Delta_m(s)\sigma(s)}=\frac{\mathrm{d}}{\mathrm{d}t}\int_S\mathbf{B}\mathrm{d}\mathbf{a}.
\end{align}
Since the thickness $\Delta_m(s)$ and the conductivity $\sigma(s)$ are uniform we get
\begin{align}\label{physik_5}
    \frac{KP}{\Delta_m\sigma}=\frac{\mathrm{d}}{\mathrm{d}t}\int_S\mathbf{B}\mathrm{d}\mathbf{a},
\end{align}
where $P\in\R$ is the peripheral length of $C$.

We assume that the material is perfectly conducting, so $\sigma=\infty$ holds, thus the left side of the above equation is $0$.

We thus get
\begin{align}\label{physik_7}
0=\frac{\mathrm{d}}{\mathrm{d}t}\int_S\mathbf{B}\mathrm{d}\mathbf{a} =\int_S\frac{\mathrm{d}}{\mathrm{d}t}\mathbf{B}\mathrm{d}\mathbf{a},    
\end{align}
where the second equation follows the Leibniz integral rule.

Since equation \ref{physik_7} holds for any surface $S$,
the integrand
\begin{align}
    \frac{\mathrm{d}}{\mathrm{d}t}\mathbf{B}\mathrm{d}\mathbf{a}\doteq 0.
\end{align}  
Therefore, at the surface boundary any time varying magnetic inductance $B$ perpendicular to the surface must be 0. 

The line current of a perfectly conducting sheet can be modelled by placing a "mirror line current" opposite to the line current (similar to the construction of the light path of a mirror in geometrical optics).

The resulting magnetic field is then perpendicular to the surface vector. 


Then, we calculate the resulting field, which is shown in figure \ref{fig:plot_FEM_vs_approx} in the following subsection.

The following figure shows the position of all the wires and where the "mirror line currents" have been placed.
\begin{figure}[H]
    \centering
    \includegraphics[width=10cm]{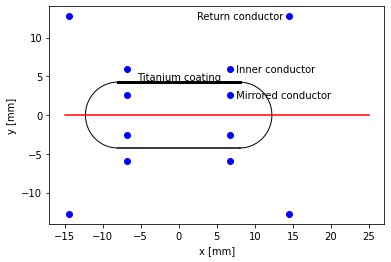}
    \caption{Position of the wires and "mirror line currents" in the storage ring cross-section.}
    \label{fig:nlk_positions}
\end{figure}
In the figure \ref{fig:nlk_positions}, we can see the positions of the wires and the "mirror line currents" in the cross-section of the storage ring. Note that the titanium coating is both on the top and bottom of the storage ring. The "mirror line currents" have been placed in the positions shown, to model the magnetic field being orthogonal to the titanium surface. In the real non-linear kicker there are no "mirror line currents".

\subsection{Modelling the Non-Linear Kicker}
We model the non-linear kicker field as a 2-dimensional field, since the effects on the beam are modelled as an infinitely thin kick. The kicker consists of 4 wires which are activated by a short pulse. A thin titanium coating is placed on the surface.

This arrangement creates a magnetic field, which we model in the following way:
\begin{itemize}
    \item The field generated by the wires is modelled using equation \ref{physik_6}.
    \item The effect of the thin titanium coating is modelled in the following manner:
\begin{itemize}
    \item The magnetic field along the symmetry plane was calculated in \cite{nlk_paper} using an FEM model. Further, the magnetic field was measured on the real device.
    \item The position of the mirror image currents is derived by mirroring the current close to the layer, as one would do for a perfectly conducting layer.   
    \item But because the titanium layer is 10$\mu$m thick, the assumption of perfect conductivity does not hold, so it expels only a fraction of the field
    \item We model this by adjusting the mirror current fraction to minimise the deviation from the field calculated by the FEM methods.
\end{itemize}
\end{itemize}

The following figure shows the magnetic field measured by the FEM model and the approximation we used for our studies.
\begin{figure}[H]
    \centering
    \includegraphics[width=10cm]{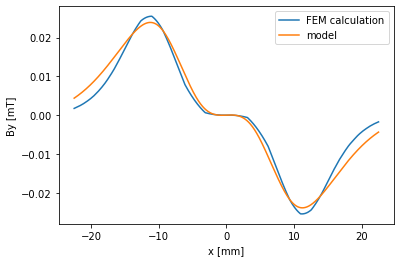}
    \caption{On the x-axis the position relative to the centre of the cross-section of the storage ring can be seen, on the y-axis the  magnetic field measured by the FEM model and the approximation we used for our studies.}
    \label{fig:plot_FEM_vs_approx}
\end{figure}
We can see that our model does not fit the FEM calculation perfectly. But for x-values between -15mm and +15mm the approximation is close enough for our further analysis.

\subsection{Thin Lens Approximation}
This subsection should give an intuitive idea of why we can model the non-linear kicker field as an infinitely thin element. 

First, note that magnetic fields deflect electric charges $q$ by 
\begin{align}
    \mathbf{F}_L=q\mathbf{v}\times\mathbf{B},
\end{align}
where $\mathbf{v}$ is the velocity and the force $\mathbf{F}_L$ is called the Lorentz-force.

Since the electron velocity in BESSY II is highly relativistic, i.e. very close to the speed of light, the electrons "see" the non-linear kicker as a rather short element, in particular due to the relativistic space contraction effect. Therefore, we model the non-linear kicker as an infinitely thin lens, and it only affect the direction of the electrons path. The non-linear kicker changes the electron impulse $p$ by 
\begin{align}
    \Delta p_x=-L\cdot B_y\quad\Delta p_y=-L\cdot B_x.
\end{align}
Properly,  beam dynamics can be treated using Hamiltonian dynamics, see for example \cite{hamiltonian_paper}. These analyses show that one must take into account the effect of the derivative of the magnetic field.

Furthermore, the path of the electrons through a device are modelled by a sequence of "drift-kicks", i.e. a sequence of drifting/flying for a finite length through a device and then adding the effect of the magnetic fields to the  direction of flight. For a short device this can be reduced to a single kick and drift cycle.

\clearpage
\section{Round-to-round behaviour}\label{sec:r2rbehav}
As described at the end of the last section, the path of electrons through a device can be modelled by a sequence of "drift-kicks", see \cite{thscsi_paper}. This can also be done for the whole storage ring of BESSY II, which consists of a large number of magnets and drift sections. 

The following figure shows the typical behaviour of newly injected electrons in the first few rounds:
\begin{figure}[H]
    \centering
    \includegraphics[width=10cm]{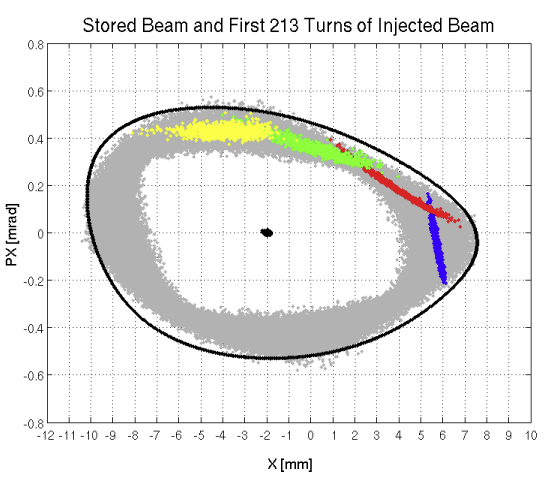}
    \caption{Example of phase space transition at the NLK for storage ring Sirius (in Brazil). Blue first round, red second, green third and yellow fourth round. The beam for next 210 rounds is shown in grey. Original plot from Lin Liu et al. \cite{Liu2016InjectionDF}. Image was mirrored, so that the injected is placed on the right side. For this storage ring the injection happens at around 6mm, in our case at 15mm}
    \label{fig:stored beam}
\end{figure}
Note that in figure \ref{fig:stored beam}, the x-axis shows the x-dimensional position of the electrons with respect to the centre of the beam line in mm. On the y-axis the direction of flight is given in mrad. As we are only kicking in x dimension, these values are sufficient to consider.

Each dot represents an electron and each coloured cloud consists of 1000 electrons. The grey dots represent the electrons in the next 210 turns. We observe that they follow a phase space transition path.

Since the path of an electron can be calculated from the sequence of "drift-kicks" we assume that there exists a function $\mathcal{T}\colon[-1,1]\times\R^2\to\R^2$, i.e. for all non-linear kicker strengths $a\in[-1,1]$ and $(x,p_x)\in\R^2$ values of an electron at an beam monitor close to the septum-sheet, which will be explained below in subsection \ref{subsec:without_kick}, the vector 
\[
(x',p_x')=\mathcal{T}(a,x,p_x)\in\R^2,
\]
gives us the x-position $x'$ and the direction of flight value $p_x'$ of the electron in the next round at the septum-sheet. For now, we assume that the round-to-round behaviour $\mathcal{T}$ is deterministic. In the model description section \ref{sec:model_description} we  add noise to this round-to-round behaviour, which can be induced by measurement errors and imperfect magnetic fields.

We also assume that functions 
\[
\mathcal{T}^\mu\colon[-1,1]\times\R^2\times\R^{2,2}\to\R^2\quad\text{and}\quad \mathcal{T}^\Sigma\colon[-1,1]\times\R^2\times\R^{2,2}\to\R^{2,2}
\]
exists, which allows the prediction of the distribution of the electrons in the next round, given the kicker strength $a$, the $x$, $p_x$ mean and covariance of the normally distributed electrons at the septum-sheet.

\section{Intuition}
In the last section we have seen how to calculate the round-to-round behaviour using "drift-kicks". In this section we present a first analysis of the injection problem using a simulation of BESSY II, that was made using the "drift-kicks". First, we analyse the number of rounds before the electrons collide with the septum sheet. Second, we investigate how to use the non-linear kicker optimally.

\subsection{Without Kick}\label{subsec:without_kick}
In this subsection, we investigate how many rounds the electrons survive without being kicked by the non-linear kicker.

Note that in BESSY II, there is a barrier in the beam line which causes the loss of electrons. The following figure \ref{fig:septum} shows a visualisation.
\begin{figure}[H]
    \centering
    \includegraphics[width=7cm]{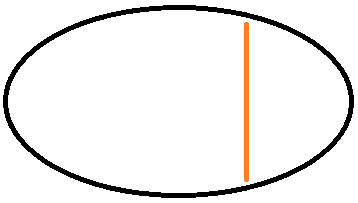}
    \caption{Visualisation of the cross-section of the beam line and a barrier, called the septum sheet. In the first round, the newly injected electrons are to the right of the septum-sheet. In the next rounds they are located to the left of the septum-sheet. The stored beam is located in the centre of the beam line.}
    \label{fig:septum}
\end{figure}
The stored beam is located in the centre of the cross section of the beam line, as seen in figure \ref{fig:septum}. The barrier shown is a magnet, called the septum-sheet. It plays an important role in the 4-kicker bump, seen in figure \ref{fig:4kicker_bump}, as it merges the injected electrons with the stored electrons. It is also important for the non-linear kicker injection, as it pulls the electrons towards the centre of the beam line. The position of the septum-sheet, in comparison to the non-linear kicker, is shown in figure \ref{fig:bessy2}.

In the first round, where the electrons are injected, the electrons are on the right side of the septum-sheet. In the following rounds, the location of the electrons at the septum-sheet must be on the left side of the septum-sheet. Otherwise, the electrons have crashed into the septum-sheet and are considered lost.

As the electrons follow a phase transition path, visualised in figure \ref{fig:stored beam}, they will eventually crash into the septum-sheet, so we need to influence them with the non-linear kicker.

To visualise the amount of rounds the electrons survive without the non-linear kicker, we have created a plot, that shows for $x,p_x$ values of single electrons at the septum-sheet how many rounds they survive without any influence from the non-linear kicker. We further assume that there is no noise in the round-to-round behaviour. The results are given in the following figure \ref{fig:without_kick}.
\begin{figure}[H]
    \centering
    \includegraphics[width=13cm]{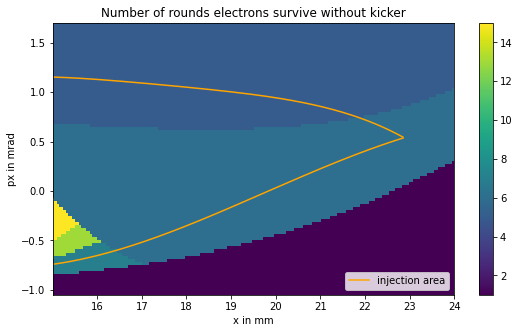}
    \caption{Plot showing the amount of rounds newly injected electrons survive without kick. The area within the orange lines shows the injection area selected for the studies presented here.}
    \label{fig:without_kick}
\end{figure}
On the x-axis, we observe the x-position of the electrons with respect to the centre of the beam line. Note that these x-positions are $>15\mathrm{mm}$ because the septum-sheet is positioned at $15\mathrm{mm}$ and we focus here on newly injected electrons. On the y-axis, we can see the direction of flight of the electrons in $\mathrm{mrad}$. The direction of flight value gives in $\mathrm{mrad}$ in which angle the electrons are moving left/right as they fly. The colour in the middle (at $x=19\mathrm{mm}$, $\mathrm{px}=0.0\mathrm{mrad}$) represents 6 rounds, above this area we have 4 rounds, below a single round and the 3 small areas represent from bottom to top: 7,13 and 15 rounds. 

Within this plot we highlighted an area which we call injection area for further studies. It was selected due to boundary conditions of the BESSY II machine. The area is described as the region between 2 polynomial functions. For our injection problem, we assume that the mean values of the electron distribution at the beginning of the injection are within this area.

In particular, the plot \ref{fig:without_kick} shows that we do not need to kick the electrons in the very first round. Instead, there are at least five rounds in the injection area that we could wait for before taking action. In the following subsection, we provide an understanding of when and how the electrons are kicked optimally. 

\subsection{Optimal Kicker Usage}\label{subsec:optimal_kicker_usage}
In this subsection, we give an introduction to when and how to use the non-linear kicker optimally. For now, we continue our analysis on a deterministic system. First, we state a necessary condition that electrons must fulfil to be successfully kicked and then show how to find the optimal kicker strength.

The necessary condition for a single electron is:
\begin{quote}
\textit{The x-position of the electron at the non-linear kick must be between $-15\mathrm{mm}$ and $+15\mathrm{mm}$.}
\end{quote}
Suppose an electron does not satisfy this condition. Since the electrons follow a phase space transition, as shown in the figure \ref{fig:stored beam}, no matter much how we change the $\mathrm{px}$ value of the electron, after a few rounds it reaches the $15\mathrm{mm}$ mark at the non-linear kicker, thus also at the septum-sheet and therefore they crash into it. For this reason, the condition must be fulfilled.

Now assume that at the non-linear kicker the x-position of an electron is within $(-15\mathrm{mm},15\mathrm{mm})$. If we could now use the non-linear kicker to change the direction of flight value $\mathrm{px}$ of the electron to approximately $0$, the electron  survives the injection. As in the phase space transition a $\mathrm{px}$ value of $0$ indicates that the electron is at the highest/lowest x-value it will get. For this reason, all further x-values of the electron at the non-linear kicker, including the septum sheet will not exceed the value $15\mathrm{mm}$ and therefore will not crash into the septum-sheet. 

Now, we show how the linearity of the magnetic field strength with respect to the kicker strength is used to calculate the optimal kicker strength. Let $(x,p_x)$ be the information of an electron at the non-linear kicker. Assume that $x\in(-15\mathrm{mm},15\mathrm{mm})$. Let $\mathcal{B}(\alpha,x)$ be the function that maps the kicker strength $\alpha\in[-1,1]$ and $x$ position of an electron to the change in $p_x$ value, i.e. $p_x' = p_x+\mathcal{B}(\alpha,x)$, where $p_x'$ is the $p_x$ value of the electron immediately after the non-linear kicker. Note that $\mathcal{B}$ is linear in $\alpha$. The optimal kicker strength $\alpha^*$ is calculated as
\[
p_x' \overset{!}{=} 0 = p_x+\mathcal{B}(\alpha^*,x) \quad\Leftrightarrow\quad \alpha^*\mathcal{B}(1,x) = -p_x \quad\Leftrightarrow\quad \alpha^* = -\frac{p_x}{\mathcal{B}(1,x)}.
\]
With this equation, we have found an easy way to calculate the optimal kicker strength $\alpha^*$. Note that for small values of $\mathcal{B}(1,x)$ or large values of $p_x$, the calculated strength $\alpha^*$ may not be realisable, as it is no longer within $[-1,1]$. In this case, it is not possible to kick the electrons optimally.

The following plots show the $x,p_x$ values at the septum-sheet in the very first round, when the values meet the condition (not the white regions) and when the kick was successful (yellow (1000) region). In all other regions (not white or yellow), the kick was not successful. There are several reasons for this, but the most common one is that the simulation has difficulties with some borderline cases.

\begin{figure}[H]
    \centering
    \includegraphics[width=9cm]{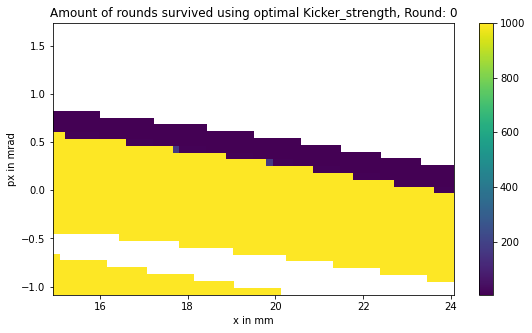}
    \caption{Regions where electrons can be optimally kicked in round 0}
    \label{fig:round0}
\end{figure}
\begin{figure}[H]
    \centering
    \includegraphics[width=9cm]{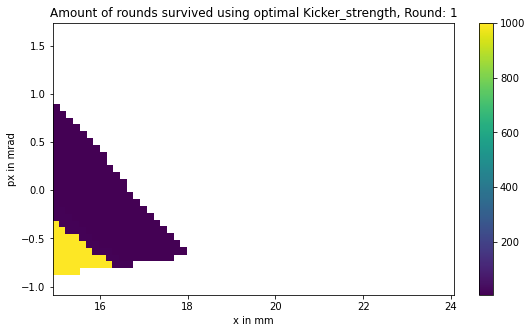}
    \caption{Regions where electrons can be optimally kicked in round 1}
    \label{fig:round1}
\end{figure}
\begin{figure}[H]
    \centering
    \includegraphics[width=9cm]{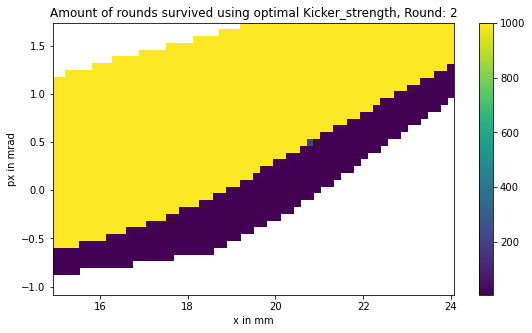}
    \caption{Regions where electrons can be optimally kicked in round 2}
    \label{fig:round2}
\end{figure}
\begin{figure}[H]
    \centering
    \includegraphics[width=9cm]{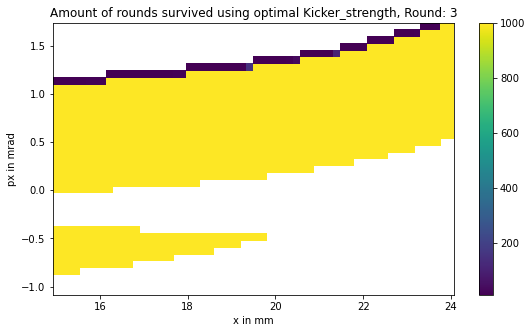}
    \caption{Regions where electrons can be optimally kicked in round 3}
    \label{fig:round3}
\end{figure}
\begin{figure}[H]
    \centering
    \includegraphics[width=9cm]{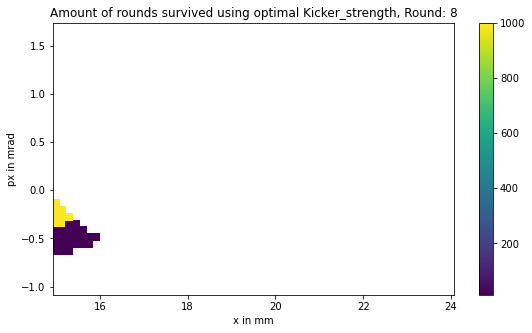}
    \caption{Regions where electrons can be optimally kicked in round 8}
    \label{fig:round8}
\end{figure}
These figures show for which initial values the electrons are successfully kicked, depending on which round the non-linear kicker was activated. We observe that it is not optimal to kick only in round 0 or in any other round, as different rounds cover different regions. Interestingly, using the non-linear kicker in round 1 is significantly worse compared to the rounds 0,2,3. The plots for all other rounds can be seen in the programming reference, referenced at the beginning of the chapter \ref{chap:calculations}.

    

The next plots show the area in which electrons are successfully kicked, together with the optimal kicker strength for the whole beam line. The x and px values shown are measured at the septum-sheet.  All regions in figure \ref{fig:optimal_area}, with values between 0 and 1000, meet the necessary condition for a successful kick. For the yellow area (value: 1000) the kick is successful. The values -200 in figure \ref{fig:optimal_area} and -1 in figure \ref{fig:optimal_strength} are just background colours to highlight the other regions.
\begin{figure}[H]
    \centering
    \includegraphics[width=9cm]{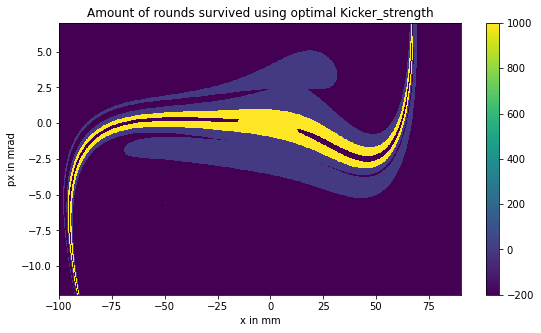}
    \caption{Given x,px values at the septum sheet, the plot shows for which of them the injection can be done (score: 1000 rounds).}
    \label{fig:optimal_area}
\end{figure}
\begin{figure}[H]
    \centering
    \includegraphics[width=9cm]{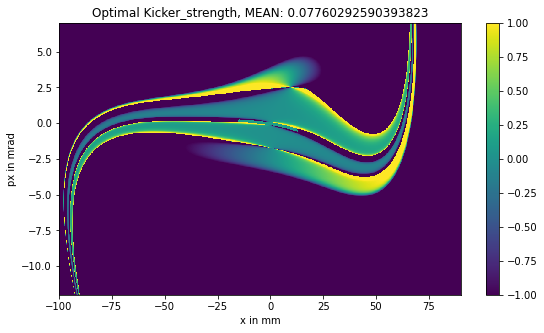}
    \caption{Corresponding optimal kicker strength, see figure \ref{fig:optimal_area}.}
    \label{fig:optimal_strength}
\end{figure}

We intend to answer the question of how the injection efficiency changes for different kicker strengths. For this experiment, we added noise to the round-to-round behaviour, that is for both x and px we added zero mean normal distributed noise, where the noise for the x values had std: $6.5\cdot 10^{-6}\mathrm{m}$ and for px std: $5.2\cdot 10^{-6}\mathrm{m}$). For an explanation of why we chose these noise values see subsection \ref{subsec:creating_stochas}. No specific noise was added to the kicker strength. The following figure shows the result
\begin{figure}[H]
    \centering
    \includegraphics[width=9cm]{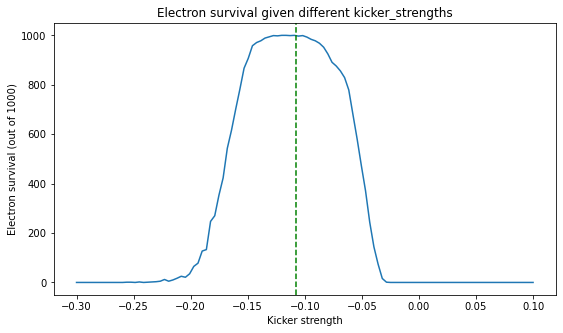}
    \caption{For each kicker strength, we run 1000 simulations with the same parameters (x-value $16.4\mathrm{mm}$, px $0.0\mathrm{mm}$, activated NLK in round 4 with given kicker strength). The y-value shows in how many of these 1000 experiments the electron survived. The green dashed line shows the optimal kicker strength.}
    \label{fig:kicker_strength_behav}
\end{figure}

In figure \ref{fig:kicker_strength_behav}, we observe that the injection efficiency is highest around the optimal kicker strength with a small plateau. Outside the plateau we see a very steep drop in efficiency.

\clearpage
\subsection{Effect on stored beam}
The following figure shows the absolute effect in the direction of flight (in mrad) for different x/px values at the septum sheet close to the centre, where the stored beam is.

\begin{figure}[H]
    \centering
    \includegraphics[width=9cm]{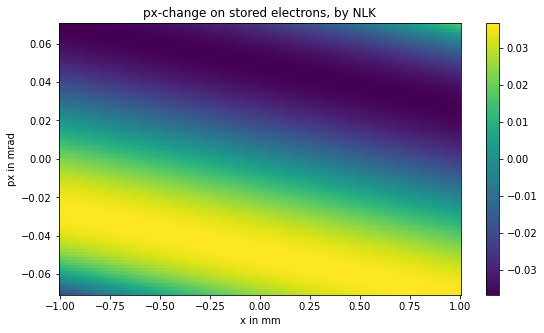}
    \caption{Plot showing the difference in direction of flight, created by the non-linear kicker with strength 1. Each x/px value shows the value at the septum sheet.}
    \label{fig:effect_on_stored_electrons}
\end{figure}
In the figure \ref{fig:effect_on_stored_electrons} we observe the effect of the non-linear kicker, activated with the highest possible strength 1, on the stored electrons which lie in the range of -1mm to 1mm. We observe that the change in direction of flight is in the range of $[-0.04\mathrm{mrad},0.04\mathrm{mrad}]$, which is not significantly higher than the previous values.  

As the non-linear kicker is typically not used at full strength, the effect is even less, see figure \ref{fig:optimal_strength}.

\chapter{Problem Modelling}\label{chap:problem_modelling}
In this chapter, we examine how to model the non-linear kicker injection as a Markov Decision Process.

\section{Problem Description}
%
%
In this section, we describe exactly what the non-linear kicker injection problem is.

The goal of the injection problem is to maximise the number of electrons that survive the injection. We say that an electron has survived the injection if it does not crash into the septum-sheet in the first 1000 rounds. 

To achieve this goal, the use of the non-linear kicker must be optimised. The activation round of the kicker and the kicker-strength are to choose. Underlying, there is the constraint that the non-linear kicker can only be activated once, as we do not have enough time during the injection to recharge the non-linear kicker.

At the beginning of the injection process the electrons have a normal distributed shape. They maintain this normal distributed shape in each round, unless some of the electrons crash into the septum-sheet. In this case, the electrons follow a truncated normal distribution.

In the next section, we describe this problem as a Markov Decision Process. We propose four descriptions. 

In the first 3 descriptions we consider Markov Decision Models with a planning horizon of $N=1000$ and assume that we are able to adjust the non-linear kicker as we observe the electrons at the septum-sheet. Although this assumption is unrealistic as the electrons travel at almost the speed of light and the septum-sheet and non-linear kicker are positioned close to each other, we want to analyse how reinforcement learning algorithms deal with such environments. The first environment shows the model environment for truncated normal distributions. In the second, we model the injection for a single electron and then, in the third model, for 1000 electrons.

The last description is the most realistic to implement. In this fourth model, we require that the activation round of the non-linear kicker together with the kicker strength be predicted after the electrons are observed at the septum sheet in the very first round. For details of how this model is implemented at BESSY II see remark \ref{rem:one_step_MDP}.

\section{Model}\label{sec:model_description}

\begin{defn}[truncated normal distributions]\label{def:trunc_normal}
    Let $\mu\in\R^2$ be the mean vector and $\Sigma\in\R^{2,2}$ be the covariance matrix of a 2-dimensional normal distribution with density function
    \[
    f(x;\mu,\Sigma))=\frac{1}{\sqrt{(2\pi)^2\det(\Sigma)}}\exp{\left(-\dfrac12(x-\mu)^T\Sigma^{-1}(x-\mu)\right)}.
    \]
    We truncate this normal distribution by dividing the $\R^2$ space into 2 half spaces and setting the density on one half space to zero. To do this, we need a vector $y\in\R^2$ and direction $z\in\R^2$. Further, we truncate the density function $f$ by:
    \[
    f^1(x;\mu,\Sigma,y,z)=\left\{\begin{array}{rl}
        f(x;\mu,\Sigma,y,z), & \text{ if } \langle x-y,z\rangle > 0  \\
        0, & \text{ else}
    \end{array}\right.
    \]
    where $\langle.\,,.\rangle$ is the standard scalar product on $\R^2$.
    
    Note that the truncated function $f^1$ is no longer a density function, since the integral over $\R^2$ is not 1. We intentionally do not normalise it, to calculate in a straightforward manner how much of the initial distribution remains after the truncation, by taking the integral over the remaining density.

    In the following, we consider $n$ times truncated normal distributions, where for the sequences of vectors $(y_k)_{k=1}^m\coloneqq(y_k)_{k=1,\dots,m},(z_k)_{k=1}^m\coloneqq(z_k)_{k=1,\dots,m}\subset \R^2$ and all $m=1,\dots,n$, we have
    \[
    f^m(x;\mu,\Sigma,(y_k)_{k=1}^m,(z_k)_{k=1}^m)=\left\{\begin{array}{rl}
        f^{m-1}(x;\mu,\Sigma,(y_k)_{k=1}^{m-1},(z_k)_{k=1}^{m-1}), & \text{ if } \langle x-y_m,z_m\rangle > 0  \\
        0, & \text{ else}
    \end{array}\right.
    \]
    with $f^0=f$.
\end{defn} 
\clearpage
Now, we provide the definition of the Markov Decision Model, which shows how the non-linear kicker can be used. After the definition, the details are explained in remark \ref{rem:mdp_explanation}.

We consider the Markov Decision Model with planning horizon $N=1000$ with
\begin{itemize}
        \item state space
        \[
        E\coloneqq \bigcup_{m=0}^{999}\{(\R^2\times \R^{2,2}\times\{0,1\})^{m+1}\times\{\dagger,\dagger,\dagger\}^{999-m}\},
        \] 
        where each element $((\mu_0,\Sigma_0,k_0),\dots,(\mu_m,\Sigma_m,k_m),\{\dagger,\dagger,\dagger\}^{999-m})\in E$ consists of mean vector $\mu_t$, covariance matrix $\Sigma_t$, the information if the kicker has already been used $k_t$  for all rounds $t=0,\dots,m$. For the remaining rounds, to indicate that there is no information yet, we denote them with $\dagger$. Further, we equip the state space $E$ with the Borel $\sigma$-algebra, i.e. $\mathcal{E}\coloneqq\mathcal{B}(E)$. 
        \item Action space $A\coloneqq[-1,1]$ with Borel $\sigma$-algebra $\mathcal{A}\coloneqq\mathcal{B}([-1,1])$
        \item To define the stochastic transition kernel $p$, consider the following transitions between time steps. Let $\mathcal{T}^\mu,\mathcal{T}^\Sigma$ be the round-to-round transition functions defined in section \ref{sec:r2rbehav}. For state $s_t=((\mu_0, \Sigma_0, k_0), \dots,(\mu_t, \Sigma_t, k_t), \{\dagger,\dagger,\dagger\}^{999-t})\in E$, for a $t\in\{0,\dots,999\}$ and action $a_t\in A$ the state at time step $t+1$ follows the rules:
        \begin{itemize}
                \item $\tilde a_t = (a_t+\varepsilon_t^a)\mathbbm{1}_{\{a_t\neq 0\}}$
                \item $\mu_{t+1}=\mathcal{T}^\mu(\tilde a_t(1-k_t),\mu_t,\Sigma_t)+\varepsilon_t^\mu$
                \item $\Sigma_{t+1}=\mathcal{T}^\Sigma(\tilde a_t(1-k_t),\mu_t,\Sigma_t)+\varepsilon_t^\Sigma$
                \item $k_{t+1}=\max\{k_t,\mathbbm{1}_{\{a_t\neq 0\}}\}$
            \end{itemize}
        where $\varepsilon_t^a,\varepsilon_t^\mu$ and $\varepsilon_t^\Sigma$ are normal distributed independent noise in $\R,\R^2$ and $\R^{2,2}$ respectively. We then set the state $s_{t+1}$ equal to
        \[
        s_{t+1}=((\mu_1, \Sigma_1, k_1), \dots,(\mu_t, \Sigma_t, k_t),(\mu_{t+1}, \Sigma_{t+1}, k_{t+1}), \{\dagger,\dagger,\dagger\}^{999-t-1})\in E
        \]
        We consider the process $(s_t)_{t=0,\dots,100}$ as stochastic process on the measurable space $(E,\mathcal{B}(E))$. 

        The stochastic transition kernel is set as the natural kernel given the transition rules. 

        \item the initial distribution $p_0$ is set as follows way: All $s_0$ have structure
        \[
        s_0=((\mu_0,\Sigma_0,0),\{\dagger,\dagger,\dagger\}^{999}),
        \]
        where $\mu_0,\Sigma_0$ are uniformly distributed from an area, defined in section \ref{subsec:creating_stochas}.
        \item the running reward function is set to 0, so $r\colon (x,a)\mapsto0$
        \item discount factor is set to $\beta=1$
        \item terminal reward function $g\colon E\to\R$ is defined as
        \[
        g(s)=\int_{\R^2}\left(\mathcal{T}^g(s)\right)(x) \mathrm{d}x,
        \]
        where $\mathcal{T}^g(s)$ is a 1001-times truncated density of a normal distribution. The function $\mathcal{T}^g(s)$ is defined in the following remark \ref{rem:mdp_explanation}. 
    \end{itemize}

\begin{remark}\label{rem:mdp_explanation}
   We have chosen to define the state space $E$ as
   \[
        E\coloneqq \bigcup_{m=0}^{999}\{(\R^2\times \R^{2,2}\times\{0,1\})^{m+1}\times\{\dagger,\dagger,\dagger\}^{999-m}\},   
   \]
   to capture the history of the mean and sigma values, as this is important in calculating the terminal reward $g$, and with the information $k_t\in\{0,1\}$, we know whether the kicker was used or not. We added the $\dagger$ to keep track of which rounds we have not yet visited.  

   We set the action space $A$ to $A=[-1,1]$ because the intensity of the kick can be varied, and negative values are possible if the wires can be pulsed with the opposite current. 

   For the transitions we first added noise to the chosen action $a_t$, because there is also noise at the chosen kicker strength.  
   With the transition function $\mathcal{T}^\mu,\mathcal{T}^\Sigma$ from section \ref{sec:r2rbehav}, we can calculate $\mu_{t+1},\Sigma_{t+1}$. We add noise to both, for some uncertainty. Note that the action is set to $\tilde a_t(1-k_t)$ because if the kicker has not been activated, this value will be $\tilde a_t$, and if the kicker has already been activated, i.e. $k_t=1$, the value will be $0$. For the update of $k_t$ we used the formula $k_{t+1}=\max\{k_t,\mathbbm{1}_{\{a_t\neq 0\}}\}$, such that if the kicker has already been activated, i.e. $k_t=1$, the value $k_{t+1}$ is set to $1$ and if the kicker was activated at time $t$, i.e. $k_t=0$ and $a_t\neq0$, the value is also updated to 1.   
   The running reward function $r$ is set to 0, since the goal is to maximise the number of electrons that survive round $1000$, and thus only round $1000$ is used for the reward. The discount factor $\beta$ is set in a similar way to $1$.

    We move on to define the terminal reward function $g\colon E\to\R$. The idea is that in each round a fraction of the electrons crash into the septum sheet. We model the survived electrons using truncated densities of the normal distribution defined in definition \ref{def:trunc_normal}.

    Let $s_t=((\mu_0, \Sigma_0, k_0), \dots,(\mu_t, \Sigma_t, k_t), \{\dagger,\dagger,\dagger\}^{999-t})\in E$ be arbitrary. We iteratively define the $t+1$ times truncated density of a normal distribution $\mathcal{T}^g(s_t)\colon \R^2\to\R$.

    In definition \ref{def:trunc_normal}, we denoted the density of a normal distribution with mean $\mu$ and covariance matrix $\Sigma$ by $f(x;\mu,\Sigma)$ and the truncation was given by 
    \[
    f^1(x;\mu,\Sigma,y,z)=\left\{\begin{array}{rl}
        f(x;\mu,\Sigma,y,z), & \text{ if } \langle x-y,z\rangle > 0  \\
        0, & \text{ else}
    \end{array}\right..
    \]
    In the first round, the truncation is given as $f^1(x;\mu_0,\Sigma_0,(15,0),(1,0))$. We use $y_1=(15,0)\in\R^2$, because the septum-sheet is at the x-value $15\mathrm{mm}$ and $z_1=(1,0)$, as the electrons with x-value lower or equal the value of $15\mathrm{mm}$ are truncated, since in the first round the electrons are outside (to the right) of the septum-sheet.  

    To propagate truncations from one round to the next, we assume that functions 
    \[
    y'=\mathcal{T}^y(\mu_t,\Sigma_t,\mu_{t+1},\Sigma_{t+1},y) \quad\text{and}\quad
    z'=\mathcal{T}^z(\mu_t,\Sigma_t,\mu_{t+1},\Sigma_{t+1},z)
    \] 
    exist, that are able to propagate the truncations. This can be achieved by using the round-to-round behaviour function for single electrons $\mathcal{T}$, where one could follow the edge of the truncation. Note that for the function $\mathcal{T}
    ^y$ and $\mathcal{T}^z$, the information of the mean and covariance at time-points $t$ and $t+1$ are needed, in order to take the noise into account. For well-posedness, we set $\mu_{1002}\coloneqq\mathcal{T}^\mu(0,\mu_{1001})$, analogous for $\Sigma_{1002}$. We do not go into more detail, as these functions are not needed in our implementation. 

    Instead, we continue to iteratively define the function $\mathcal{T}^g$. Suppose we have already truncated the normal distribution $m$-times with vectors $(y_k)_{k=1}^m,(z_k)_{k=1}^m$.
    Now we propagate them to the next round with
    \begin{align}
    y_k'\coloneqq    \mathcal{T}^y(\mu_t,\Sigma_t,\mu_{t+1},\Sigma_{t+1},y_k) \quad\text{and}\quad
    z_k'\coloneqq\mathcal{T}^z(\mu_t,\Sigma_t,\mu_{t+1},\Sigma_{t+1},z_k).
    \end{align}
     Additionally, we add the truncation from the new round and set
     \[
     (y_1',\dots,y_m',(15,0))\quad\text{and}\quad(z_1',\dots,z_m',(-1,0))
     \]
    as the truncations for the next round. Note that the new z-value is set to $(-1,0)$, as now all electrons with an x-value greater or equal $15\mathrm{mm}$ are truncated.

    We continue with these truncations up to round $t$. Let $(y_k)_{k=1}^t,(z_k)_{k=1}^t$ be the final truncations. We then define the function $\mathcal{T}^g(s_t)$ as
    \[
    \mathcal{T}^g(s_t)\coloneqq  f^t(\cdot;\mu_t,\Sigma_t,(y_k)_{k=1}^t,(z_k)_{k=1}^t).
    \]
    Now, we use use this function to determine the percentage of electrons that survived the 1000 rounds by calculating
    \[
    g(s_t)=\int_{\R^2}\left(\mathcal{T}^g(s_t)\right)(x) \mathrm{d}x.
    \]
    Note that $\mathcal{T}^g$ is integrable because only finite truncations have been made. Also, the function $g$ is measurable since $\mathcal{T}^g(s_t)$ is measurable. 
\end{remark}
\vspace{.8cm}
This ends the part on the Markov Decision Process (MDP). The explained MDP models the problem of deciding when to activate the non-linear kicker and with what strength, under the unrealistic assumption that we have enough time to process the information and activate the non-linear kicker. The problem with this modelling is the high dimensional state space and the computational cost of dealing with distributions. 

Instead, we introduce slightly modified MDPs and Partially Observable Markov Decision Processes (POMDPs), which describe the problem of optimally using the non-linear kicker in different scenarios and levels of difficulty. In the first model, we have limited ourselves to a single electron. After that we see a POMDP for 1000 electrons and at the end a one-step POMDP for 1000 electrons.  

\subsection{Single Electron Kick}\label{subsec:sing_elec}
For the single electron kick, we consider the following Markov Decision Model with planning horizon $N=1000$ with
\begin{itemize}
        \item state space
        \[
        E\coloneqq \R^2\times\{0,1\},\times\{0,\dots,1000\}\times\{0,1\},
        \] 
        where each element $(x,p^x,k,r,i)\in E$ consists of the current  x-position $x$ and the direction of flight information $p^x$, if the kicker has already been activated $k$, the round number $r\in\{0,\dots,1000\}$ and the information $i$ whether the electron has hit the septum-sheet or not. Furthermore, we equip the state space $E$ with the Borel $\sigma$-algebra, i.e. $\mathcal{E}\coloneqq\mathcal{B}(E)$. 
        \item Action space $A\coloneqq[-1,1]$ with Borel $\sigma$-algebra $\mathcal{A}\coloneqq\mathcal{B}([-1,1])$
        \item To define the stochastic transition kernel $p$, consider the following transitions between time steps. Let $\mathcal{T}$ be the round-to-round transition function defined in section \ref{sec:r2rbehav}. For state $s_t=(x_t,p^x_t,k_t,r_t,i_t)$, for a $t\in\{0,\dots,999\}$ and action $a_t\in A$, the state at time step $t+1$ follows the rules:
        \begin{itemize}
                \item $\tilde a_t = (a_t+\varepsilon_t^a)\mathbbm{1}_{\{a_t\neq 0\}}$
                \item $(x_{t+1},p^x_{t+1})=\mathcal{T}(\tilde a_t(1-k_t),x_t,p^x_t)+(\varepsilon_t^x,\varepsilon_t^{p^x})$
                \item $k_{t+1}=\max\{k_t,\mathbbm{1}_{\{a_t\neq 0\}}\}$
                \item $r_{t+1}=r_t+1$
                \item $i_{t+1}=\max\{i_t,\mathbbm{1}_{\{x_{t+1}\geq 15\mathrm{mm}\}}\}$   
            \end{itemize}
        where $\varepsilon_t^a,\varepsilon_t^x$ and $\varepsilon_t^{p^x}$ are normal distributed independent noise in $\R$. We then set the state $s_{t+1}$ equal to
        \[
        s_{t+1}=(x_{t+1},p^x_{t+1},k_{t+1},r_{t+1},i_{t+1})\in E
        \]
        We consider the process $(s_t)_{t=0,\dots,1000}$ as a stochastic process on the measurable space $(E,\mathcal{B}(E))$. 

        The stochastic transition kernel is set as the natural kernel given the transition rules. 

        \item the initial distribution $p_0$ is set as follows: All $s_0$ have the structure
        \[
        s_0=(x_0,p^x_0,0,1,\mathbbm{1}_{\{x_0\leq 15\}}),
        \]
        where $x_0,p^x_0$ are uniformly distributed from an area defined in section \ref{subsec:without_kick}.
        \item the running reward function is set to 0, so $r\colon (x,a)\mapsto0$
        \item discount factor is set to $\beta=1$
        \item terminal reward function $g\colon E\to\R$ is defined as
        \[
        g((x_t,p^x_t,k_t,r_t,i_t))\coloneqq\mathbbm{1}_{\{i_{t}=0\}},
        \]
    \end{itemize}
\begin{remark}
    We have added the round number $r_t$ in the state space because the TD3 algorithm does not change the policy depending on the round. By adding the round number to the state space, it may depend on the round. Additionally, the information $i_t$ is added, whether the single electron has crashed into the septum-sheet or not. 

    The terminal reward is $1$ if the electron survived the 1000 rounds, otherwise it is $0$.
\end{remark}

\subsection{1000-Electron Kick}\label{subsec:1000_elec}
For the 1000 electron kick, we consider the following Partially Observable Markov Decision Model with planning horizon $N=1000$ with
\begin{itemize}
        \item state space
        \[
        E_X\times E_Y\coloneqq \big(\R^2\times\R^4\times\{0,1\},\times\{0,\dots,1000\}\big)\times \big((\R^2\times\{0,1\})^{1000}\times \{0,1\}\big),
        \] 
        where the observable part $(\mu,\mathcal{I},k,r)\in E_X$ consists of the empirical mean of the $x/p^x$-values $\mu\in\R^2$, a shape description of the electron cloud $\mathcal{I}\in\R^4$, the information whether the kicker has already been activated $k$ and a round number $r\in\{0,\dots,1000\}$. 

        The non-observable part $(((x_l,p^x_l),i_l)_{l=1}^{1000},k)\in E_Y$ consists of the x-position $x_l$, the direction of flight $p^x_l$ and the information $i_l$ whether the electron with number $l\in\{1,\dots,1000\}$ has crashed against the septum-sheet or not. The number $k$ describes whether the kicker has already been activated or not.
        \item Action space $A\coloneqq[-1,1]$ with Borel $\sigma$-algebra $\mathcal{A}\coloneqq\mathcal{B}([-1,1])$
        \item To define the stochastic transition kernel $p$, consider the following transitions of the non-observable space between time steps. Let $\mathcal{T}$ be the round-to-round transition function defined in section \ref{sec:r2rbehav}. For a state $s_t=(((x_l^t,p_{x,l}^t),i_l^t)_{l=1}^{1000},k_t)\in E_Y$, for a $t\in\{0,\dots,999\}$ and action $a_t\in A$ the state in time step $t+1$ follows the rules for all $l=1,\dots,100$:
        \begin{itemize}
                \item $\tilde a_t = (a_t+\varepsilon_t^a)\mathbbm{1}_{\{a_t\neq 0\}}$
                \item $(x^{t+1}_l,p_{x,l}^{t+1})=\mathcal{T}(\tilde a_t(1-k_t),x^t_l,p_{x,l}^t)+(\varepsilon_t^x,\varepsilon_t^{p^x})$
                \item $i^{t+1}_l=\max\{i^t_l,\mathbbm{1}_{\{x^{t+1}_l\geq 15\mathrm{mm}\}}\}$  
                \item $k_{t+1}=\max\{k_t,\mathbbm{1}_{\{a_t\neq 0\}}\}$
            \end{itemize}
        where $\varepsilon_t^a,\varepsilon_t^x$ and $\varepsilon_t^{p^x}$ are normal distributed independent noises in $\R$ and in particular $\varepsilon_t^x,\varepsilon_t^{p^x}$ are equal for all $l$. We then set the state $s_{t+1}$ equal to
        \[
        s_{t+1}=(((x_l^{t+1},p_{x,l}^{t+1}),i_l^{t+1})_{l=1}^{1000},k_{t+1})\in E_Y.
        \]
        The observable part is then deterministically given by
        \begin{itemize}
            \item Let $n\coloneqq\sum_l\mathbbm{1}_{\{i_l^{t+1}=0\}}$, the number of electrons that  have not crashed by time-point $t+1$, then we define
            \[
            \mu_{t+1}=\frac{1}{n}\sum^{1000}_{\substack{l=1\\ \text{with }i_l^{t+1}=0}}(x_l^{t+1},p_{x,l}^{t+1})\in\R^2
            \]
            \item Let $(H^z_{ij})_{i,j=0,\dots,50}\in[0,1000]$ be the histogram matrix around the mean $\mu_{t+1}$ with window lengths $z\in\R^2$. Let $\mathcal{H}_\theta\colon\R^{51,51}\to\R^4$ be a continuous function which may depend on parameters $\theta\in\R^{m}$, for some $m\in\N$. Then, define
            \[
            \mathcal{I}_{t+1}\coloneqq \mathcal{H}_\theta(H^z).
            \]
        \end{itemize}
        The observable part is then given by $o_{t+1}=(\mu_{t+1},\mathcal{I}_{t+1},k_{t+1},t+1)\in E_X$.
 
        The stochastic transition kernel is set as the natural kernel given the transition rules. 
        \item The initial distribution $p_0$ is set as follows: All $s_0\in E_Y$ have structure
        \[
        s_0=(((x_l,p_{x,l}),0)_{l=1}^{1000},0)\in E_Y
        \]
        where all $x_l,p_{x,l}$ are sampled from a normal distribution with mean uniformly sampled from the area defined in section \ref{subsec:without_kick} and the covariance matrix
        \[
        \Sigma = \begin{pmatrix}
            0.00112 & 0  \\
            0 & 6.29\cdot10^{-5}
                 \end{pmatrix}
        \]
        in meters and rad.
        
        An explanation of why this covariance matrix was chosen can be found in subsection \ref{subsec:creating_stochas}.
        \item the running reward function is set to 0, so $r\colon (x,a)\mapsto0$
        \item discount factor is set to $\beta=1$
        \item terminal reward function $g\colon E_X\times E_Y\to\R$ is defined as
        \[
        g(o_t,s_t)\coloneqq\sum_{l=1}^{1000}\mathbbm{1}_{\{i_l^{t}=0\}},
        \]
        for all $s_t=(((x_l^t,p_{x,l}^t),i_l^t)_{l=1}^{1000},k_t)\in E_Y$ and $o_t\in E_X$.
    \end{itemize}
\begin{remark}
    The function $H_\theta$ in the definition of $\mathcal{I}_{t+1}$ is used to describe what the distribution of the electrons looks like. There are a number of ways to do this. For example, one could calculate the covariance matrix or some of the stochastic moments of the distribution. In our calculations we trained a convolutional neural network, that was given the histogram matrix and returned a 4-dimensional vector, with the goal of maximising the reward. 
\end{remark}

\subsection{One-step 1000-Electron Kick}\label{subsec:1000_elec_1step}
Now, we slightly modify the 1000-electron kick to a one-step 1000-electron kick. We do this, by considering the following Partially Observable Markov Decision Model with planning horizon $N=1$ with
\begin{itemize}
        \item state space
        \[
        E_X\times E_Y\coloneqq \big(\R^2\times\R^4\big)\times \big((\R^2\times\{0,1\})^{1000}\big),
        \] 
        where the observable part $(\mu,\mathcal{I})\in E_X$ consists of the empirical mean of the $x/p^x$ values $\mu\in\R^2$, a shape description of the electron cloud $\mathcal{I}\in\R^4$.

        The non-observable part $((x_l,p^x_l),i_l)_{l=1}^{1000}\in E_Y$ consists of the x-position $x_l$, the direction of flight $p^x_l$ and the information $i_l$ whether the electron with number $l\in\{1,\dots,1000\}$, has crashed against the septum-sheet or not. 
        \item Action space $A\coloneqq[-1,1]\times[-1,1]$ with Borel $\sigma$-algebra $\mathcal{A}\coloneqq\mathcal{B}(A)$.
        For an action $a=(a^1,a^2)$ the number $a_1\in[-1,1]$, is used to determine in which round the electrons are kicked. Let $\tilde a^1\coloneqq (a^1+1)\cdot 5$, then the round $r$ in which the electrons are kicked is given by
        \[
        r=\underline{\tilde a^1}+\varepsilon
        \]
        where $\varepsilon\sim\mathrm{Ber}(\tilde a^1-\underline{\tilde a^1})$ distributed. With $\underline a$ we denote the operation to round down $a$ to the next integer.

        The number $a_2\in[-1,1]$ is used for the kicker strength.
        \item To define the stochastic transition kernel $p$, consider the following transitions of the non-observable space between time steps. Let $s_0=((x_l^0,p_{x,l}^0),i_l^0)_{l=1}^{1000}\in E_Y$ be the initial state and  $a_0=(a^1,a^2)\in A$ the chosen action. Let $r$ be calculated as above.

        Define $a_t=0$ for all $t\in\{0,\dots,1000\}\setminus\{r\}$ and set $a_r=a^2$. 

        With these actions we now follow the transition rules given in \ref{subsec:1000_elec}, and let $\tilde s_{1000}=((x_l^{1000},p_{x,l}^{1000}),i_l^{1000})_{l=1}^{1000}\in E_Y$ be the last state.
        
         Then, we set the state $s_{1} = \tilde  s_{1000}$.
         
        The observable part is deterministically given for all $s_t\in E_Y$ by
        \begin{itemize}
            \item Let $n\coloneqq\sum_l\mathbbm{1}_{\{i_l^{t}=0\}}$, the number of electrons that have not crashed by time-point $t$, then we define
            \[
            \mu_{t}=\frac{1}{n}\sum^{1000}_{\substack{l=1\\ \text{with }i_l^{t}=0}}(x_l^{t},p_{x,l}^{t})\in\R^2
            \]
            \item Let $(H^z_{ij})_{i,j=0,\dots,50}\in[0,1000]$ be the histogram matrix around the mean $\mu_{t+1}$ with window lengths $z\in\R^2$. Let $\mathcal{H}_\theta\colon\R^{51,51}\to\R^4$ be a continuous function that may depend on parameters $\theta\in\R^{m}$, for some $m\in\N$. Then, define
            \[
            \mathcal{I}_{t}\coloneqq \mathcal{H}_\theta(H^z).
            \]
        \end{itemize}
        The observable part is then given by $o_{t}=(\mu_{t},\mathcal{I}_{t})\in E_X$.
 
        The stochastic transition kernel is set as the natural kernel given the transition rules. 
        \item the initial distribution $p_0$ is set as follows: All $s_0\in E_Y$ have structure
        \[
        s_0=(((x_l,p_{x,l}),0)_{l=1}^{1000},0)\in E_Y
        \]
        where all $x_l,p_{x,l}$ are sampled from a normal distribution with mean uniformly sampled from the area defined in section \ref{subsec:without_kick} and the covariance matrix
        \[
        \Sigma = \begin{pmatrix}
            0.00112 & 0  \\
            0 & 6.29\cdot10^{-5}
                 \end{pmatrix}
        \]
        in meters and rad.
        \item the running reward function is set to 0, so $r\colon (x,a)\mapsto0$
        \item discount factor is set to $\beta=1$
        \item terminal reward function $g\colon E_X\times E_Y\to\R$ is defined as
        \[
        g(o_t,s_t)\coloneqq\sum_{l=1}^{1000}\mathbbm{1}_{\{i_l^{t}=0\}},
        \]
        for all $s_t=((x_l^t,p_{x,l}^t),i_l^t)_{l=1}^{1000}\in E_Y$ and $o_t\in E_X$.
    \end{itemize}
\begin{remark}\label{rem:one_step_MDP}
The advantage of the one-step kick is that it can be realistically implemented at BESSY II. 

Since the position and shape of the newly injected electrons change only slightly from one injection to the next, we can use the information from the last injection to predict in which round and with which kicker strength the non-linear kicker should be used in the next injection. 

\end{remark}

\section{Summary}
In this chapter, we have seen 4 different ways of formulating the injection process.

At the beginning of chapter \ref{sec:model_description}, we introduced the first model. It showed the injection for a distribution of electrons. Although, only a finite number of electrons are injected each time since the sensor data is given as a histogram plot, we decided to model it as a distribution of electrons. We have shown the dynamics of the problem and described the final reward. 

The following 3 models are the models on which we applied reinforcement learning algorithms to find good policies. They were designed to have a low dimensional state space.

First, we described the single electron kick in subsection \ref{subsec:sing_elec}. In each round, we obtained the x and px information of an electron and could decide if we wanted to activate the non-linear kicker and if so with what strength.

Second, we described the 1000-step 1000-electron kick in subsection \ref{subsec:1000_elec}. In this model we obtained some information about the shape of the distribution of electrons and the mean values for x and px, in each round. With that we could decide once again in each round if we wanted to kick the electrons and with what strength.

Third, we described the 1-step 1000-electron injection in subsection \ref{subsec:1000_elec_1step}. Here, we only had the information about the shape and mean of the electrons in the very first round. Once again, we could use this information to decide in which round and with what strength we want to kick the electrons. We saw there, how to make the round decision a continuous action value, as this is needed for the DDPG and TD3 algorithms.

In section \ref{sec:programming_results}, we showed the performance of reinforcement learning algorithms on the different models.

\chapter{Calculations}\label{chap:calculations}
In this chapter, we investigate all the steps needed to train the reinforcement learning algorithms. We start with how we added stochasticity to the "drift-kick" simulation of BESSY II, explained in \ref{sec:r2rbehav} and how we made the simulation significantly faster, by approximating the environment. 

Then, we describe small modifications we made on the Markov Decision Models explained in section \ref{sec:model_description}. Finally, we show the best results in all the different models.

The programming reference can be found under the url \url{https://tubcloud.tu-berlin.de/s/RdKEm7sw9e5gyzY}.

\section{Implementation overview}
%

\subsection{Introduction of Stochasticity}\label{subsec:creating_stochas}
In section \ref{sec:r2rbehav}, we discovered how to create a simulation of BESSY II. We assume that the round-to-round behaviour function $\mathcal{T}$ is known. In the explanation of the Markov Decision Models in section \ref{sec:model_description}, we argued to add noise to the px and x values of the electrons in each round, as well as to the kicker strength. In this subsection, we show how the noises are distributed. 

All noise is zero mean normally distributed with the following standard deviations:

\begin{center}
    \begin{tabular}{lc}
        \toprule
        Noise type & std\\
        \midrule
        first round x and px:  & $6.5\cdot 10^{-5}$ in $\mathrm{m}$ and $\mathrm{rad}$  \\
        \midrule
        x: & $6.5\cdot 10^{-6}$ in $\mathrm{m}$\\
        \midrule
        px: & $5.2\cdot 10^{-6}$ in $\mathrm{rad}$\\
        \midrule
        noise kicker: & $0.0125$\\
    \bottomrule
    \end{tabular}
\end{center}

The amount of noise added to the electrons is the highest in the first round. This is because the magnets that guide the electrons close to the septum-sheet may create an additional unwanted influence on the electrons.

In every other round we also add noise to account for measurement errors and other possible unwanted influences from magnets. 
Additionally, we add noise to the non-linear kicker strength, since the desired strength can never be set exactly.

The last missing stochastic element is the distribution of the 1000 electrons, which we need in the 1000-electron kick and the one-step 1000-electron kick. We know that in the first round the distribution of the newly injected electrons has standard deviations of $\sigma_x=0.00112\mathrm{m}$ and $\sigma_{\mathrm{px}}=6.29\cdot10^{-5}\mathrm{rad}$. See the PhD-Thesis \textit{Transverse Resonance Island Buckets at BESSY II -A new Bunch Separation Scheme-} by Felix Armborst (\cite{felixarmborst}, section 4.2). Also, we assume that there is no correlation between both values.

\subsection{Creating the approximated simulation}
In this subsection, we investigate the creation of our own approximated simulation of non-linear kicker problem at BESSY II. Since the "drift-kick" simulation, see \cite{thscsi_paper}, took 180 seconds to simulate 1000 rounds for 1000 electrons, the idea is to create a similar but much faster environment. It is used to find good hyperparameters within a short time. Then, only the training for the best performing hyperparameters is re-executed using on original slower environment. In subsection \ref{subsec:1000s1000e_results}, we will show that the faster simulation is so close to the real simulation that good policies in the faster simulation also work well in the real simulation.

We have divided our analysis into two parts. In the first, we will consider how the round-to-round behaviour without the non-linear kicker influence can be approximated. In the second, we are interested in what to do when the non-linear kicker is activated. 
\subsubsection{Round-to-round behaviour}
We have approximated the round-to-round behaviour without the non-linear kicker by interpolation. Let $\mathbf{x}=(x_1,\dots,x_m)$ and $\mathbf{p^x}=(p^x_1,\dots,p^x_n)$  be a grid (sorted and equidistant) over all realistic x and px values. Now calculate for all $x_i,p^x_j$ the actual round-to-round behaviour $\tilde x_i,\tilde p^x_j=\mathcal{T}(0,x_i,p^x_j)$, with the "drift-kick" simulation, provided by \cite{thscsi_paper}. 

If we want to compute the round-to-round behaviour for new values $x,p^x$, we can do this in the following way:
\begin{enumerate}
    \item Assume $x\in[x_1,x_m]$ and $p^x\in[p^x_1,p^x_n]$. Otherwise, we assume that the electron is so far out, that it is lost.
    \item Let $i,j$ be the indices, such that $x\in[x_i,x_{i+1}]$ and $p^x\in[p^x_j,p^x_{j+1}]$. We call $\mathbf{b}=\{(x_{i},p^x_{j}),(x_{i},p^x_{j+1}),(x_{i+1},p^x_{j}),(x_{i+1},p^x_{j+1})\}$ the set of all border points.
    \item Calculate the euclidean distances of $(x,p^x)$ to all border points. Let $\mathbf{d}\in\R^4$ be the distances (sorted as in the definition of $\mathbf{b}$).
    \item Normalise the distance vector $\mathbf{d}$ by setting $\mathbf{d}_n\coloneqq\frac{\mathbf{d}}{||\mathbf{d}||_2}$.
    \item As closer points have smaller distances, but are probably more accurate, define 
    \[
    \mathbf{d}_r=\frac{(1-\mathbf{d}_n)}{||(1-\mathbf{d}_n)||_2}.
    \]
    \item Let $\widetilde{\mathbf{x}}\coloneqq\{\mathcal{T}(0,\tilde x,\tilde p^x)|_{x}\,|\, (\tilde x,\tilde p^x)\in\mathbf{b}\}$ and $\widetilde{\mathbf{p^x}}\coloneqq\{\mathcal{T}(0,\tilde x,\tilde p^x)|_{p^x}\,|\, (\tilde x,\tilde p^x)\in\mathbf{b}\}$ be the sets of the new x- and px-values of the border points. Assume that we can treat these sets as vectors (sorted as in the definition of $\mathbf{b}$).
    \item The transitioned x- and px-values starting at $x,p^x$ can be calculated using:
    \[
    x'\coloneqq \mathbf{d}_r^T\widetilde{\mathbf{x}}\quad\text{and}\quad {p^x}'\coloneqq \mathbf{d}_r^T\widetilde{\mathbf{p^x}}.
    \]
\end{enumerate}
With this method, we have found a fast way to approximate the transition function $\mathcal{T}(0,\cdot,\cdot)$. Note that the smaller the grid size, the better the approximation, and that the calculation for 1000 electrons can be nicely parallelised in NumPy, see \cite{numpy}.

\subsubsection{Non-linear kicker behaviour}
Now, we show what was done to approximate the non-linear kicker behaviour.

The first thing we did was to fit an easy to evaluate function to the electron survival rate seen.
\begin{figure}[H]
    \centering
    \includegraphics[width=9cm]{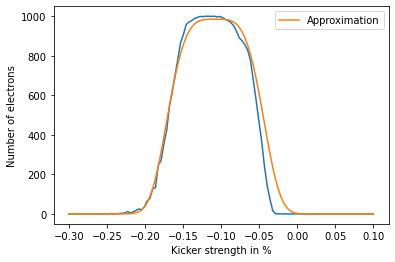}
    \caption{Extended figure \ref{fig:kicker_strength_behav}. The orange line maps a kicker strength $\alpha$ to \\ $985\cdot\mathrm{exp}(-(14.5\cdot(\alpha-\alpha^*))^4)$, where $\alpha^*=-0.1075$ is the optimal kicker strength. The value represents: How many electrons out of of 1000  survive the injection, given the kicker strength. Note that the optimal kicker strength $\alpha^*$ depends on the x,px position at the kicker, see section \ref{subsec:optimal_kicker_usage} for details.}
    \label{fig:kicker_strength_behav_w_approx}
\end{figure}
In figure \ref{fig:kicker_strength_behav}, one can see which function was fitted to the electron survival rate. 

Note that we used the scaling factor $985$ in the approximation shown in figure \ref{fig:kicker_strength_behav_w_approx}, since the real injection will never reach 100\% efficiency. 

We have now combined the models shown in the figures \ref{fig:optimal_area}, \ref{fig:optimal_strength} and \ref{fig:kicker_strength_behav_w_approx} to determine the electron survival probability at current x- and px-values $x,p^x$ and kicker strength $\alpha$. Figure \ref{fig:optimal_area} shows whether the electron could be successfully kicked given the current x- and px-values at the septum sheet. Figure \ref{fig:optimal_strength} shows the corresponding optimal kicker strengths and figure \ref{fig:kicker_strength_behav_w_approx} shows the injection probability for different kicker strengths. We have approximated the non-linear kicker usage for a single electron with x- and px-values $x,p^x$ and kicker strength $\alpha$ in the following way:
\begin{enumerate}
    \item For the values $x,p^x$ we checked if one of the border points in the grid shown in figure \ref{fig:optimal_area} survived the kick, i.e. had value 1000. If not, we set the injection probability to 0.
    \item In a similar way to approximating the round-to-round behaviour explained above, we approximated the optimal kicker strength by interpolation (the optimal kicker strength for all not survived border points have been set to -1, such that it remained well-defined. One could have ignored the points and used only the survived border points, which we did not do. The way we did it was simpler and faster).
    \item Now, using the approximation $985\cdot\mathrm{exp}(-(14.5\cdot(\alpha-\alpha^*))^4)$ from figure \ref{fig:kicker_strength_behav_w_approx}, where $\alpha^*$ is the approximated optimal kicker strength and $\alpha$ is the chosen kicker strength, we can now approximate the survival probability by dividing the value  $985\cdot\mathrm{exp}(-(14.5\cdot(\alpha-\alpha^*))^4)$ by 1000.
\end{enumerate}
So, we have found a way to approximate the electron survival probability in a fast way. Note that the approximation $985\cdot\mathrm{exp}(-(14.5\cdot(\alpha-\alpha^*))^4)$ was only fitted to a single pair of $x,p^x$ values, namely $x=16.4\mathrm{mm}$ and $p^x=0.0\mathrm{mrad}$. As an approximation, we assume that this holds for all $x,p^x$ values. 

Now, we can calculate the probability of the electrons surviving the injection at the moment the electrons are kicked. We no longer need to transition them to round 1000.

This also gives us a way of simultaneously approximating the expected number of electrons that survive the injection for 1000 electrons, by summing up each individual survival probability.

Note that with this approximation and not transitioning the electrons to round 1000, we no longer follow the Markov Decision Models defined in section \ref{sec:model_description}. We move on to show that this shortcut does not change the value functions.

From the reward iteration theorem \ref{thm:reward_iteration}, we know that
\[
V_{n,\pi}(x) = \beta^nr(x,f_n(x))+\int_E V_{n+1,\pi}(x')p(\mathrm{d}x'|x,f_n(x)),\quad\text{for all }  x\in E.
\]
Let $n$ be the round in which the non-linear kicker was activated. From the definition of the Markov Decision Models in section \ref{sec:model_description}, we know that from now on no action will affect the transitions and all rewards $r$ are set to 0. Since in all further rounds the states $x'$ are independent of the chosen action, we get 
\begin{align*}
V_{n,\pi}(x) &= \int_E V_{n+1,\pi}(x')p(\mathrm{d}x'|x,f_n(x))\\
&= \E[V_{n+1,\pi}(x')|x, f_n(x)]\\
&= \E[\E[V_{n+2,\pi}(x'')|x']|x,f_n(x)]\\
&= \E[V_{n+2,\pi}(x'')|x,f_n(x)]\\
&=...\\
&= \E[{g(x^{(1000-n)}}')|x,f_n(x)]
\end{align*}
for all  $x\in E$, which is equal to the expected amount of electrons that survive. 

This allows us to stop at the moment where the kicker was activated and return the expected number of surviving electrons as reward.

We have thus created a faster simulation of the non-linear kicker problem, which we can use to find good hyperparameters and even use the learned policies for the actual simulation.

\subsection{Adapting actions}\label{subsec:adapting_actions}
As in the 1000-step single-electron model and the 1000-step 1000-electron model, the non-linear kicker is activated as soon as the chosen action is not equal to 0. We wanted to increase the interval where electrons are not kicked, such that waiting a round and possibly kicking the electrons in the next round would be easier for the policies to learn.

To do this, we introduced the following activation function around the chosen actions
\[
\phi(x)=\sign(x)\cdot x^4
\]
and decided that in the interval $[-0.16,
0.16]$ the chosen actions do not lead to the activation of the non-linear kicker. The interval limits were chosen so that a kicker activation of $0.16^4$ leads to a the change in the electrons direction of flight similar in magnitude to the added noise.

\subsection{Reward design}\label{subsec:reward_design}
A small change in the running reward function $r$ was made to the first two models, where the information of the electrons were given in each round. In case the non-linear kicker was activated, but no electrons were in the accepted kicker area, shown in figure \ref{fig:optimal_area}, we decided to give a small reward that would increase the closer the chosen action was to 0, such that the algorithm learns not to kick the electrons in that round.

To do this, we added two hyperparameters $r_1,r_2$ to the models, which we call reward design parameters. We set the reward function $r$ to
\[
r_1-r_2\cdot|\alpha|,
\]
where $\alpha$ is the chosen kicker strength, but only in the case where the kicker was activated and all electrons were outside the kick-able area. For the hyperparameter search we allowed the values $r_1\in[-50,50]$ and $r_2\in[0,50]$.

\section{Programming results}\label{sec:programming_results}
In this section, we will present the best performing policies that we have been able to train. All policies were found using the TD3 algorithm described in subsection \ref{sec:TD3} and have been optimised using the tree-structured parzen estimator hyperparameter search algorithm, described in \ref{sec:tree_struc}. We also made the small adjustments described in subsections \ref{subsec:adapting_actions} and \ref{subsec:reward_design}, where we added an activation function around the action and changed the reward design in a specific case.

\subsection{1000-Step Single-Electron Model}
We start with the model explained in subsection \ref{subsec:sing_elec}. In this model we get the information of the x and px value of the electron in each round and can decide if we want to use the non-linear kicker and if so, with which kicker strength. 

The following figures show the performance of the best trained policy on the injection area. The first plot shows the reward, the second the round in which the different electrons were been kicked. 

\begin{figure}[H]
    \centering
    \includegraphics[width=12cm]{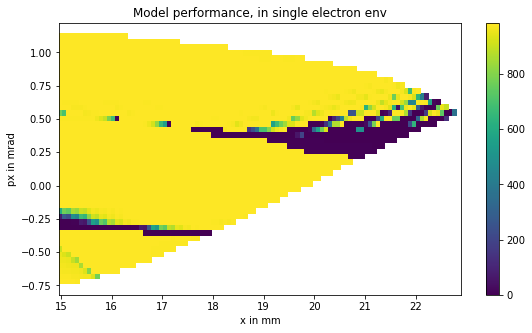}
    \caption{Model performance of the best policy on the injection area. The performance was checked on the deterministic version of the environment to make the results comparable. The x,px values are the initial values at the septum sheet. The performance shown is 1000 times the final reward. The average score over the whole injection area is 875 out of 1000 possible.}
    \label{fig:sing_elec_reward}
\end{figure}
\begin{figure}[H]
    \centering
    \includegraphics[width=12cm]{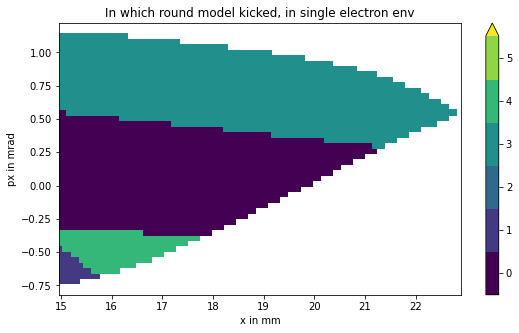}
    \caption{Corresponding rounds for the model shown in figure \ref{fig:sing_elec_reward}.}
    \label{fig:sing_elec_rounds}
\end{figure}
In figure \ref{fig:sing_elec_reward}, we can see that the performance of the model in the injection area is very good. In two smaller areas the electrons cannot be successfully injected. In figure \ref{fig:sing_elec_rounds}, we can see that rounds 0,1,3 and 4 have been used.

For the training we used the following hyperparameters: 
\begin{itemize}
    \item neural network (NN) shape of the policy network: $(256,128)$
    \item NN-shape of the value function network: $(32,32)$
    \item learning rate: $0.002722$
    \item batch size $64$
    \item number of runs before starting learning: $250$
    \item $\tau=0.005$
    \item action-noise: Ornstein-Uhlenbeck noise with standard variation: $0.377$
    \item reward design parameters $r_1=11,r_2=7$ 
    \item discount factor $\gamma = 1.3552$. 
\end{itemize}
Note that we did not use a discount factor ($\gamma$=1) to measure the performance. For the training process we have seen that it makes a difference and we have had the best results with discount factors greater than 1. This is very unusual, as the discount factor is usually chosen between 0 and 1. In our case, it helped the model to explore later rounds to kick the electrons.

\subsection{1000-Step 1000-Electron Model}\label{subsec:1000s1000e_results}
We continue  with the model explained in subsection \ref{subsec:1000_elec}. Here we have the information of the x and px means, and the shape of the 1000 electrons in each round and can decide if we want to use the non-linear kicker and if so, with which kicker strength. 

In this subsection, we look at some additional analysis that can also be applied to the next model. First, we show the performance of the best trained agent. Then, we compare it with the theoretical best performing policy and finally, we show how the agent performs in the real simulation.  

The following figures show the performance of the best trained policy on the injection area. The first plot shows the reward, the second the round in which the different electrons were kicked. 

\begin{figure}[H]
    \centering
    \includegraphics[width=12cm]{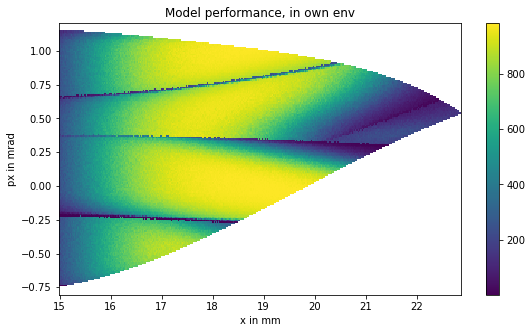}
    \caption{Model performance of the best policy on the injection area. The performance was checked on the deterministic version of the environment to make the results comparable. The x,px values show the mean of the distribution of electrons at the septum sheet. The performance shown corresponds to how many of the 1000 electrons survived the injection. The mean value over the whole injection area is 644.}
    \label{fig:1000_elec_reward}
\end{figure}
\begin{figure}[H]
    \centering
    \includegraphics[width=12cm]{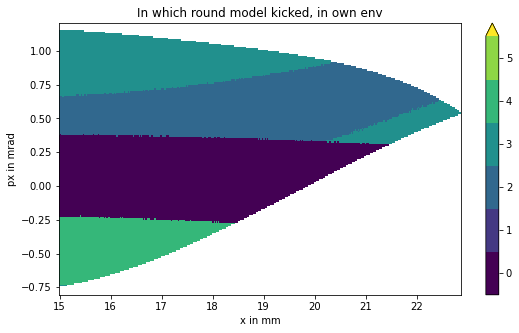}
    \caption{Corresponding rounds for the model shown in figure \ref{fig:1000_elec_reward}.}
    \label{fig:1000_elec_rounds}
\end{figure}
In figure \ref{fig:1000_elec_reward}, we can see that the performance of the model in the injection area is good. The area close to the septum sheet is not injected as successfully. The reason for this is that when the x-mean is close to 15mm, half of the electrons are already lost in the septum-sheet, before we are able to use the non-linear kicker. Similar to the performance of the model in the single-electron case, the electrons with x-mean greater than 20.5mm are not injected successfully. In figure \ref{fig:1000_elec_rounds}, we can see that the rounds 0,2,3 and 4 have been used.

The following hyperparameters were used for training: 
\begin{itemize}
    \item  NN-shape of the policy network: $(256,128)$
    \item NN-shape of the value function network: $(512)$
    \item learning rate: $0.000464$
    \item batch size $64$
    \item number of runs before starting learning: $250$
    \item $\tau=0.005$
    \item action-noise: Ornstein-Uhlenbeck noise with standard variation: $0.841$
    \item reward design parameters $r_1=11.129,r_2=6.67$
    \item discount factor $\gamma = 1.463$. 
\end{itemize}
Again, we did not use a discount factor ($\gamma$=1) in the performance measure.
As a continuous transformation function $\mathcal{H}_\theta\colon\R^{51,51}\to\R^4$, which takes the histogram plots as an input, we trained a convolutional neural network with the continuous ReLU activation function.  

The following code shows the used convolutional neural network:
\begin{figure}[H]
    \centering
\begin{lstlisting}
from torch import nn

cnn_used = nn.Sequential(
    nn.Conv2d(1, 8, kernel_size=6, stride=2, padding=0),
    nn.ReLU(),
    nn.MaxPool2d(2),
    nn.Conv2d(8, 4, kernel_size=3, stride=1, padding=0),
    nn.ReLU(),
    nn.MaxPool2d(2),
    nn.Conv2d(4, 4, kernel_size=2, stride=1, padding=0),
    nn.ReLU(),
    nn.Conv2d(4, 2, kernel_size=2, stride=1, padding=0),
    nn.ReLU(),
    nn.Flatten(),
    nn.Linear(8,4))    
\end{lstlisting}
\caption{Used CNN structure.}
\label{code:cnn_structure}
\end{figure}
In each \verb|nn.Conv2d| the first input is the number of input channels, the second input is the number of output channels. The input of all \verb|nn.MaxPool2d| is the window length. Note that we used multiple times in our code an even window length. As in our definition of CNNs (definition \ref{defn_cnn}) only odd windows lengths were allowed, we refer to the paper \textit{A guide to convolution arithmetic for deep learning}  by V. Dumoulin and F. Visin \cite{cnn_even} to see how it is used.

To evaluate the performance, we calculated the theoretical optimal performance on the injection area for the 1000 electron injections. For each combination of x and px means, we calculated the optimal kicker strength in each round and how many electrons would survive the injection.

The following results were obtained:
\begin{figure}[H]
    \centering
    \includegraphics[width=12cm]{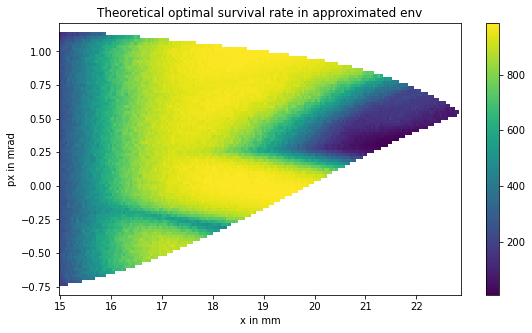}
    \caption{Theoretical best performance on the injection area. The performance was checked on the deterministic version of the environment to make the results comparable. The x,px values show the mean of the distribution of electrons at the septum sheet. The performance shown corresponds to how many of the 1000 electrons survived the injection. The mean value over the whole injection area is 676.57.}
    \label{fig:optimal_1000_reward}
\end{figure}
\begin{figure}[H]
    \centering
    \includegraphics[width=12cm]{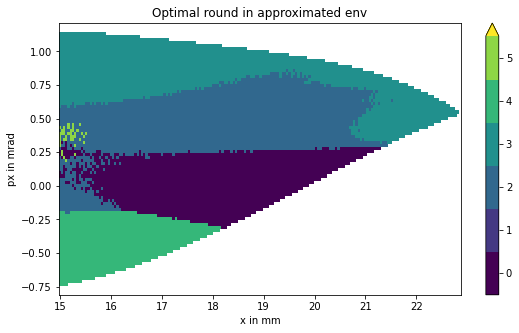}
    \caption{Corresponding rounds for the model shown in figure \ref{fig:optimal_1000_reward}.}
    \label{fig:optimal_1000_round}
\end{figure}
We observe that the performance of our trained policy is close to the performance of the optimal policy. The optimal policy has lower scores for the electrons that are far outside than the trained policy. The rounds chosen to activate the non-linear kicker are slightly different.

Now, we show the performance of the model on the real simulation we have approximated, to ensure that the approximated simulation is close to the "drift-kick" simulation.

\begin{figure}[H]
    \centering
    \includegraphics[width=12cm]{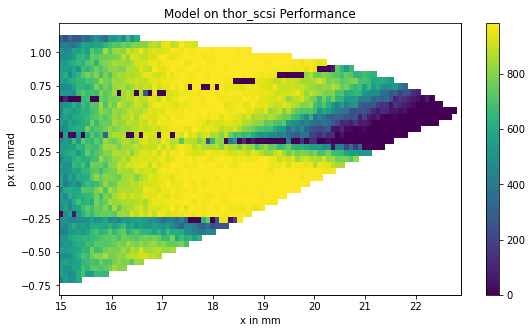}
    \caption{Performance of best trained model on the injection area in the "drift kick" simulation. The performance was checked on the deterministic version of the "drift kick" simulation to make the results comparable. The x,px values show the mean of the distribution of electrons at the septum sheet. The performance shown corresponds to how many of the 1000 electrons survived the injection. The mean value over the whole injection area is 727.7.}
    \label{fig:model_on_thrscsi}
\end{figure}
\begin{figure}[H]
    \centering
    \includegraphics[width=12cm]{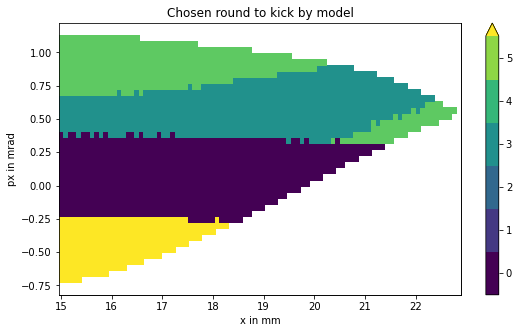}
    \caption{Corresponding rounds for the model shown in figure \ref{fig:model_on_thrscsi}.}
    \label{fig:model_on_thrscsi_rounds}
\end{figure}

We see that the performances are similar. The total score for the injection area is 727.7. This is even higher than the performance in the approximated simulation. However, in the regions where the injection round changes, the performance is worse. These observations suggest that the approximation of the non-linear kicker behaviour is not exact. In some regions where the optimal kicker strength changes rapidly, see figure \ref{fig:optimal_strength}, the plateau drop in electron survival probability should have been faster, see figure \ref{fig:kicker_strength_behav_w_approx}. In other regions, where the optimal kicker strength changes slowly, it could have been more relaxed. Overall, our approximation remains good, because the overall score is lower and the approximation of the round-to-round behaviour is very close.

\subsection{1-Step 1000-Electron Model}
We end with the model explained in subsection \ref{subsec:sing_elec}. Here we only get the information of the x and px means and the shape of the 1000 electrons in the very first round and after that we have to decide directly when to use the non-linear kicker and with which kicker strength. 

The following figures show the performance of the best trained policy on the injection area. The first plot shows the reward, the second the round in which round the different electrons were kicked. 

\begin{figure}[H]
    \centering
    \includegraphics[width=12cm]{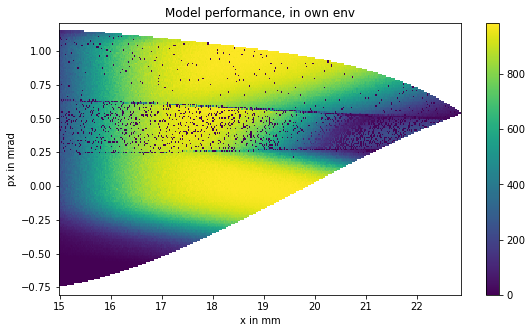}
    \caption{Model performance of the best policy on the injection area. The performance was checked on the deterministic version of the environment to make the results comparable. The x,px values show the mean of the distribution of electrons at the septum sheet. The performance shown corresponds to how many of the 1000 electrons survived the injection. The mean value over the whole injection area is 600.}
    \label{fig:single_step_best}
\end{figure}
\begin{figure}[H]
    \centering
    \includegraphics[width=12cm]{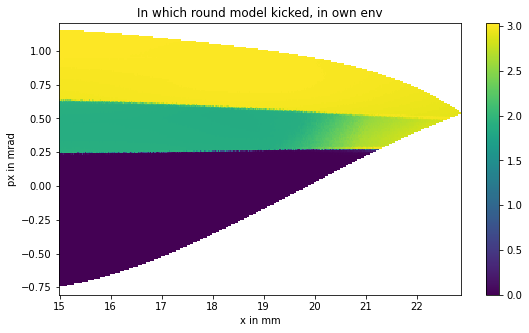}
    \caption{Corresponding rounds for the model shown in figure \ref{fig:single_step_best}.}
    \label{fig:single_step_best_rounds}
\end{figure}
The figure \ref{fig:single_step_best} shows the performance of the model in the injection area. Note that the image noise seen in the figure is the result of the stochastic choice of round, described in \ref{subsec:1000_elec_1step}. There is no noise in the lower half of the injection area in the performance figure \ref{fig:single_step_best}, because the lower end of an action space is easier to hit by clipping operators than values in between.

Compared to the other results, its performance is the worst. In particular, the performances for electrons with high x-value and low px values is not great.

In figure \ref{fig:single_step_best_rounds}, we can see that the used rounds almost equals the used rounds of the best 1000-step 1000-electron case, see figure \ref{fig:1000_elec_rounds}. For the electrons with low px values, the best 1000-step 1000-electron model used the round 4 to kick the electrons, which the model here was model able to learn. Possibly, with other hyperparameters and a longer training duration, this would have been possible. 

The following hyperparameters were used for training:

\begin{itemize}
    \item NN-shape of the policy network: $(256, 128)$
    \item NN-shape of the value function network: $(256, 128)$
    \item learning rate: $0.001$
    \item batch size $64$
    \item number of runs before learning starts: $250$
    \item $\tau=0.005$
    \item action-noise: Ornstein-Uhlenbeck noise with standard variation: $1.4777$ 
    \item discount factor $\gamma = 1$.
\end{itemize}
    Again, no discount factor ($\gamma$=1) was used in the performance measure. We used the same convolutional neural network  for the continuous transformation function $\mathcal{H}_\theta\colon\R^{51,51}\to\R^4$, as shown in figure \ref{code:cnn_structure}.

\section{Summary}
At the beginning of this chapter, we saw how we added noise to the "drift-kick" simulation to make the injection process more realistic. We then created an approximate simulation of the "drift-kick" simulation. For this approximate simulation, we approximated the round-to-round behaviour using interpolation and once the non-linear kicker was activated, we approximated the number of electrons that would survive this injection. This approximated simulation is much faster, remains precise and the noise can be added in the same way. We also found a shortcut by stopping the Markov decision process as soon as the non-linear was activated, which reduced the time needed for the simulation.

We needed this faster simulation in order to find good hyperparameters and policies in a reasonable time. For all 3 different Markov decision models, with varying degrees of complexity, we found very good performing policies. We have shown that these policies trained on the approximate simulation even work well for the "drift kick" simulation, and they show how to kick the electrons dependent on their position. Particularly, we were able to identify regions where the electrons are most likely to survive the injection.

\chapter{Conclusion}\label{chap:conclusion}
This thesis started off with an introduction to reinforcement learning. We saw the definition of a Markov decision process, defined the value function and proved the Bellman equation, which allowed us to compute optimal policies. As this was not efficient for continuous action spaces, we saw the DDPG and TD3 algorithms and how they solved the problem using function approximators.

We also gave an introduction to hyperparameter search. As the performance of reinforcement learning algorithms depends heavily on how well the hyperparameters fit to the problem, we presented in section \ref{section:hyper} different ways to search the hyperparameter space for good parameters. 

We also defined partially observable Markov decision models and explained the difference to Markov decision models.

We then took a closer look at an application for reinforcement learning. The aim was to optimise the non-linear kicker injection at BESSY II. We looked at the physical background of BESSY II and the non-linear kicker and gave some intuition on how to use the non-linear kicker. We explained how we added stochasticity to the environment and created our own approximate simulation of the BESSY II injection.

The non-linear kicker injection was described in 3 different ways with varying degrees of complexity. For all 3 different models we found very good policies. 

In particular, the 1-step 1000-electron model, where the policies decided the kicker activation round and strength uses only the information from the electrons in round 0, could be used for the real injection at BESSY II, as this environment reflects the real machine capabilities of BESSY II. Since the shape and position of the newly injected electrons typically change only slightly from one injection to the next, the information from the last injection could be used to find the optimal activation round and strength.

Note that before actually applying a trained policy on the real machine BESSY II, one would need to train the policy to be more robust in case of environment changes. 

In addition, another reinforcement learning agent could be trained to use the magnetic steerers, leading the electrons to the septum-sheet, to always steer the electrons into a region of high injection efficiency.  

The results show that it is possible to train reinforcement learning agents to optimise the non-linear kicker injection.
The best trained 1-step 1000-electron policy is now ready to be tested on the BESSY II machine.

\nocite{*}
\printbibliography

@misc{mdp_in_finance,
publisher = {Springer},
isbn = {9783642183232},
year = {2011},
title = {Markov decision processes with applications to finance},
author = {Nicole Bäuerle and Ulrich Rieder},
}

@misc{ddpg_paper,
      title={Continuous control with deep reinforcement learning}, 
      author={Timothy P. Lillicrap and Jonathan J. Hunt and Alexander Pritzel and Nicolas Heess and Tom Erez and Yuval Tassa and David Silver and Daan Wierstra},
      year={2019},
      eprint={1509.02971},
      archivePrefix={arXiv}
}

@InProceedings{dpg_paper,
  title = 	 {Deterministic Policy Gradient Algorithms},
  author = 	 {Silver, David and Lever, Guy and Heess, Nicolas and Degris, Thomas and Wierstra, Daan and Riedmiller, Martin},
  booktitle = 	 {Proceedings of the 31st International Conference on Machine Learning},
  pages = 	 {387--395},
  year = 	 {2014},
  editor = 	 {Xing, Eric P. and Jebara, Tony},
  volume = 	 {32},
  number =       {1},
  series = 	 {Proceedings of Machine Learning Research},
  address = 	 {Bejing, China},
  publisher =    {PMLR},
}

@misc{td3_paper,
      title={Addressing Function Approximation Error in Actor-Critic Methods}, 
      author={Scott Fujimoto and Herke van Hoof and David Meger},
      year={2018},
      eprint={1802.09477},
      archivePrefix={arXiv},
      primaryClass={cs.AI}
}

@misc{overestimation_paper,
  title={Issues in Using Function Approximation for Reinforcement Learning},
  author={Sebastian Thrun and Anton Schwartz},
  year={1993}
}

@misc{hyper_para_paper,
 author = {Bergstra, James and Bardenet, R\'{e}mi and Bengio, Yoshua and K\'{e}gl, Bal\'{a}zs},
 title = {Algorithms for Hyper-Parameter Optimization},
 year = {2011}
}

@article{random_vs_grid_paper,
  author  = {James Bergstra and Yoshua Bengio},
  title   = {Random Search for Hyper-Parameter Optimization},
  journal = {Journal of Machine Learning Research},
  year    = {2012},
  volume  = {13},
  number  = {10},
  pages   = {281--305},
  url     = {http://jmlr.org/papers/v13/bergstra12a.html}
}

@book{gaussian_processes_book,
    author = {Rasmussen, Carl Edward and Williams, Christopher K. I.},
    title = "{Gaussian Processes for Machine Learning}",
    publisher = {The MIT Press},
    year = {2005},
    isbn = {9780262256834},
    doi = {10.7551/mitpress/3206.001.0001},
}

@misc{BA_schuett,
  title={Bachelorthesis},
  author={Schütt, Alexander},
  year={2021},
  url={https://tubcloud.tu-berlin.de/s/k4s33NQGxA6ad7X}
}

@book{physics_book,
    author = {H. A. Haus and J. R. Melcher},
    title = "{Electromagnetic fields and energy}",
    publisher = {Prentice-Hall International Editions},
    year = {1988},
}

@article{nlk_paper,
    author = "Rast, H. and Atkinson, T. and Dirsat, M. and Dressler, O. and Kuske, P.",
    editor = "Petit-Jean-Genaz, Christine",
    title = "{Development of a Non-Linear Kicker System to Facilitate a New Injection Scheme for the BESSY II Storage Ring}",
    reportNumber = "IPAC-2011-THPO024",
    journal = "Conf. Proc. C",
    volume = "110904",
    pages = "3396--3398",
    year = "2011"
}

@article{hamiltonian_paper,
    author = "Bengtsson, J. and Rogers, W. and Nicholls, T.",
    title = "{A CAD Tool for Linear Optics Design: A Controls Engineer's Geometric Approach to Hill's Equation}",
    eprint = "2109.15066",
    archivePrefix = "arXiv",
    primaryClass = "physics.acc-ph",
    year = "2021"
}

@article{synchrotron_schwinger,
  title = {On the Classical Radiation of Accelerated Electrons},
  author = {Schwinger, Julian},
  year = {1949},
  url = {https://link.aps.org/doi/10.1103/PhysRev.75.1912},
  doi = {10.1103/PhysRev.75.1912},
}

@online{Liu2016InjectionDF,
title={Injection Dynamics for Sirius Using a Nonlinear Kicker},
author={Lin Liu and Ximenes R. Resende and Ant{\^o}nio Ricardo Droher Rodrigues and Fernando de S{\'a}},
year={2016},
url ={https://accelconf.web.cern.ch/ipac2016/papers/thpmr011.pdf},
organization = {THPMR011 Proceedings of IPAC2016, Busan, Korea pp 3406-3408},
urldate = {2024-01-30}
}

@book{rl_book,
author = {Sutton, Richard S. and Barto, Andrew G.},
title = {Reinforcement Learning: An Introduction},
year = {2018},
isbn = {0262039249},
publisher = {A Bradford Book},
address = {Cambridge, MA, USA}
}

@ARTICLE{td3_for_pomdp_1,
  author={Kabbani, Taylan and Duman, Ekrem},
  journal={IEEE Access}, 
  title={Deep Reinforcement Learning Approach for Trading Automation in the Stock Market}, 
  year={2022},
  volume={10},
  number={},
  pages={93564-93574},
  doi={10.1109/ACCESS.2022.3203697}}

@book{nn_book,
author = {Shalev-Shwartz, Shai and Ben-David, Shai},
title = {Understanding Machine Learning: From Theory to Algorithms},
year = {2014},
isbn = {1107057132},
publisher = {Cambridge University Press},
address = {USA}
}

@misc{wolf2018mathematical,
  title={Mathematical foundations of supervised learning},
  author={Wolf, Michael M},
  year={2018},
  publisher={July}
}

@inproceedings{thscsi_paper,
    author = {P. Schnizer et al.},
    title = {Progress on Thor SCSI development},
    booktitle = {Proc. IPAC'23},
    pages = {3413-3416},
    paper = {WEPL127},
    venue = {Venice, Italy},
    series = {IPAC'23 - 14th International Particle Accelerator Conference},
    number = {14},
    publisher = {JACoW Publishing, Geneva, Switzerland},
    month = {05},
    year = {2023},
    issn = {2673-5490},
    isbn = {978-3-95450-231-8},
    doi = {10.18429/JACoW-IPAC2023-WEPL127},
    url = {https://indico.jacow.org/event/41/contributions/2085},
    language = {English}
}

@inproceedings{cnn_book,
author = {Venkatesan, R. and Li, B.},
year = {2017},
title = {Convolutional Neural Networks in Visual Computing: A Concise Guide (1st ed.)},
doi = {10.4324/9781315154282}
}

@online{4kickerbump,
author = {{Simon White}},
title = {Experience at ESRF},
url ={https://indico.cern.ch/event/635514/contributions/2660458/attachments/1518207/2370454/Berlin_TWIIS.pdf},
organization = {Presentation given at  Topical Workshop on Injection and Injection Systems},
date = {2017-08-28},
urldate = {2024-01-30}
}

@misc{mlfin,
      title={Machine Learning with Financial Applications}, 
      author={Knochenhauer, Christoph and Bayer, Christian},
      url={https://isis.tu-berlin.de/course/view.php?id=25665},
      urldate = {2024-01-29},
}

@misc{cnn_even,
      title={A guide to convolution arithmetic for deep learning}, 
      author={Vincent Dumoulin and Francesco Visin},
      year={2018},
      eprint={1603.07285},
      archivePrefix={arXiv},
}

@Article{         numpy,
 title         = {Array programming with {NumPy}},
 author        = {Charles R. Harris and K. Jarrod Millman and St{\'{e}}fan J.
                 van der Walt and Ralf Gommers and Pauli Virtanen and David
                 Cournapeau and Eric Wieser and Julian Taylor and Sebastian
                 Berg and Nathaniel J. Smith and Robert Kern and Matti Picus
                 and Stephan Hoyer and Marten H. van Kerkwijk and Matthew
                 Brett and Allan Haldane and Jaime Fern{\'{a}}ndez del
                 R{\'{i}}o and Mark Wiebe and Pearu Peterson and Pierre
                 G{\'{e}}rard-Marchant and Kevin Sheppard and Tyler Reddy and
                 Warren Weckesser and Hameer Abbasi and Christoph Gohlke and
                 Travis E. Oliphant},
 year          = {2020},
 month         = sep,
 journal       = {Nature},
 volume        = {585},
 number        = {7825},
 pages         = {357--362},
 doi           = {10.1038/s41586-020-2649-2},
 publisher     = {Springer Science and Business Media {LLC}},
 url           = {https://doi.org/10.1038/s41586-020-2649-2}
}

@phdthesis{felixarmborst,
author = {Armborst, Felix},
title = {Transverse Resonance Island Buckets at BESSY II},
school = {Humboldt-Universität zu Berlin, Mathematisch-Naturwissenschaftliche Fakultät},
year = {2022},
doi = {http://dx.doi.org/10.18452/23851}
}

\end{document}